\documentclass[12pt]{article}
\usepackage{subcaption}
\usepackage{changepage}
\usepackage[dvipsnames]{xcolor}
\usepackage{xr}
\usepackage[english]{babel}
\usepackage{amsfonts}
\usepackage{amssymb}
\usepackage{bm}
\usepackage{amsmath}
\usepackage{tikz}
\usetikzlibrary{calc,intersections,shapes,decorations,arrows,arrows.meta,fit,positioning, decorations.pathreplacing}
\usepackage{graphicx}
\usepackage{caption}
\usepackage{adjustbox}
\usepackage{threeparttable}
\usepackage{scrextend}
\usepackage{relsize}
\usepackage{booktabs}	
\usepackage[normalem]{ulem}
\usepackage{enumerate}
\usepackage{bbm}
\usepackage{mathtools}
\usepackage{booktabs}	
\usepackage{setspace}
\usepackage{graphicx}
\usepackage{fullpage}

\usepackage{soul}
\usepackage{xspace}
\usepackage{csquotes}

\usepackage{placeins}
\usepackage[footfont=normalsize,font=normalsize]{floatrow}
\usepackage{lscape}
\usepackage{longtable}
\usepackage{mathabx}
 \usepackage{accents}
\usepackage{float}
\usepackage{fullwidth}
\usepackage{verbatim}
\usepackage{autobreak}
\usepackage[bottom]{footmisc} 
\usepackage[titletoc,title, toc,page]{appendix}
\usepackage{tikz}
\usepackage{booktabs}
\usepackage{graphicx}
\usetikzlibrary{calc,intersections}
\usepackage{threeparttable}
\usepackage{pdflscape}
\usepackage{adjustbox}

\usepackage[left=2cm,right=2cm,top=2cm,bottom=2cm]{geometry}
    
	\usepackage[longnamesfirst]{natbib}
	\usepackage{bibunits}
	\defaultbibliography{Bibliography.bib}  
	\defaultbibliographystyle{aer}

\definecolor{cite_blue}{rgb}{0.0, 0.1, 0.3}
\definecolor{BrickRed}{rgb}{0.8, 0.25, 0.33}
	\usepackage{hyperref}
	\hypersetup{
    colorlinks=true,
    linkcolor=cite_blue,
    filecolor=cite_blue,      
    urlcolor=cite_blue,
    citecolor = cite_blue
}
	\usepackage{cleveref}

    \usepackage{rotating}
	\usepackage[framemethod=tikz]{mdframed}

    \usepackage{tcolorbox}
    \newtcbox{\feedback}{nobeforeafter,colframe=black,colback=white,boxrule=0.5pt,arc=2pt,
      boxsep=0pt,left=2pt,right=2pt,top=2pt,bottom=2pt,tcbox raise base}
      
    \usepackage{amsthm}
    
    \newtheorem{theorem}{Theorem}[section]

    \newtheorem{lemma}{Lemma}[section]
    \newtheorem{corollary}{Corollary}[section]
    \theoremstyle{definition}
    \newtheorem{remark}{Remark}[section]

\usepackage{ragged2e}
\newlength\ubwidth

\newtheorem{innercustomass}{Assumption}
\newenvironment{namedassumption}[1] {
    \renewcommand\theinnercustomass{#1}
    \innercustomass 
} {\endinnercustomass}

\allowdisplaybreaks

\AtBeginDocument{
\setlength{\abovedisplayskip}{6pt}
\setlength{\belowdisplayskip}{6pt}
\setlength{\abovedisplayshortskip}{6pt}
\setlength{\belowdisplayshortskip}{6pt}
}

\deffootnote{1em}{1.6em}{\thefootnotemark \enskip}

\numberwithin{equation}{section}

\renewcommand{\P}{\mathop{}\!\mathbb{P}}

\newcommand{\E}{\mathop{}\!\mathbb{E}}

\newcommand{\Var}{\mathop{}\!\textnormal{Var}}

\newcommand{\bracks}[1]{\left[#1\right]}

\newcommand{\R}{\mathbb{R}}

\newcommand{\gc}{g_{\text{c}}}

\makeatletter
\newcommand*{\indep}{%
  \mathbin{%
    \mathpalette{\@indep}{}%
  }%
}
\newcommand*{\nindep}{%
  \mathbin{
    \mathpalette{\@indep}{/}%
  }%
}
\newcommand*{\@indep}[2]{%
  \sbox0{$#1\perp\m@th$}
  \sbox2{$#1=$}
  \sbox4{$#1\vcenter{}$}
  \rlap{\copy0}
  \dimen@=\dimexpr\ht2-\ht4-.2pt\relax
  \kern\dimen@
  \ifx\\#2\\%
  \else
    \hbox to \wd2{\hss$#1#2\m@th$\hss}%
    \kern-\wd2 %
  \fi
  \kern\dimen@
  \copy0 
}
\makeatother

\usepackage{rotating}  
\makeatletter
\newcommand*{\addFileDependency}[1]{
  \typeout{(#1)}
  \@addtofilelist{#1}
  \IfFileExists{#1}{}{\typeout{No file #1.}}
}
\makeatother

\begin{document}
\begin{bibunit}

\title{Better Understanding Triple Differences Estimators\thanks{We thank participants from the 2023 SEA Conference, 2024 Midwest Econometrics Group Conference, 2025 Georgia Econometrics Workshop, and participants of the ``Difference-in-Differences'' Workshops from Causal Solutions in 2022 and 2023. This paper builds on and supersedes the material used for the 2022 Causal Solutions YouTube Lecture available \href{https://www.youtube.com/watch?v=LTuBEwASEJQ&t=3s}{here}.
A companion \texttt{R} package, \href{https://marcelortiz.com/triplediff/}{\texttt{triplediff}}, is freely available to automate all the DDD estimators proposed in this article.}}

\author{Marcelo Ortiz-Villavicencio \thanks{Department of Economics, Emory University. Email: marcelo.ortiz@emory.edu} \and Pedro H. C. Sant'Anna\thanks{Department of Economics, Emory University. Email: {pedro.santanna@emory.edu}}}
\date{\today}

\maketitle

\begin{abstract}
Triple Differences (DDD) designs are widely used in empirical work to relax parallel trends assumptions in Difference-in-Differences (DiD) settings. This paper highlights that common DDD implementations---such as taking the difference between two DiDs or applying three-way fixed effects regressions---are generally invalid when identification requires conditioning on covariates. In staggered adoption settings, the common DiD practice of pooling all not-yet-treated units as a comparison group can introduce additional bias, even when covariates are not required for identification. These insights challenge conventional empirical strategies and underscore the need for estimators tailored specifically to DDD structures. We develop regression adjustment, inverse probability weighting, and doubly robust estimators that remain valid under covariate-adjusted DDD parallel trends. For staggered designs, we demonstrate how to effectively utilize multiple comparison groups to obtain more informative inferences. Simulations and three empirical applications highlight bias reductions and precision gains relative to standard approaches. A companion \texttt{R} package is available.

\vspace{4ex}
\noindent \textbf{JEL: }C10; C14; C21; C23. \\
\noindent \textbf{Keywords:} Triple Differences; Difference-in-Differences; Difference-in-Difference-in-Differences; Parallel Trends; Doubly Robustness; Staggered Adoption.
\end{abstract}

\sloppy
\newpage

\section{Introduction}

Over the last few years, we have seen a big Difference-in-Differences (DiD) ``methodological revolution'' with multiple DiD estimators being proposed to address the interpretability shortcomings associated with using simple two-way fixed-effects specifications in the presence of treatment effect heterogeneity.\footnote{See, e.g., \citet{Roth2023a} and \citet{Baker_etal_2025_JEL} for overviews.} Although these modern DiD estimators can capture richer notions of heterogeneity, in practice, they rely on parallel trends (PT) assumptions, and a concern relates to how plausible these PT assumptions are. When such PT assumptions are not reasonable enough approximations of reality, one may doubt the conclusions of DiD studies \citep{rambachan_roth, Chiu_etal_2025}.

In some setups, however, it is possible to naturally relax such DiD-type PT assumptions and retain the simplicity and empirical appeal of DiD-type analysis. This is particularly the case when a unit needs to fulfill two criteria to be treated, e.g.,  it belongs to (i) a group (e.g., a state) in which the treatment is enabled, and (ii) a partition of the population that qualifies (or is \emph{eligible}) for treatment (e.g., women). Such setups are often referred to as Triple Differences (DDD), and allow for group-specific and partition-specific violations of parallel trends. Since its introduction by \citet{Gruber1994}, DDD has gained considerable popularity among empirical researchers. Some prominent recent DDD applications include  \citet{antwi_effects_2013}, \citet{Walker_2013_QJE}, \citet{garthwaite_public_2013},  \citet{Alsan_Wanamaker_2018_QJE}, \citet{Patnaik_2019}, \citet{hansen_national_2023}, and \citet{Bailey_Helgerman_Stuart_2024_QJE}---see \citet{Olden2022} for additional documentation of DDD applications, and Section \ref{application} for three DDD applications. Despite its empirical popularity, little attention has been devoted to better understanding the econometric foundations of DDD setups with covariates, multiple periods, and/or staggered treatment adoption. 
    
This article aims to improve our understanding of Triple Difference (DDD) designs.  Our main goal is to provide a set of clear, easy-to-use, and theoretically grounded tools that empirical researchers can utilize whenever they wish to explore DDD designs. To that end, we study identification, estimation, and inference procedures for DDD when covariates may be important for the reliability of the identification assumptions, multiple periods are available, and treatment adoption is potentially staggered over time. We tackle the DDD problem using causal inference first principles and uncover interesting results that challenge some conventional wisdom and common practices.

For instance, although DDD estimators can be understood as the difference between two DiD estimators in setups with two periods and no covariates \citep{Olden2022}, our results highlight that this is no longer the case when covariates are required to justify the plausibility of a DDD-type parallel trends assumption. As we illustrate through simulations, proceeding as if DDD were just the difference between two DiD estimators can lead to severely biased results. Such bias arises because this naive DDD strategy fails to integrate the covariate distribution over the correct reference group---the treated group. We show that it is straightforward to avoid these problems and propose regression adjustment, inverse probability weighting, and doubly robust DDD estimators. We also show that if one wants to cast DDD in terms of DiD, one would need three---and not two---DiD terms. Each of these DiD terms compares effectively treated units with a different type of untreated units, e.g., units in a treated state but ineligible for treatment, units that are eligible but are in an untreated state, or ineligible units in an untreated state. 

In setups with staggered treatment adoption, our results once again challenge the interpretation of DDD as the difference between two DiDs. In DiD with staggered treatment adoptions, it is now common to pool all not-yet-treated units at a time period and use that aggregate set of units as a valid comparison group.\footnote{See, e.g., \citet{Callaway2021, deChaisemartin2020, Wooldridge2021, Borusyaetal2024}.} Thus, one may think that a similar strategy should work with DDD. However, our results suggest that this is generally not the case and that pooling all not-yet-treated units and proceeding as in staggered DiD procedures can lead to biased estimators for average treatment effect parameters, even when covariates do not play a significant role. This arises because the proportion of units eligible for treatment may change across groups that enable treatment over time. As DDD allows for group-specific and partition-specific violations of DiD-type PT, these differential trends do not average out when pooling all not-yet-treated units, leading to potentially misleading estimates. We propose DDD estimators that bypass this drawback by using any specific not-yet-treated unit as a comparison group (e.g., the set of units in groups that never enabled treatment). As one can potentially use different not-yet-treated cohorts as comparison groups, we also discuss combining these to form more precise estimators. Our proposed DDD estimator that aggregates across different comparison groups can be understood as a two-step Generalized Method of Moments (GMM) procedure based on recentered influence functions. Importantly, our staggered DDD procedures can flexibly accommodate covariates using regression adjustment, inverse probability weighting, or doubly robust methods and can also be used to form event-study estimators that highlight how average treatment effects evolve with elapsed treatment time.

We illustrate our results through Monte Carlo simulations and three different empirical applications that address various DDD setups. Our Monte Carlo results highlight that ignoring the key takeaways in this paper and relying on overly rigid estimators can lead to biases and imprecise conclusions; our proposed DDD estimators bypass these limitations. In terms of empirical applications, we (a) revisit \cite{cai_insurance_2016} and analyze the effects of agricultural insurance programs on financial decisions in China; (b) build on \cite{carbon_pricing} and examine the impact of the emission trading scheme on carbon emissions in China; and (c) \cite{hansen_national_2023} and assess the impact of genetically modified crop adoption on countrywide yields. Compared to the three-way fixed effects estimators used by \cite{cai_insurance_2016}, our doubly robust DDD estimator yields substantial gains in precision, with their confidence intervals being up to 115\% wider than ours. In our application examining the effect of the emission trading scheme on carbon emissions with staggered adoption, our doubly robust DDD estimators yield a more modest and statistically insignificant effect. In contrast, the three-way fixed estimators indicate a statistically significant effect on the share of low-carbon patents. When we apply our staggered DDD tools to \cite{hansen_national_2023}'s data, we find that their conclusions are robust to dropping observations that are part of the never-enabling crop-country group, which is not the case when using more standard three-way fixed effects estimators. Taken together, these results highlight that our proposed tools can indeed lead to interesting new insights that are of practical interest.

\textbf{Related literature:} This article contributes to the rapidly expanding literature on DiD-related methods. In particular, we contribute to the scarce literature on DDD procedures. Our paper is related to \citet{Olden2022}, though we cover substantially more general DDD setups with (a) multiple periods, (b) staggered treatment adoption, and (c) covariates potentially playing an important role for the plausibility of the identification assumptions. In this sense, our paper can be understood as the DDD analog of \citet{Callaway2021}. However, and in sharp contrast with \citet{Callaway2021}, our DDD procedures cannot pool all not-yet-treated units as an aggregated comparison group, highlighting some interesting differences between staggered DiD and DDD designs. Our paper is also related to \cite{strezhnev2023}, who introduced a decomposition of the DDD estimators based on three-way fixed effects specifications, illustrating when and why it fails to recover an easily interpretable causal parameter of interest when treatment effects are heterogeneous. To some extent, \citet{strezhnev2023} can be understood as the analog of \citet{Goodman2021} to DDD setups. As such, our results complement \citet{strezhnev2023}, and since our estimators do not rely on a rigid three-way fixed effect specification, they circumvent the issues highlighted in his paper. Our paper is also related to \citet{Sloczyski2022,Sloczyski2024_IVLATE} in the sense that our proposed tools avoid issues related to potentially misleading weights related to model misspecifications.

In this paper, we use the term ``triple differences'' to describe designs under which units must satisfy two criteria to be (effectively) treated. However, we note that sometimes different researchers use the term triple differences more broadly, for instance, when they are interested in analyzing treatment effect heterogeneity across subgroups. In such cases, it is important to note that the underlying identification assumptions and the parameters of interest would differ from those studied in this paper; see \citet{Caron2025} for a recent example. Those procedures should be understood as a complement to the ones we discuss in this paper, as they can be used to answer different questions of interest.

\textbf{Organization of the paper:} The rest of the paper is organized as follows. In the next section, we present our framework. In Section \ref{sec:implication}, we challenge some standard empirical practices for DDD analyses in terms of interpreting them as a simple extension of DiD analysis, and we also highlight some important practical takeaway messages from our main results. Section \ref{sec:ddd-theory} introduces our formal identification, estimation, and inference results. Section \ref{simulations} presents a Monte Carlo study to demonstrate the finite sample properties of our estimator, while Section \ref{application} presents three empirical illustrations. Section \ref{conclusion} concludes. Detailed mathematical proofs and additional results can be found in the Supplemental Appendix. To ease the adoption of our proposed DDD tools, we provide an open-source \texttt{R} package, \texttt{triplediff}, which automates all the procedures described in this article.

\section{Framework}
\label{sec:framework}

We start our analysis by discussing the specifics of our DDD research design, including potential outcomes, parameters of interest, and identification assumptions. We consider a setup with $T$ time periods, $t=1,2,\dots, T$. Units are indexed by $i$, with $i=1,2,\dots, n$. We focus on setups where $n$ is much larger than $T$, as our inference procedures are asymptotically justified using the ``fixed-$T$, large-$n$'' panel data framework. 

Each unit may be exposed to a binary treatment in any time period $t>1$. Treatment is an absorbing state such that once a unit is treated, it remains treated for the remainder of the panel. Each unit $i$ belongs to a group (e.g., a state or a country) that enables treatment for the first time in period $g>1$. Let $S_i \in \mathcal{S} \subseteq \{2,...,T\} \cup \{\infty\}$ be a variable that indicates the first time the policy/treatment was enabled, with the notion that $S = \infty$ if the policy is not enabled by $t=T$. In addition, each unit belongs to a population partition that qualifies (or is eligible) for the treatment or not (e.g., being a woman, or an indicator for specific crops). We denote this variable by $Q_i$ with $Q_i=1$ if unit $i$ (eventually) qualifies for the treatment and $Q_i=0$ otherwise. For simplicity, we assume that this population partition that eventually qualifies for treatment is time-invariant. 

In our DDD setup, a unit $i$ is treated in period $t$ if $t \ge S_i$ and $Q_i=1$, i.e., if it belongs to a group that has already enabled treatment by period $t$ (i.e., $t \ge S_i$) and it qualifies for treatment (i.e., $Q_i=1$). With this notation that makes it clear that a unit $i$ is treated if it satisfies two criteria, let $D_{i,t} = 1\{t\ge S_i, Q_i = 1\}$ be an indicator for whether unit $i$ receives treatment in period $t$, and let $G_i = \min\{t : D_{i,t} = 1\}$ be the earliest period at which unit $i$ has received treatment. If $i$ is never treated during the sample, then $G_i = \infty$. Here, we have that $G_i = S_i$ if $Q_i = 1$ and $G_i = \infty$ if $Q_i = 0$.\footnote{Note that when all units are eligible for treatment, we have $G_i = S_i$, getting us back to a (staggered) DiD setup; see, e.g., \citet{Callaway2021} and \citet{Sun2021}.} Let $\mathcal{G}$ denote the support of $G_i$ and $\mathcal{G}_{\text{trt}} = \mathcal{G}\setminus \{\infty\}$. We assume that a group of ``never-enabled'' units always exists, i.e., $S_i = \infty$ for some units. In an application where all units belong to a group that eventually enables treatment, we remove all that data from all units from the time the last cohort enabled treatment onwards, i.e., we drop all observations from periods $t>\max S_i$, and retain the remaining data as the ``effective'' data to be used in our analysis, where the last-to-be-eligible group becomes the ``never-eligible'' group.\footnote{If needed, we also update the notion of support of all variables to reflect this change. In particular, $T$ here denotes the number of available periods in the subset of the data that we will use in our analysis.} Finally, we also assume that a vector of pre-treatment covariates $X_{i}$, whose support is denoted by $\mathcal{X}\subseteq \mathbb{R}^d$ is available.

Regarding potential outcomes, we adopt the potential outcome framework of \citet{Robins1986} with potential outcomes indexed by treatment sequences. Let $\mathbf{0}_s$ and $\mathbf{1}_s$ be $s$-dimensional vectors of zeros and ones, respectively, and denote the potential outcome for unit $i$ at time $t$ if unit $i$ is first treated at time $g$ by $Y_{i,t}(\mathbf{0}_{g-1},\mathbf{1}_{T-g+1})$, and denote by $Y_{i,t}(\mathbf{0}_{T})$ the outcome if untreated by time $t=T$. As we focus our attention on staggered treatment adoptions, we can simplify notation and index potential outcomes by the time treatment begins, $g$: $Y_{i,t}(g) = Y_{i,t}(\mathbf{0}_{g-1},\mathbf{1}_{T-g+1}) $ and use $Y_{i,t}(\infty) = Y_{i,t}(\mathbf{0}_{T})$ to denote never-treated potential outcomes.  In practice, though, we observe, 
\begin{equation}
\label{def:observed_outcome}
    Y_{i,t} = \sum_{g \in \mathcal{G}} 1\{G_i = g\} Y_{i,t}(g),
\end{equation}
where $1\{A\}$ represents the indicator function, which equals one if $A$ is true and zero otherwise. Additionally, we assume the observation of a random sample of $(Y_{t=1}, \dots, Y_{t=T}, X', G, S, Q)'$.
\begin{namedassumption}{S}[Random Sampling] \label{ass:sampling_panel} $\{(Y_{i,t=1},\dots,Y_{i,t=T}, X_i', G_i, S_i, Q_i)'\}_{i=1}^n$ is a random sample from $(Y_{t=1},\dots,Y_{t=T}, X', G,S,Q)'$.
\end{namedassumption}

\subsection{Parameters of interest}

In this paper, we aim to gain a deeper understanding of how average treatment effects vary across periods and different groups defined by the treatment starting period. More specifically, we want to make inferences on functionals of the group-time average treatment effects, $ATT(g,t)$'s, defined as
\begin{align}
ATT(g,t) \equiv \E[Y_{i,t}(g)-Y_{i,t}(\infty)|G_i=g] = \E[Y_{i,t}(g)-Y_{i,t}(\infty)|S_i=g, Q_i = 1],\label{eqn:att_gt}
\end{align}
By exploring that in our context $G_i=g$ if and only if $S_i=g$ and $Q_i = 1$, we have that $ATT(g,t) = \E[Y_{i,t}(g)-Y_{i,t}(\infty)|S_i=g, Q_i = 1]$. Note that $ATT(g,t)$ captures how the average treatment effects evolve over time for each treatment group $g$ \citep{Callaway2021}. As such, one can use $ATT(g,t)$ to construct group-$g$-specific event studies by analyzing how average treatment effects vary with elapsed treatment time $e=t-g$. 

In some setups with multiple groups $g$, researchers may want to summarize over the many $ATT(g,t)$'s into a more aggregate parameter. A natural summary parameter that still allows one to understand treatment effect dynamics with respect to elapsed treatment time is the aggregated event study parameter $ES(e)$, defined as
\begin{align}
ES(e) \equiv \E\big[ATT\left( G, G+e\right) \big| G+e \in [2,T]\big] = \sum_{g\in\mathcal{G}_{\text{trt}}} \P(G=g|G+e \in [2,T]) ATT(g,g+e).\label{eqn:ES}
\end{align}

One may also want to aggregate the event study coefficients further to recover a scalar summary measure. Let $\mathcal{E}$ denote the support of post-treatment event time $E=t-G$, $t\geq G$, and let $N_E$ denote its cardinality. Then, 
\begin{align}
ES_{\text{avg}} &\equiv \dfrac{1}{N_E}\sum_{e\in \mathcal{E}} ES(e)\label{eqn:overall_ATT},
\end{align}
provides a simple average of all post-treatment event study coefficients. Many other summary parameters are possible; see \citet{Callaway2021} for discussions of several alternatives.

\subsection{Identification assumptions}\label{sec:id_assumptions}
To identify the $ATT(g,t)$'s and their functionals $ES(e)$ and $ES_{\text{avg}}$, we impose the following assumptions.

\begin{namedassumption}{SO}[Strong Overlap]
    \label{ass:overlap_staggered}
For every $(g,q) \in \mathcal{S} \times \{0,1\}$ and for some $\epsilon>0$, $\mathbb{P}[S=g, Q=q |X] > \epsilon$ with probability one.
\end{namedassumption}
Assumption \ref{ass:overlap_staggered} is an overlap condition that ensures that for any value of $X$, there are units with any combination $(g,q) \in \mathcal{S}$ that have comparable $X$ values. Heuristically, this condition guarantees that we cannot perfectly predict which $(g,q)$-partition a unit belongs to using information from $X$. This assumption also rules out irregular identification \citep{Khan2010}.\footnote{As our focus is on $ATT(g,t)$-type parameters, it is possible to relax Assumption \ref{ass:overlap_staggered} to hold only over $X$ in the support of the covariates among the (eventually) treated units. To simplify the discussion, we abstract from these subtle points.}

We also impose the following no-anticipation assumptions.
\begin{namedassumption}{NA}[No-Anticipation]
\label{ass:anticipation}
   For every $g \in \mathcal{G}_{\text{trt}}$, and every pre-treatment period $t<g$, $\E[Y_{i,t}(g)|S=g, Q = 1,X] = \E[Y_{i,t}(\infty)|S=g, Q = 1,X]$ with probability one.
\end{namedassumption}
Assumption \ref{ass:anticipation} rules out anticipatory effects among treated units as in, e.g., \citet{Abbring2003}, \citet{Callaway2021}, and \citet{Sun2021}. This assumption is important as it allows us to consider observations in pre-treatment periods $t<g$ as effectively untreated. If units are expected to anticipate some treatments---for example, if treatment is announced in advance---it is important to adjust the definition of the treatment date to account for this; see \citet{Malani2015} for a discussion.

Next, we impose our final identification assumption that restricts the evolution of average untreated potential outcomes across groups.

\begin{namedassumption}{DDD-CPT}[DDD-Conditional Parallel Trends] \label{ass:PT-NYT}
For each $g \in \mathcal{G}_{\text{trt}}$,  $g' \in \mathcal{S}$ and time periods $t$ such that $t\ge g$ and $g'>\max\{g,t\}$, with probability one,
\begin{eqnarray*}
\mathbb{E}\left[Y_{t}(\infty) - Y_{t-1}(\infty)| S = g, Q=1, X\right] &-& \mathbb{E}\left[Y_{t}(\infty)- Y_{t-1}(\infty) | S = g, Q=0, X \right]  \\[-0.2cm]
&=&\\[-0.2cm]
\mathbb{E}\left[Y_{t}(\infty) - Y_{t-1}(\infty)|S = g', Q=1, X\right] &-& \mathbb{E}\left[Y_{t}(\infty)- Y_{t-1}(\infty) | S = g', Q=0, X \right].
\end{eqnarray*}
\end{namedassumption}

Assumption \ref{ass:PT-NYT} is a conditional parallel trends assumption for DDD setups that generalizes the unconditional DDD parallel trend assumption for the two-period setup of \citet{Olden2022} to setups with multiple periods, staggered treatment adoption, and when assumptions are only plausible after conditioning on $X$. Assumption \ref{ass:PT-NYT} can also be understood as an extension of the conditional PT assumption based on not-yet-treated units from the DiD setup of \citet{Callaway2021} to our DDD setup--- i.e., we can use any unit from groups that either never enabled treatment or those that will eventually enable treatment. Moreover, if covariates do not play any important identification role in the analysis, one can take $X=1$ for all units, so Assumption \ref{ass:PT-NYT} would hold unconditionally.

Several remarks about Assumption \ref{ass:PT-NYT} are worth making. First, if all units in a group $S$ are eligible for treatment, Assumption \ref{ass:PT-NYT} reduces to Assumption 5 of \citet{Callaway2021} under the no-anticipation condition in Assumption \ref{ass:anticipation}. However, this case is not appealing to us, as that would not qualify as a DDD design. Second, as Assumption \ref{ass:PT-NYT} only holds after conditioning on covariates, it does not restrict the evolution of untreated potential outcomes across covariate-strata, i.e., it allows for covariate-specific trends, which can be very important in applications. Third, and perhaps the most empirically relevant, Assumption \ref{ass:PT-NYT} 
does not impose DiD-type parallel trends among units with $S=g$---i.e., it does not impose that $\mathbb{E}\left[Y_{t}(\infty) - Y_{t-1}(\infty)| S = g, Q=1, X\right] = \mathbb{E}\left[Y_{t}(\infty)- Y_{t-1}(\infty) | S = g, Q=0, X \right]$---nor impose DiD-type parallel trends across treated groups---i.e., it does not impose that $\mathbb{E}\left[Y_{t}(\infty) - Y_{t-1}(\infty)| S = g, Q=1, X\right] = \mathbb{E}\left[Y_{t}(\infty) - Y_{t-1}(\infty)|S = g', Q=1, X\right]$. As such, Assumption \ref{ass:PT-NYT} allows for violations of traditional DiD-based PT, as provided that these violations are stable across groups. This observation is arguably what makes DDD appealing in settings where it can be applied.

\section{Implications for Empirical Practices}\label{sec:implication}

Before presenting our formal results on identification, estimation, and inference for average treatment effects in DDD designs, we challenge some standard empirical practices for DDD analyses and highlight some important practical takeaways from our paper. 

We first start with a simple setup involving only two periods, $t=1$ and $t=2$, and two eligibility groups, $S_i = 2$ (who enabled treatment in period 2) and $S_i = \infty$ (who have not enabled treatment by period two). As before, units are either eligible ($Q_i =1$) or ineligible ($Q_i = 0$) for the treatment, and we let $D_{i,t}$ be a treatment indicator for unit $i$ in time period $t$, i.e., $D_{i,t} = 1\{t\ge S_i, Q_i = 1\}$. Since there are only two eligibility groups and two time periods, the relevant group-time ATT in such a scenario is $ATT(2,2)$. As discussed in \citet{Olden2022}, when covariates are not important for the analysis, one can use ordinary least squares (OLS) based on the following three-way fixed effects linear regression specification to recover the $ATT(2,2)$:
  \begin{align}
    Y_{i,t}=& \gamma_i + \gamma_{s,t} + \gamma_{q,t} + \beta_{\text{3wfe}} D_{i,t} + \varepsilon _{i,t}, \label{eqn:3WFE_no_X}
    \end{align}
where $\gamma_i$ are unit fixed effects, $\gamma_{s,t}$ and $\gamma_{q,t}$ are enabled-group-by-time  and qualified-group-by-time fixed effects, and $\beta_{\text{3wfe}}$ is the parameter of interest. Indeed, in this particular setup, under Assumptions \ref{ass:sampling_panel}, \ref{ass:overlap_staggered}, \ref{ass:anticipation}, and \ref{ass:PT-NYT} with $X=1$ a.s.,  it is straightforward to show that
 \begin{align}
  \beta_{\text{3wfe}}  =&~\Bigg[\underbrace{\bigg(\mathbb{E}\left[Y_{2} - Y_{1}|S = 2, Q=1\right]\bigg)-\bigg( \mathbb{E}\left[ Y_{2} - Y_{1} | S = 2, Q=0 \right] \bigg)}_{\text{DiD estimand among } S=2}\Bigg]  \nonumber\\
&~~~~- \Bigg[\underbrace{\bigg(\mathbb{E}\left[ Y_{2} - Y_{1}| S = \infty, Q=1 \right] \bigg)\Bigg.-\Bigg.\bigg(\mathbb{E}\left[ Y_{2} - Y_{1}|S=\infty, Q=0 \right]\bigg)}_{\text{DiD estimand among } S=\infty}\Bigg] \label{eqn:3wtfe_att22}\\
=& ~ ATT(2,2) \nonumber.
    \end{align}

The observation that $\beta_{\text{3wfe}} = ATT(2,2)$ in this simple setup has two implications: (i) one can use a simple three-way fixed effects (3WFE) regression specification and use OLS to estimate $ATT(2,2)$ in the DDD design, and (ii) DDD estimates can be understood as the difference between two DiD estimates \citep{Olden2022}. Based on these, one may be tempted to extrapolate these claims to more general setups. In what follows, we highlight that, unfortunately, this is not warranted and that proceeding in this manner can lead to non-negligible biases. The solution to these issues is relatively simple and involves adopting a ``forward-engineering'' approach to DDD setups \citep{Baker_etal_2025_JEL}, recognizing its specific characteristics. See also \citet{Mogstad_Torgovitsky_2024} for a related discussion in an instrumental variable context.

\subsection{DDD setup with two periods, with covariates being important}\label{sec:ddd_2periods_x}
In this section, we illustrate the challenges of leveraging simple regression-based and DiD tools to DDD setups using simple simulations in a setup where covariates are important for identification, i.e., when Assumption \ref{ass:PT-NYT} is satisfied only after accounting for covariates. We consider a design with four different time-invariant, unit-specific covariates, $X_i = (X_{i,1}, X_{i,2}, X_{i,3}, X_{i,4})'$, and two periods and two treatment-enabling groups. The true $ATT(2,2)$ in our simulations is zero. To ease the exposition, we abstract from further details about the DGP and refer the reader to Section \ref{sec:sims_2_periods} and Supplemental Appendix \ref{appendix:sims_2_periods} for a more detailed discussion.

Based on the discussion on DDD without covariates above, it is natural to consider three alternative ways to incorporate covariates in the analysis. The first approach would be to ``extrapolate'' from \eqref{eqn:3WFE_no_X}, add the interactions of the time-invariant covariates with post-treatment dummies, 
  \begin{align}
    Y_{i,t}=& \gamma_i + \gamma_{s,t} + \gamma_{q,t} + \tilde{\beta}_{\text{3wfe}} D_{i,t} + (X_i 1_{ \{t=2\}})'\theta + u_{i,t}, \label{eqn:3WFE_with_X}
    \end{align}
and interpret the OLS estimates of $\tilde{\beta}_{\text{3wfe}}$ as estimates of $ATT(2,2)$. The second natural way to proceed is similar, but it would leverage the Mundlak device and replace unit fixed effects in \eqref{eqn:3WFE_no_X} with $S$-by-$Q$ fixed effects, add covariates linearly, 
 \begin{align}
    Y_{i,t}=& \gamma_{s,q} + \gamma_{s,t} + \gamma_{q,t} + \check{\beta}_{\text{3wfe}} D_{i,t} + X_i'\theta + e_{i,t}, \label{eqn:3WFE_with_X_Mundlak}
    \end{align}
and interpret the OLS estimates of $\check{\beta}_{\text{3wfe}}$ as estimates of $ATT(2,2)$. Both strategies leverage a presumption that it is sufficient to add covariates linearly into the 3WFE regression specification \eqref{eqn:3WFE_no_X} to allow for covariate-specific trends. A third strategy, which is also a priori intuitive, presumes that we can write DDD estimates as the difference between two DiD estimates: one DiD using the subset with $S=2$ and considering units treated if $Q=1$, and another DiD using the subset with $S=\infty$ and considering units treated if $Q=1$. Here, one could consider different estimation strategies. We focus on the doubly robust (DR) DiD estimators proposed by \citet{SantAnna2020}, as a DR estimator is generally more resilient to model misspecifications.

\begin{figure}[!htp!]
	\begin{center}
		\begin{subfigure}[t]{0.48\textwidth}
			\centering
			\includegraphics[width = \textwidth]{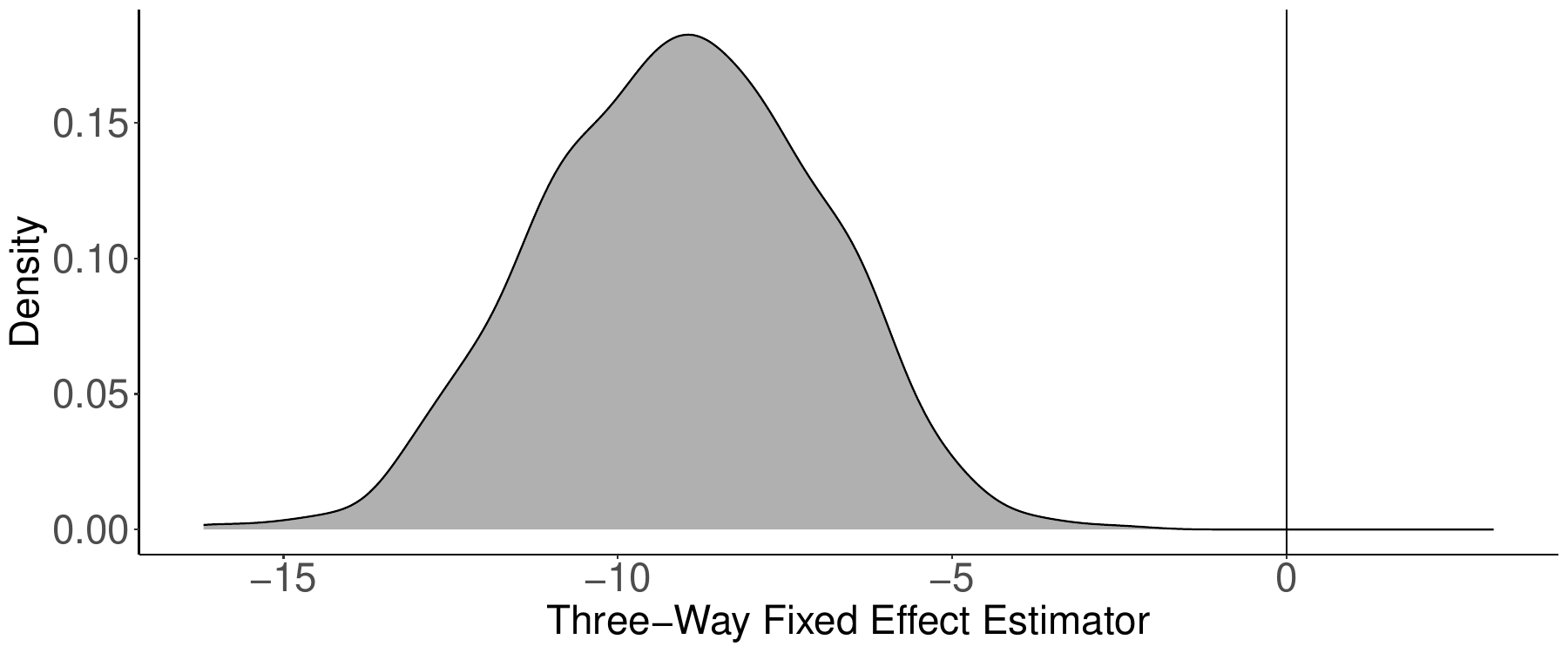}
			\caption{3WFE with covariates interacted with post}
			\label{fig:twfe_x_post}
		\end{subfigure}
		\hfill
		\begin{subfigure}[t]{0.48\textwidth}
		  \begin{center}
            \includegraphics[width = \textwidth]{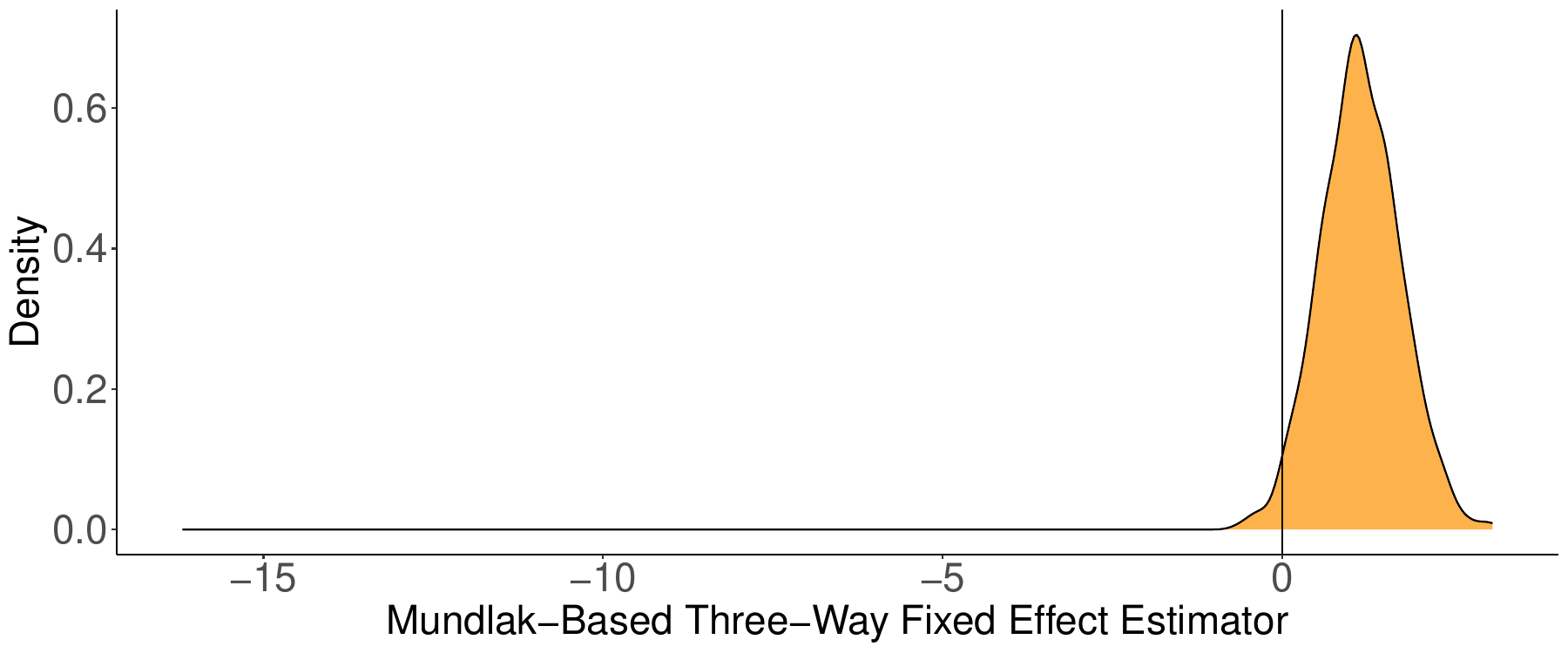}
            \caption{Mundlak-based 3WFE with covariates}
            \label{fig:Mudlak}
            \end{center}
		\end{subfigure}
        \begin{subfigure}[t]{0.48\textwidth}
			\centering
			\includegraphics[width = \textwidth]{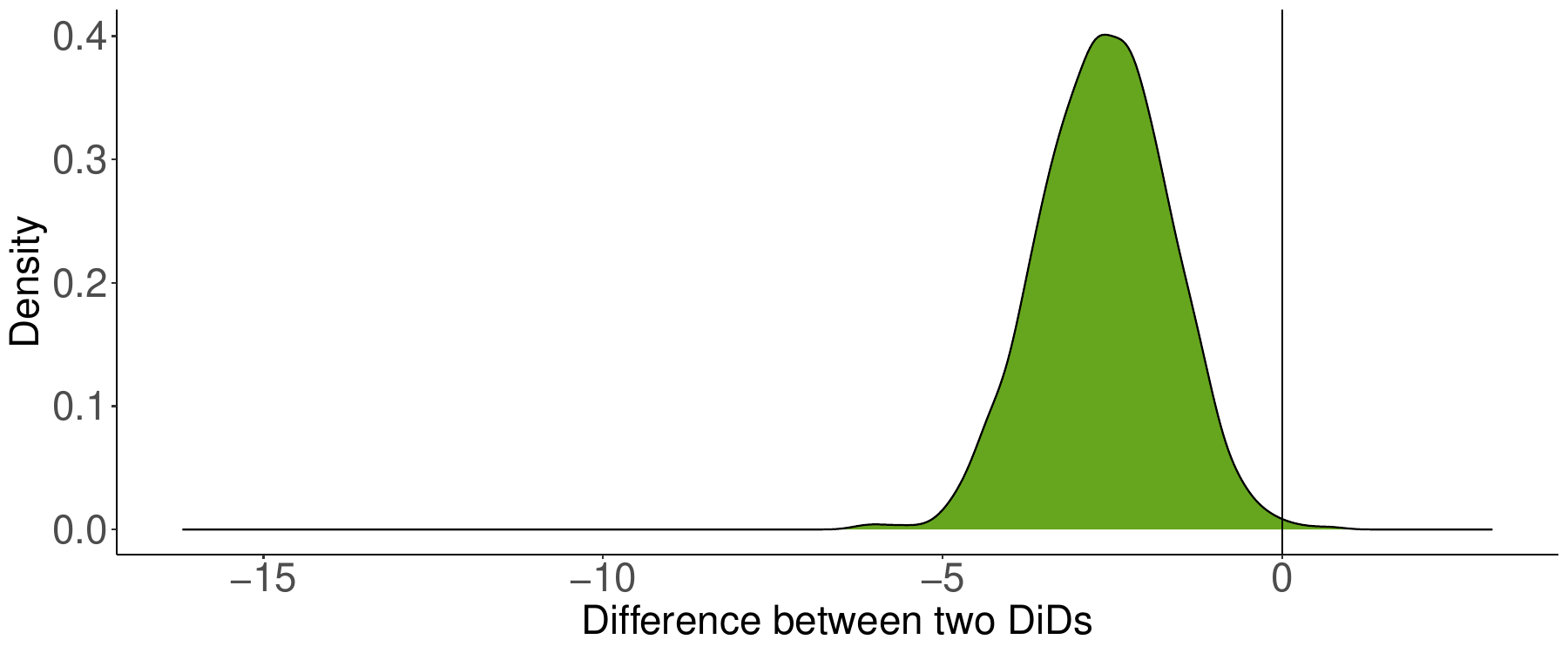}
			\caption{Difference between two Doubly Robust DiDs}
			\label{fig:dif_DRDID}
		\end{subfigure}
		\hfill
		\begin{subfigure}[t]{0.48\textwidth}
		  \begin{center}
            \includegraphics[width = \textwidth]{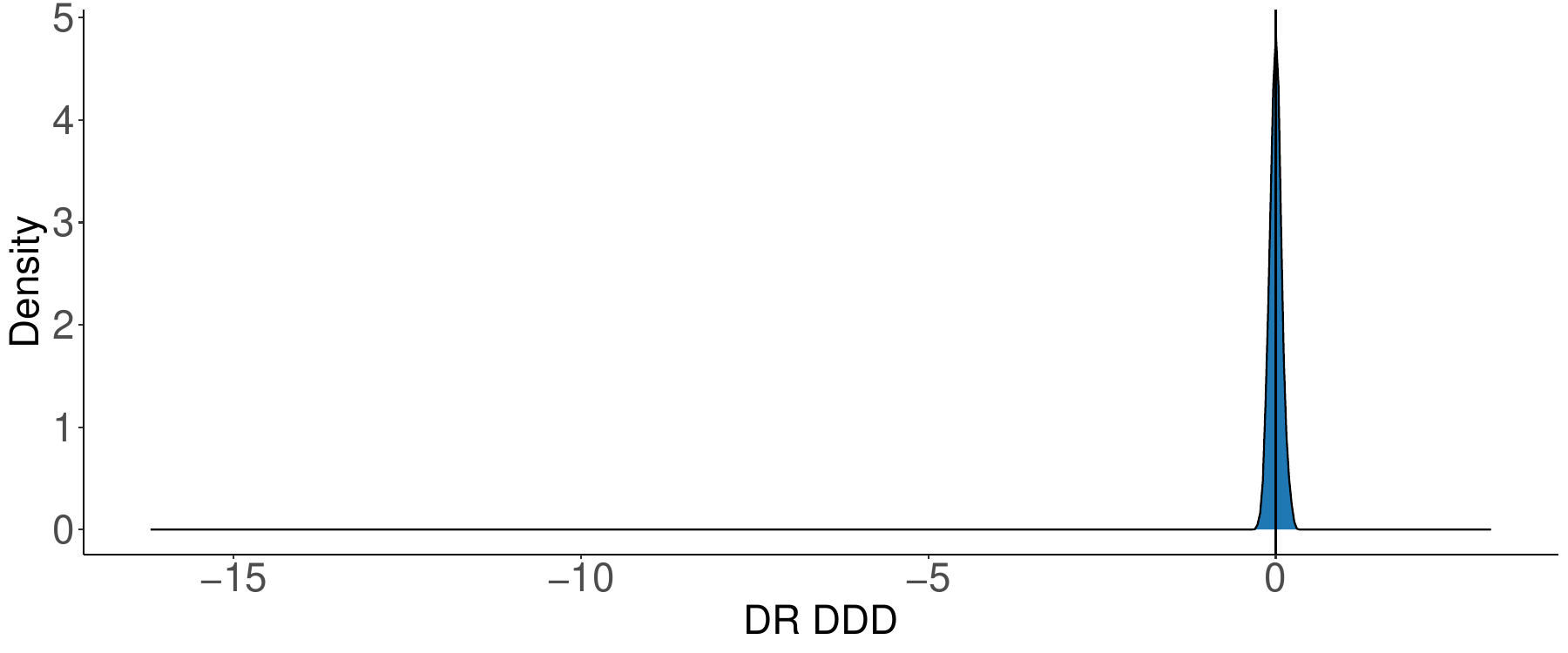}
            \caption{Doubly Robust DDD }
            \label{fig:DRDDD}
            \end{center}
		\end{subfigure}
		\caption{Density of different DDD estimates for ATT(2,2): two-period setup with covariates}
		\label{fig:Sims_2_periods}
	\end{center}
        \justifying
	\vspace{-.2cm}\scriptsize{Notes: Simulation designs based on DGP 1 described in Section \ref{sec:sims_2_periods} and Supplemental Appendix \ref{appendix:sims_2_periods}, with $n=5,000$ and $1,000$ Monte Carlo repetitions. True $ATT(2,2)$ is zero and is indicated in the solid vertical line in all panels. Panel (a) displays the density of OLS estimates of $\tilde{\beta}_{\text{3wfe}}$ based on \eqref{eqn:3WFE_with_X}. Panel (b) displays the density of OLS estimates of $\check{\beta}_{\text{3wfe}}$ based on \eqref{eqn:3WFE_with_X_Mundlak}. Panel (c) displays the density of the DDD estimates based on the difference between two doubly robust DiD estimators \citep{SantAnna2020}. Panel (d) displays the density of the estimates based on our proposed doubly robust DDD estimator described in \eqref{eqn:DRDDD_2periods}. All densities are computed across all simulation draws. Panels have the same x-axis range but different y-axis.
 }
\end{figure}

To check if such alternative strategies recover the $ATT(2,2)$, we draw $5,000$ units in each simulation draw, compute estimates using these three alternative estimators, and repeat this $1,000$ times---we defer all details of the data generating process to Section \ref{sec:sims_2_periods} and Supplemental Appendix \ref{appendix:sims_2_periods}.  Panels (a) and (b) from Figure \ref{fig:Sims_2_periods} display the density of OLS estimates for the $D_{i,t}$ coefficient in the regression specifications \eqref{eqn:3WFE_with_X} and \eqref{eqn:3WFE_with_X_Mundlak}, while Panel (c) displays the density of the DDD estimates based on the difference between two \citet{SantAnna2020} DR DiD estimators. These three panels make it clear that when covariates are necessary to justify the plausibility of the DDD research design, using any of these three procedures can lead to substantial biases and harm policy recommendations and evaluations. In other words, these results highlight that traditional 3WFE linear regression specifications are ``too rigid'' to be reliable for DDD analysis. They also highlight that, in general, one should not claim that DDD is the difference between two DiD procedures.

A natural question that then arises is: What should we do instead? As we discuss in Section \ref{sec:ddd-theory}, one can form regression adjustment, inverse probability weighting, and doubly robust DDD estimators that do not suffer from the shortcomings highlighted in Panels (a) - (c) in Figure \ref{fig:Sims_2_periods}. Among these, we generally favor the DR DDD estimator as it is more resilient against model misspecifications than the other alternatives. To form the DR DDD estimator for $ATT(2,2)$, we need to estimates for the outcome regression models $m_{Y_2-Y_1}^{S=g,Q=q}(X) \equiv \mathbb{E}\left[Y_2 - Y_1|S=g,Q=q ,X\right]$, and for the generalized propensity score model $p^{S=g,Q=q}(X) \equiv \mathbb{P}[S=g, Q=q | X]$. Let $\widehat{m}_{Y_1-Y_0}^{S=g,Q=q}(X)$ and $\widehat{p}^{S=g,Q=q}(X)$ be working models for these---e.g, a linear regression model and a multinomial logistic linear model, though much richer, potentially machine-learning based estimators can also be used \citep{Ahrens2025_DML_JEL}. Based on these estimates, we propose the following DR DDD estimator for the $ATT(2,2)$:
\begin{small}
\begin{align}\label{eqn:DRDDD_2periods}
\widehat{ATT}_{\text{dr}}(2,2) =&~~ \mathbb{E}_n\left[ \left(\widehat{w}^{S=2,Q=1}_{\text{trt}}(S,Q) - \widehat{w}^{S=2,Q=0}_{\text{comp}}(S,Q,X) \right)\left(Y_2- Y_1 - \widehat{m}_{Y_2-Y_1}^{S=2,Q=0}(X) \right)\right]\nonumber \\
& +\mathbb{E}_n\left[ \left(\widehat{w}^{S=2,Q=1}_{\text{trt}}(S,Q) - \widehat{w}^{S=\infty,Q=1}_{\text{comp}}(S,Q,X) \right)\left(Y_2- Y_1 -  \widehat{m}_{Y_2-Y_1}^{S=\infty,Q=1}(X)  \right)\right] \\
&- \mathbb{E}_n\left[ \left(\widehat{w}^{S=2,Q=1}_{\text{trt}}(S,Q) -\widehat{w}^{S=\infty,Q=0}_{\text{comp}}(S,Q,X)\right) \left(Y_2- Y_1 - \widehat{m}_{Y_2-Y_1}^{S=\infty,Q=0}(X)\right)\right] \nonumber,
\end{align}
\end{small}
where $\mathbb{E}_n [A] = n^{-1}\sum_{i=1}^n A_i$ denotes the sample mean, and the estimated weights $\widehat{w}$ are given by
\begin{small}
\begin{align*}
    \widehat{w}^{S=2,Q=1}_{\text{trt}}(S, Q)\equiv\dfrac{1{\{S=2,Q=1\}} }{\mathbb{E}_n[1{\{S=2,Q=1\}}]}, \quad \widehat{w}^{S=g, Q=q}_{\text{comp}}(S, Q, X) \equiv \dfrac{\dfrac{1{\{S=g,Q=q\}} \cdot \widehat{p}^{S=2,Q=1}(X) }{\widehat{p}^{S=g,Q=q}(X)}} {\mathbb{E}_n\bracks{\dfrac{1{\{S=g,Q=q\}} \cdot \widehat{p}^{S=2,Q=1}(X) }{\widehat{p}^{S=g,Q=q}(X)}}}.
\end{align*}
\end{small}
Interestingly, it is worth mentioning that although the DR DDD estimator in \eqref{eqn:DRDDD_2periods} cannot be expressed as the difference between two DR DiD estimators, it is a function of \emph{three} DR DiD estimators, each one using a particular subset of the untreated units as a comparison group.

For comparisons, we report in Panel (d) of Figure \ref{fig:Sims_2_periods} the density of the estimates using the DR DDD estimates based on \eqref{eqn:DRDDD_2periods}. Our proposed DR DDD estimator not only mitigates the biases associated with other estimation strategies but also yields substantially more precise estimates. All in all, the results in Figure \ref{fig:Sims_2_periods} highlight that common DDD practices can lead to misleading conclusions. However, it is straightforward to bypass these limitations by adopting our DR DDD estimators.

\subsection{DDD setups with variation in treatment timing}\label{sec:ddd_stag_noX}

The practical challenges of estimating average treatment effects in DDD setups are not confined to the presence of covariates. Even in designs without covariates, the use of too-rigid 3WFE regression specifications like \eqref{eqn:3WFE_no_X} can lead to misleading estimates when there is variation in treatment timing across groups \citep{strezhnev2023}. In such cases, new identification and estimation concerns emerge that the recent DiD literature does not address. In particular, in this section, we highlight that, unlike in staggered DiD procedures like \citet{Callaway2021}, \citet{Borusyaetal2024}, and \citet{Wooldridge2021}\footnote{See also \citet{deChaisemartin2020, deChaisemartin2023b-intertemporal-treatments} for related procedures that also pool not-yet-treated units. Their estimators allow for treatment turning on and off. However, they impose additional assumptions that restrict how past treatments affect future outcomes. We do not consider these setups in this paper.}, pooling all not-yet-treated units and using them as a comparison group does not respect the triple-differences identification assumptions and, as such, can lead to biased estimates for the parameters of interest. We also discuss straightforward and computationally simple estimators that bypass these problems. Throughout this section, we assume that all identification assumptions discussed in Section \ref{sec:id_assumptions} hold without covariates, i.e., by taking $X=1$ almost surely. Our methods naturally extend to setups where covariates are necessary for identification, as discussed in Section~\ref{sec:ddd-theory}. There, we also explore extensions to event-study aggregations and treatment effect heterogeneity across groups and time.

To build intuition, we begin by noting that the way the DiD literature has addressed the shortcomings of using regression specifications akin to \eqref{eqn:3WFE_no_X} to infer overall average treatment effects is to decompose the problem into a series of $2$-period $2$-group ($2\times 2$) DiDs; for an overview, see \citet{Roth2023a} and \citet{Baker_etal_2025_JEL}. A popular strategy involves using the units not yet treated by period $t$ as a comparison group when estimating $ATT(g,t)$ (\citealp{Callaway2021}, \citealp{Borusyaetal2024}, \citealp{Wooldridge2021}). It is thus intuitive and natural to build on these DiD papers, \citet{Olden2022}'s DDD procedure, and the linear regression specification \eqref{eqn:3wtfe_att22}, and attempt to estimate $ATT(g,t)$ in a DDD setup using
\begin{small}
 \begin{align}
  \widehat{ATT}_{\text{cs-nyt}}(g,t)  =&~\Bigg[\bigg(\mathbb{E}_n\left[Y_{t} - Y_{g-1}|S = g, Q=1\right]\bigg)-\bigg( \mathbb{E}_n\left[ Y_{t} - Y_{g-1} | S = g, Q=0 \right] \bigg)\Bigg]  \nonumber\\
&- \Bigg[\bigg(\mathbb{E}_n\left[ Y_{t} - Y_{g-1}| S >t, Q=1 \right] \bigg)\Bigg.-\Bigg.\bigg(\mathbb{E}_n\left[ Y_{t} - Y_{g-1}|S>t, Q=0 \right]\bigg)\Bigg] \label{eqn:att_gt_pooled}
    \end{align}
\end{small}
\noindent in any post-treatment periods $t\ge g$.\footnote{Since there are no covariates, we do not need to use three DiDs as we discussed in Section \ref{sec:ddd_2periods_x}. We use the notation $\text{cs}$ in \eqref{eqn:att_gt_pooled} to denote the estimator discussed above, which pinpoints the baseline period at period $g-1$.} The question now is whether \eqref{eqn:att_gt_pooled} indeed recovers $ATT(g,t)$'s under our identification assumptions.

To answer this practically relevant question, we conduct some Monte Carlo simulations for a setup with three time periods, $t=1,2,3$, three treatment-enabling groups, $S\in\{2,3,\infty\}$, and two eligibility groups $Q=1$ and $Q=0$. We focus on $ATT(2,2)$, i.e., the average treatment effect in period two of being treated in period two, among units treated in period two. The true $ATT(2,2)$ in our simulations is 10. We considered a setup with $n=5,000$ and conducted $1,000$ simulation draws. To ease the exposition, we abstract from further details about the DGP and refer the reader to Section \ref{sims_staggered} and Supplemental Appendix \ref{appendix:sims_staggered} for a more detailed discussion. 

Panel (a) from Figure \ref{fig:Sims_stagg} displays the density of the DDD estimates for $ATT(2,2)$ based on the estimator in \eqref{eqn:att_gt_pooled}. This result makes it clear that, in general, \eqref{eqn:att_gt_pooled} is not a valid estimator for the $ATT(2,2)$ in DDD setups, as it is systematically biased. In fact, in our simulations, \eqref{eqn:att_gt_pooled} always leads to a negative estimate while the true effect is positive. This bias arises because the DDD parallel trends assumption is more flexible than its DiD counterpart: it allows for treatment-enabling-groups- and partition-specific violations of DiD-type parallel trends. In particular, when the fraction of eligible units differs across treatment-enabling groups $S$, pooling not-yet-treated units may conflate trends across heterogeneous populations, violating the assumptions necessary to interpret differences as causal.

\begin{figure}[!htp!]
	\begin{center}
		\begin{subfigure}[t]{0.48\textwidth}
			\centering
			\includegraphics[width = \textwidth]{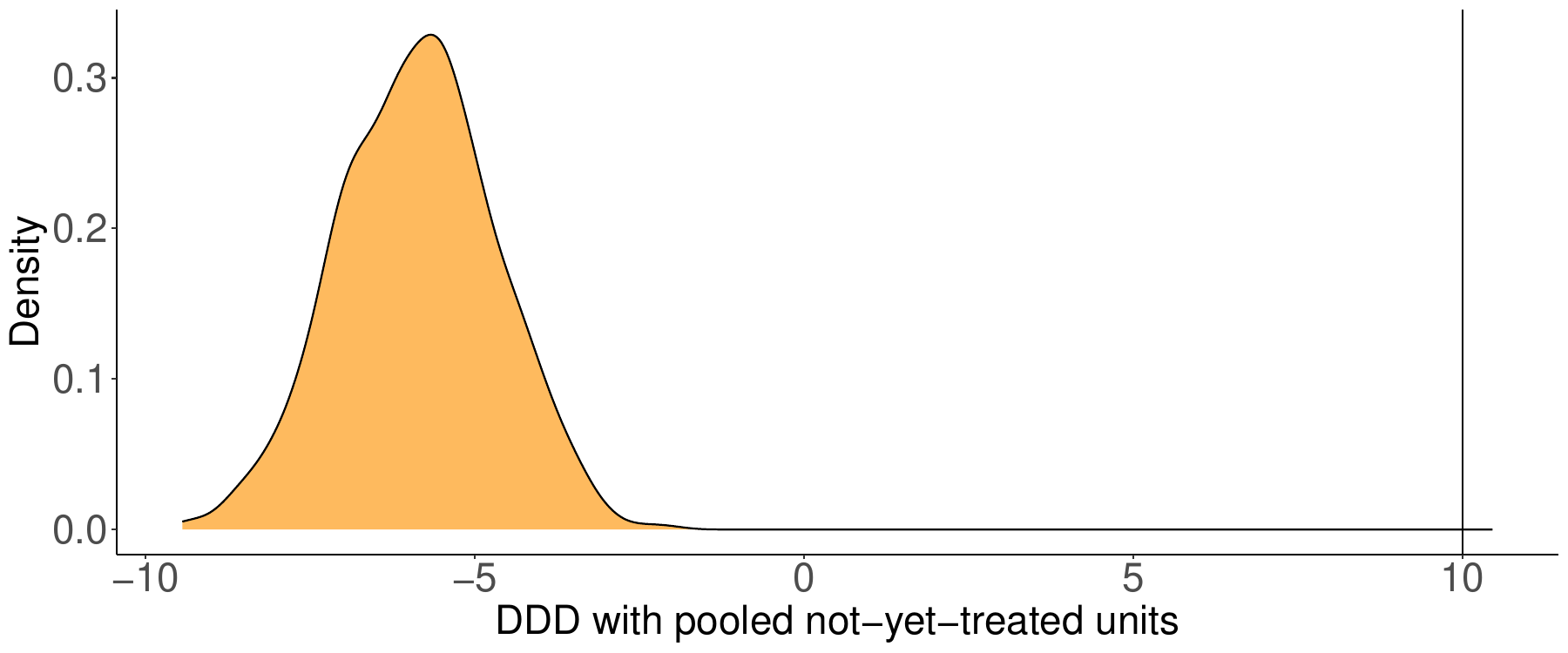}
			\caption{DDD using pooled not-yet-treated units }
			\label{fig:stag_ddd_pooled_nyt}
		\end{subfigure}
		\hfill
		\begin{subfigure}[t]{0.48\textwidth}
		  \begin{center}
            \includegraphics[width = \textwidth]{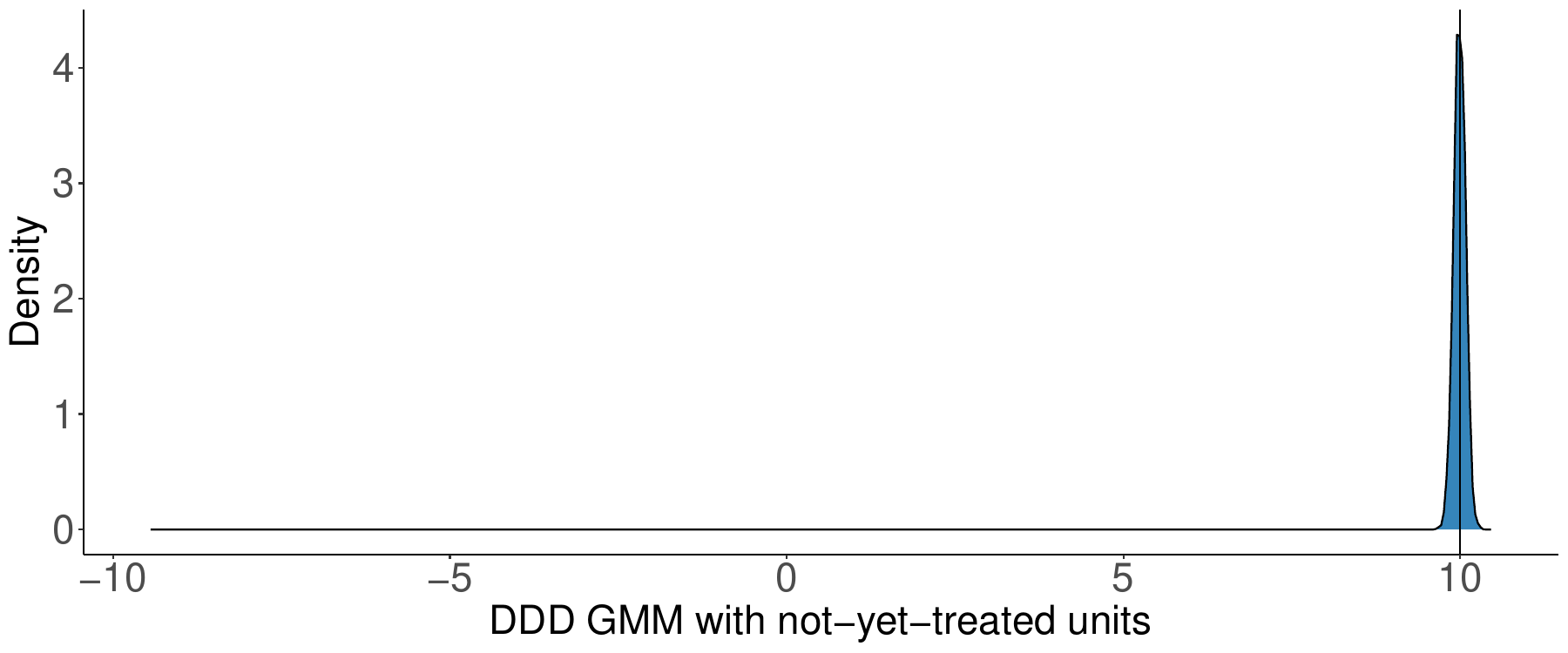}
            \caption{DDD GMM using all not-yet-treated units}
            \label{fig:stag_ddd_gmm}
            \end{center}
		\end{subfigure}
            \hfill
		\begin{subfigure}[t]{0.48\textwidth}
		  \begin{center}
            \includegraphics[width = \textwidth]{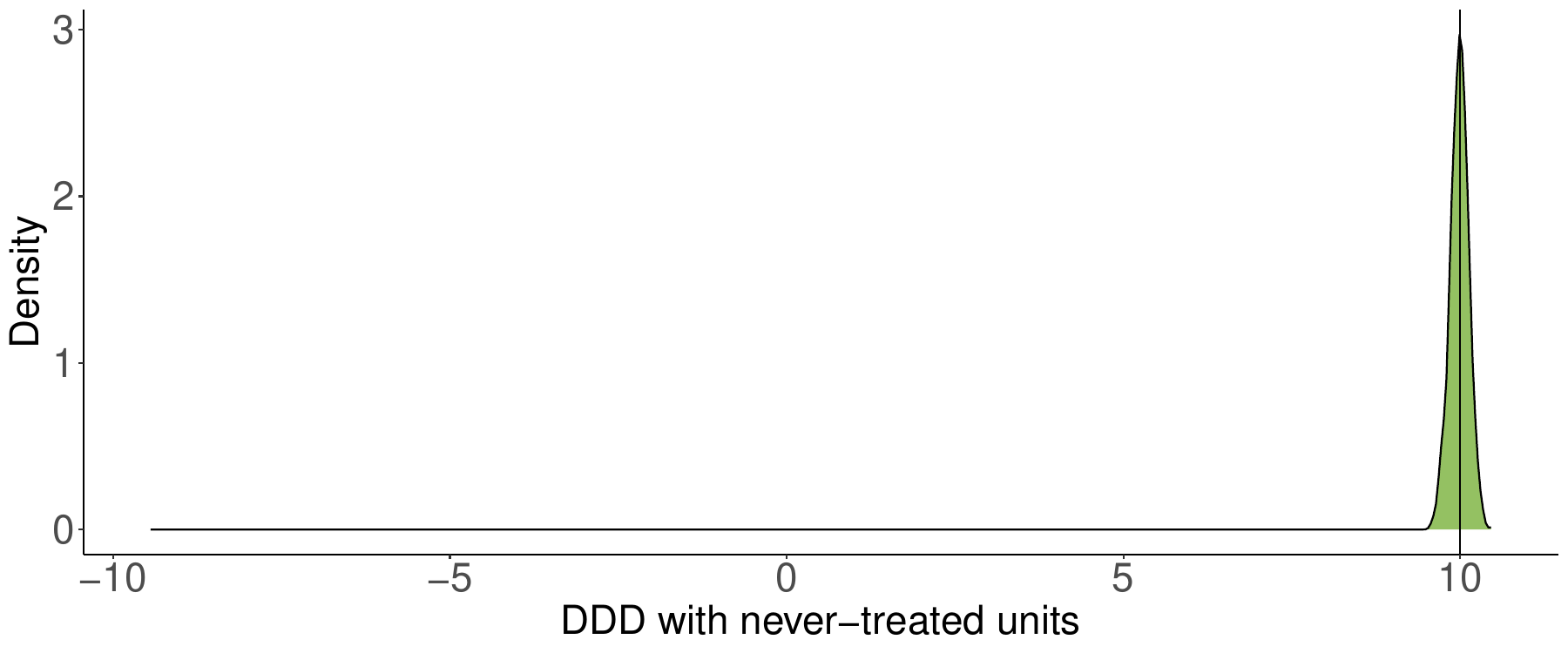}
            \caption{DDD using never-treated units}
            \label{fig:stag_ddd_never}
            \end{center}
		\end{subfigure}
            \begin{subfigure}[t]{0.48\textwidth}
		  \begin{center}
            \includegraphics[width = \textwidth]{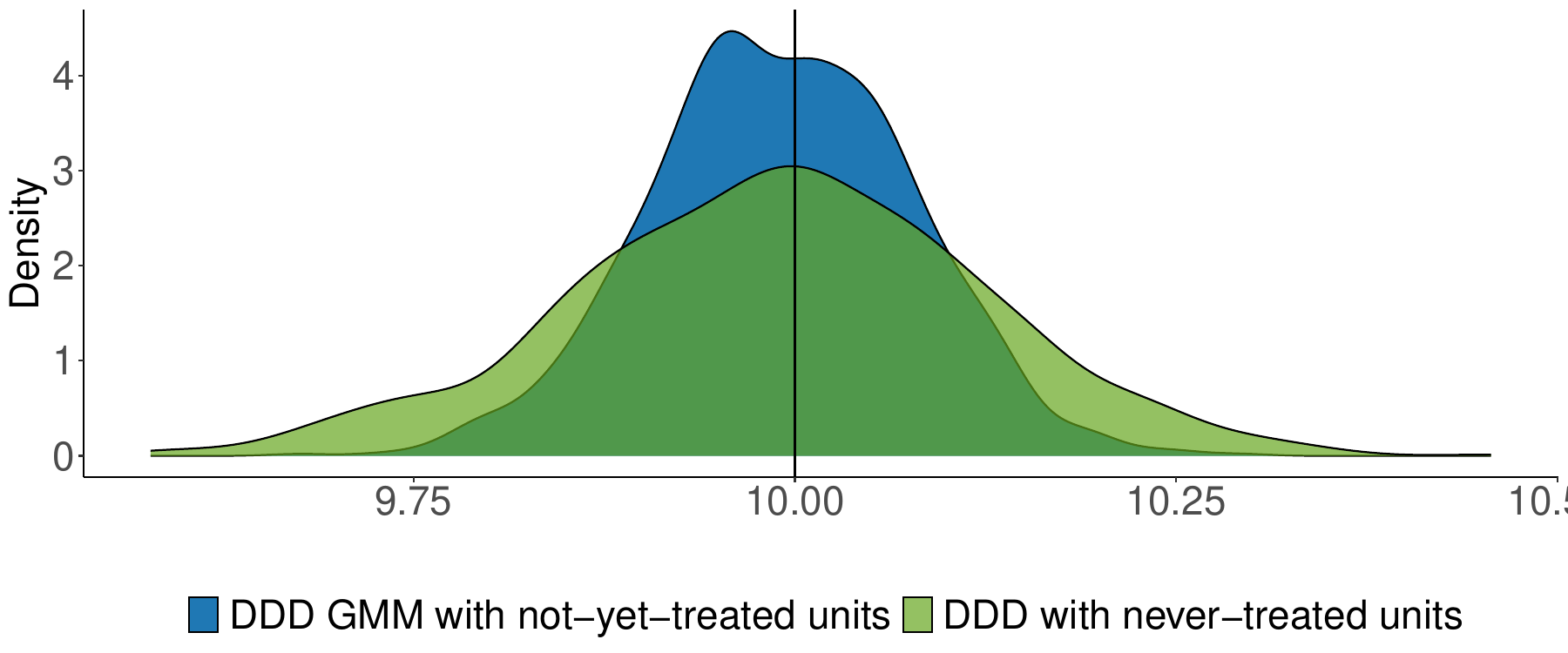}
            \caption{DDD using never-treated and GMM with not-yet-treated units}
            \label{fig:Sims_stagg_comp}
            \end{center}
		\end{subfigure}
		\caption{Density of different staggered DDD estimates for ATT(2,2), without covariates}
		\label{fig:Sims_stagg}
	\end{center}
        \justifying
	\vspace{-.2cm}\scriptsize{Notes: Simulation designs based on the design described in  Section \ref{sims_staggered} and Supplemental Appendix \ref{appendix:sims_staggered}, with $n=5,000$ and $1,000$ Monte Carlo repetitions. The true $ATT(2,2)$ is ten and is indicated in the solid vertical line in all panels. Panel (a) displays the density of DDD estimates that use the pooled not-yet-treated units as a comparison group as described in \eqref{eqn:att_gt_pooled}. Panel (b) displays the density of the estimates based on our proposed DDD GMM estimator that uses all not-yet-treated units as a comparison group described in \eqref{eqn:att_gt_gmm}. Panel (c) displays the density of the estimates based on our proposed DDD estimator that uses the never-treated units as a comparison group described in \eqref{eqn:att_gt_one-at-time} with $g_{\text{c}}=\infty$. Panel (d) compares DDD estimates using never-treated units (yellow curve) with GMM-based DDD using not-yet-treated (green curve), on the same scale. All densities are computed across all simulation draws. Panels(a)-(c) have the same x-axis range but different y-axis.
 }
\end{figure}

The key insight to address these problems is that we should be cautious when selecting the comparison group to estimate each $ATT(g,t)$. Such comparison groups must satisfy the DDD identification assumptions, which need to be verified on a group-by-group basis. Upon close inspection of Assumption \ref{ass:PT-NYT}, a natural solution is to avoid pooling across treatment-enabling groups and use one of them at a time. Doing so yields multiple valid comparisons for the same $(g,t)$ group, generating an over-identified model. More precisely, for each available not-yet-enabled group $g_{\text{c}}>t$, we can use the following estimator for $ATT(g,t)$, $t\ge g$:
\begin{small}
\begin{align}
  \widehat{ATT}_{g_{\text{c}}}(g,t)  =&~\Bigg[\bigg(\mathbb{E}_n\left[Y_{t} - Y_{g-1}|S = g, Q=1\right]\bigg)-\bigg( \mathbb{E}_n\left[ Y_{t} - Y_{g-1} | S = g, Q=0 \right] \bigg)\Bigg]  \nonumber\\
&- \Bigg[\bigg(\mathbb{E}_n\left[ Y_{t} - Y_{g-1}| S = g_{\text{c}}, Q=1 \right] \bigg)\Bigg.-\Bigg.\bigg(\mathbb{E}_n\left[ Y_{t} - Y_{g-1}|S=g_{\text{c}}, Q=0 \right]\bigg)\Bigg]. \label{eqn:att_gt_one-at-time}
    \end{align}   
\end{small}
Note that when $g_{\text{c}}=\infty$, \eqref{eqn:att_gt_one-at-time} uses the set of units that never enabled treatment $S=\infty$ as the comparison group. However, one is not restricted to this unique comparison group. In the context of our simulation, one can also use the units that enabled treatment in period three to learn about $ATT(2,2)$. In this sense, instead of choosing which comparison group to use, we propose combining all available options and forming a more precise DDD estimator for the $ATT(g,t)$s. More concretely, we propose using 
\begin{align}
  \widehat{ATT}_{\text{gmm}}(g,t)  =  \dfrac{\mathbf{1}' \widehat{\Omega}^{-1}}{\mathbf{1}' \widehat{\Omega}_{g,t}^{-1}\mathbf{1}}\widehat{ATT}_{\text{dr}}(g,t),\label{eqn:att_gt_gmm}
    \end{align}
where $\widehat{ATT}_{\text{dr}}(g,t)$ is the $k_{g,t}\times 1$-dimensional vector of all possible (non-collinear) estimators for $ATT(g,t)$ that uses a valid comparison group $g_{\text{c}}>t$, $\widehat{\Omega}_{g,t}$ is a consistent estimator of their variance-covariance matrix, and $\mathbf{1}$ is a ($k_{g,t}\times 1$-dimensional) vector of ones. We show that $ \widehat{ATT}_{\text{gmm}}(g,t)$ has a GMM interpretation based on re-centered influence functions in Remark \ref{rem:gmm}.

Panels (b) and (c) of Figure \ref{fig:Sims_stagg} display the density of the DDD estimates for $ATT(2,2)$ based on our GMM-based DDD estimator \eqref{eqn:att_gt_one-at-time} and our DDD estimator that only uses the never-treated as comparison group $g_{\text{c}}=\infty$ in \eqref{eqn:att_gt_one-at-time}, respectively. As it is easy to see, both estimators are correctly centered at the true $ATT(2,2)$. As Panels (a) - (c) in Figure \ref{fig:Sims_stagg} have the same scale, it is challenging to compare our DDD estimates that combine all not-yet-treated units with our DDD estimates that use never-treated units as comparison groups. In Figure \ref{fig:Sims_stagg}(d), we address this issue and display their densities based on the $1,000$ simulation draws. Overall, it is evident that utilizing all not-yet-treated units can yield substantial gains in precision. The results of our simulations indicate that confidence intervals based on $ \widehat{ATT}_{g_{\text{c}}=\infty}(g,t)$ are around 50\% wider than those based on $\widehat{ATT}_{\text{gmm}}(g,t)$, underscoring the appeal of using our GMM-based DDD estimator in terms of power.

Overall, the results in this section underscore the broader lesson that DDD designs with staggered adoption should not be treated as simple extensions of DiD methods. The interaction between timing, eligibility, and heterogeneity in group composition introduces complexities that necessitate more careful attention to the identification argument and the construction of the comparison group.

\section{The econometrics of DDD designs}\label{sec:ddd-theory}
In this section, we discuss the econometrics of DDD designs following the framework discussed in Section \ref{sec:framework}. We start by establishing nonparametric identification of the $ATT(g,t)$'s under the identification assumptions in Section \ref{sec:id_assumptions}. We then discuss estimation and inference procedures for $ATT(g,t)$'s and their event-study functional $ES(e)$ as defined in \eqref{eqn:ES}. Throughout this section, we focus on setups where covariates are important for identification, i.e., all the assumptions discussed in Section \ref{sec:id_assumptions} are only plausible after you condition on covariates. Results for unconditional DDD setups follow as special cases by taking all covariates $X=1$ for all units. We also focus on staggered treatment adoption DDD setups, as they nest DDD setups with a single treatment date.

\subsection{Identification}\label{sec:identification}

In this section, we establish the nonparametric identification for the $ATT(g,t)$'s in all post-treatment periods $t\ge g$ under Assumptions \ref{ass:sampling_panel}, \ref{ass:overlap_staggered}, \ref{ass:anticipation}, and \ref{ass:PT-NYT}. Furthermore, we demonstrate that one can utilize regression adjustment/outcome regression (RA), inverse probability weighting (IPW), or doubly robust estimands to recover the $ATT(g,t)$'s. We also demonstrate that one can potentially utilize different comparison groups, thereby opening the door to combining them for potential efficiency gains. 

Before formalizing our results, we need to introduce some additional notation. Let $m_{Y_t-Y_{t'}}^{S=g,Q=q}(X) \equiv \mathbb{E}\left[Y_t - Y_{t'}|S=g,Q=q ,X\right]$ denote the population regression function of changes in outcomes from period $t'$ to period $t$ given covariates $X$ among units that enabled treatment in period $g$ ($S=g$) that belongs to eligibility group $q$ ($Q=q$). Analogously, let $p^{S=g,Q=1}_{g', q'}(X) \equiv \mathbb{P}[S=g, Q=1 | X, (S=g, Q=1) \cup (S=g',Q=q')]$ denote the generalized propensity score. Note that $p^{S=g,Q=1}_{g', q'}(X)$ indicates the probability of a unit being observed in enabling group $S=g$ and being eligible for treatment ($Q=1)$, conditional on pre-treatment covariates $X$ and on either being in the $S=g$ group and being eligible for treatment, or being in the $S=g'$ group with eligibility to treatment $Q=q'$.\footnote{We use this notion of generalize propensity score as it allow us to focus on sequences of two-groups comparisons as in \citet{Lechner2002} and \citet{Callaway2021}. One can understand these generalized propensity scores as $p^{S=g,Q=1}_{g’, q’}(X) = \frac{p^{S=g,Q=1}(X)}{p^{S=g,Q=1}(X) + p^{S=g’,Q=q’}(X)}$ with $p^{S=g,Q=q}(X) \equiv \mathbb{P}[S=g, Q=q | X]$. We favor $p^{S=g,Q=1}_{g’, q’}(X)$ as this is how we implement these when constructing our DDD estimators.} 

For any $g_{\text{c}}\in \mathcal{S}$ such that $g_{\text{c}}>\max\{g,t\}$, and any post-treatment period $t\ge g$, let the doubly robust DDD estimand for the $ATT(g,t)$ be given by
\begin{small}
\begin{align}\label{eqn:DRDDD_stagg}
{ATT}_{\text{dr}, g_{\text{c}}}(g,t) =&~ \mathbb{E}\left[ \left({w}^{S=g,Q=1}_{\text{trt}}(S,Q) - {w}^{S=g,Q=1}_{g,0}(S,Q,X) \right)\left(Y_t- Y_{g-1} - {m}_{Y_t-Y_{g-1}}^{S=g,Q=0}(X) \right)\right]\nonumber \\
& +\mathbb{E}\left[ \left({w}^{S=g,Q=1}_{\text{trt}}(S,Q) - {w}^{S=g,Q=1}_{g_{\text{c}},1}(S,Q,X) \right)\left(Y_t- Y_{g-1} -  {m}_{Y_t-Y_{g-1}}^{S=g_{\text{c}},Q=1}(X)  \right)\right] \\
&- \mathbb{E}\left[ \left({w}^{S=g,Q=1}_{\text{trt}}(S,Q) -{w}^{S=g,Q=1}_{g_{\text{c}},0}(S,Q,X)\right) \left(Y_t- Y_{g-1} - {m}_{Y_t-Y_{g-1}}^{S=g_{\text{c}},Q=0}(X)\right)\right] \nonumber,
\end{align}
\end{small}
where the weights ${w}$ are given by
{\fontsize{10.4pt}{10pt}\selectfont
\vspace{-.5cm}
\begin{align}
    {w}^{S=g,Q=1}_{\text{trt}}(S, Q)\equiv\dfrac{1{\{S=g,Q=1\}} }{\mathbb{E}[1{\{S=g,Q=1\}}]}, \quad {w}^{S=g, Q=1}_{g',q'}(S, Q, X) \equiv \dfrac{\dfrac{1{\{S=g',Q=q'\}} \cdot {p}^{S=g,Q=1}_{g',q'}(X) }{1 - {p}^{S=g,Q=1}_{g',q'}(X)}} {\mathbb{E}\bracks{\dfrac{1{\{S=g',Q=q'\}} \cdot {p}^{S=g,Q=1}_{g',q'}(X) }{1 - {p}^{S=g,Q=q}_{g',q'}(X)}}}. \label{eqn:weights}
\end{align}
}
Analogously, let the RA DDD estimand for the $ATT(g,t)$ be given by 
\begin{small}
\begin{equation}\label{eqn:regDDD_stagg}
\resizebox{.93\textwidth}{!}{$
\begin{aligned}
{ATT}_{\text{ra}, g_{\text{c}}}(g,t) =\ \mathbb{E}\Big[{w}^{S=g,Q=1}_{\text{trt}}(S, Q)\Big(Y_t - Y_{g-1} - {m}_{Y_t-Y_{g-1}}^{S=g,Q=0}(X) - {m}_{Y_t-Y_{g-1}}^{S=g_{\text{c}},Q=1}(X) + {m}_{Y_t-Y_{g-1}}^{S=g_{\text{c}},Q=0}(X) \Big)\Big]
\end{aligned}
$}
\end{equation}
\end{small}
and the IPW estimand be
\begin{small}
\begin{align}\label{eqn:ipwDDD_stagg}
{ATT}_{\text{ipw}, g_{\text{c}}}(g,t) =& ~\mathbb{E}\left[\left({w}^{S=g,Q=1}_{\text{trt}}(S,Q) - {w}^{S=g,Q=1}_{g,0}(S,Q,X) \right) \left( Y_t-Y_{g-1}\right)\right]\nonumber\\ 
& - \mathbb{E}\left[\left({w}^{S=g,Q=1}_{g_{\text{c}},1}(S,Q,X) - {w}^{S=g,Q=1}_{g_{\text{c}},0}(S,Q,X)\right) \left( Y_t-Y_{g-1}\right)\right].
\end{align}
\end{small}

\begin{theorem}\label{thm:Identification}
    Let Assumptions \ref{ass:sampling_panel}, \ref{ass:overlap_staggered}, \ref{ass:anticipation}, and \ref{ass:PT-NYT} hold. Then, for all $g\in \mathcal{G}_{\text{trt}}$, $t\in \{2,\dots, T\}$, and $g_{\text{c}} \in \mathcal{S}$ such that $t\ge g$ and $g_{\text{c}}>t$,
    \begin{equation}
        ATT(g,t) = ATT_{\text{dr}, g_{\text{c}}}(g,t) = {ATT}_{\text{ra}, g_{\text{c}}}(g,t)={ATT}_{\text{ipw}, g_{\text{c}}}(g,t).
    \end{equation}
\end{theorem}

Theorem \ref{thm:Identification} is the first main result of this paper. It establishes the nonparametric identification of all post-treatment $ATT(g,t)$'s in DDD setups. It extends the DiD identification results of \citet{Callaway2021} to DDD setups. As such, it also extends the difference-in-differences identification results based on the RA approach of \citet{Heckman1997}, the IPW approach of \citet{Abadie2005}, and the DR approach of \citet{SantAnna2020} to DDD setups with multiple periods and variation in treatment time. Theorem \ref{thm:Identification} also highlights that one can use different parts of the data-generating process to identify the $ATT(g,t)$'s: the RA estimand only models the conditional expectation of evolution of outcomes among untreated units, the IPW approach only models the conditional probability of being observed in a given partition of the $S$-by-$Q$ groups, whereas the DR approach exploits both components. A big advantage of the DR approach is that it is based on a Neyman-orthogonal moment condition \citep{Belloni2017}, and, therefore, it is more robust against model misspecifications than the IPW and RA formulations. It is very easy to show that estimators based on ${ATT}_{\text{dr}, g_{\text{c}}}(g,t)$ enjoy a very attractive doubly-robust property \citep{SantAnna2020} that allows for some forms of (global) model misspecifications.\footnote{For an overview of doubly robust estimators in cross-sectional designs, see section 2 of \cite{Sloczynski2018}, and \cite{Seaman2018}.}

Another important result from Theorem \ref{thm:Identification} is that our DDD model is over-identified, as we can use multiple not-yet-treated enabling groups $g_{\text{c}}$ as valid comparison groups. For instance, in a setup with $S\in \{2,3,\infty\}$, we can set $\gc=3$ or $\gc=\infty$ to identify $ATT(2,2)$, and both will lead to the same target parameter. As a direct consequence of this result, any weighted sum of these estimands that use different $\gc$'s will also lead to the $ATT(g,t)$, as long as the weights sum up to one. We formalize this result in the following corollary, using the DR estimand; however, this also applies to the RA and IPW. Let $\mathcal{G}_{\text{c}}^{\text{g,t}}=\{ \gc \in \mathcal{S}: \gc > \max\{g,t\}\}$. 
\begin{corollary}\label{cor:over-id}
    Let Assumptions \ref{ass:sampling_panel}, \ref{ass:overlap_staggered}, \ref{ass:anticipation}, and \ref{ass:PT-NYT} hold. Then, for all $g\in \mathcal{G}_{\text{trt}}$ and $t\in \{2,\dots, T\}$ such that $t\ge g$, and any set of weights $w^{\text{g,t}}_{\\gc}$ that sum up to one over $\mathcal{G}_{\text{c}}^{\text{g,t}}$, 
    \begin{equation*}
        ATT(g,t) = \sum_{\gc\in \mathcal{G}_{\text{c}}^{\text{g,t}}} w^{\text{g,t}}_{\gc}~ ATT_{\text{dr}, g_{\text{c}}}(g,t).
    \end{equation*}
\end{corollary}

As Corollary \ref{cor:over-id} indicates that all weighted sums lead to the same $ATT(g,t)$, a natural way to choose these weights is to pick them such that we maximize precision in terms of minimizing the resulting asymptotic variance. In the next session, we will discuss this in greater detail, connecting these arguments to a formulation based on generalized methods of moments using re-centered influence functions. 

\begin{remark}\label{rem:no_difference_of_2_DiDs}
    As we discussed in Section \ref{sec:implication}, in two-period DDD setups without covariates, one can identify $ATT(2,2)$ using the difference of two DiD estimands as in \eqref{eqn:3wtfe_att22} \citep{Olden2022}. This equivalence breaks down when the DDD identification assumptions are only satisfied after you condition on covariates $X$---see Figure \ref{fig:Sims_2_periods}. The econometric reason for this failure of equivalence is that one needs to integrate the covariates using the covariate distribution among treated units, i.e., units with $S=2$ and $Q=1$. Proceeding as if $ATT(2,2)$ were the difference of two DiD estimands would integrate $X$ using the covariate distribution of untreated units ($S=\infty$ and $Q=1$), leading to biases. The results in Theorem \ref{thm:Identification} address this problem by guaranteeing that one integrates out covariates using the correct reference distribution, which leads to a combination of three DiD estimands, not just two.
\end{remark}

\begin{remark}\label{rem:no_pooling_nyt}
    Although Theorem \ref{thm:Identification} and Corollary \ref{cor:over-id} allow one to use several different not-yet-treated cohorts $\gc$ as the comparison group, it does not allow one to pool all not-yet-treated units and use that pooled set of units as the aggregate comparison group to identify $ATT(g,t)$ in DDD. This sharply contrasts DiD procedures such as those discussed in \citet{Callaway2021}---see Figure \ref{fig:stag_ddd_pooled_nyt} for an illustration of the bias that can arise by following this type of procedure. The econometric reasoning for such results is that Assumption \ref{ass:PT-NYT} allows for both enabling-group- and eligibility-group-specific trends, and it does not impose that the proportion of units in each eligibility group $Q$ is the same across all enabling groups $S$. As such, Assumption \ref{ass:PT-NYT} does not guarantee that, with probability one, 
    \begin{small}
    \vspace{-.5cm}
    \begin{eqnarray*}
\mathbb{E}\left[Y_{t}(\infty) - Y_{t-1}(\infty)| S = g, Q=1, X\right] &-& \mathbb{E}\left[Y_{t}(\infty)- Y_{t-1}(\infty) | S = g, Q=0, X \right]  \\[-0.2cm]
&=&\\[-0.2cm]
\vspace{-5cm}\mathbb{E}\left[Y_{t}(\infty) - Y_{t-1}(\infty)|S >t, Q=1, X\right] &-& \mathbb{E}\left[Y_{t}(\infty)- Y_{t-1}(\infty) | S >t, Q=0, X \right],
\end{eqnarray*}
 \end{small}
 as it would be required to use the pooled, not-yet-treated units as a comparison group. 
 \end{remark}

\begin{remark}\label{rem:event-study}
    As Theorem \ref{thm:Identification} establishes nonparametric identification of the $ATT(g,t)$'s over all post-treatment periods and that $\P(G=g|G+e \in [1,T])$ is also nonparametrically identified, it follows that event-study parameters that aggregate across eligibility-groups, $ES(e)$ as defined in \eqref{eqn:ES}, is also nonparametrically identified. For instance, it follows that for any event-time $e\ge 0$,
    \begin{align}
ES(e)  = \sum_{g\in\mathcal{G}_{\text{trt}}} \P(G=g|G+e \in [1,T]) ATT_{\text{dr}, g_{\text{c}}}(g,g+e).
\end{align}
One can also replace $ATT_{\text{dr}, g_{\text{c}}}(g,g+e)$ with their analogs in Corollary \ref{cor:over-id} or with the RA or IPW estimands in Theorem \ref{thm:Identification}. One can also use Theorem \ref{thm:Identification} to establish the identification of many other aggregate summary causal parameters discussed in Section 3 of \citet{Callaway2021}.
\end{remark}

\begin{remark}\label{rem:pre-treatment-periods}
One of the biggest appeals of DiD and DDD setups is the availability of pre-treatment periods that allow the assessment of the plausibility of PT assumptions, such as Assumption \ref{ass:PT-NYT}. Under Assumption \ref{ass:anticipation}, a very popular way to assess the plausibility of PT is to construct event-study plots based on $ES(e)$ as in \eqref{eqn:ES}, consider both pre-treatment ($e<0$) and post-treatment ($e\ge 0$) event times, and check whether pre-treatment event-study coefficients are all close to zero. It is straightforward to adapt this strategy in our DDD context by fixing the statistical estimand---for example, the ${ATT}_{\text{dr}, g_{\text{c}}}(g,t)$ in \eqref{eqn:DRDDD_stagg}---consider pre-treatment periods $t<g$, and then aggregate them using cohort-size. More specifically, for any event-time $e < 0$,
    \begin{align}
ES(e)  = \sum_{g\in\mathcal{G}_{\text{trt}}} \P(G=g|G+e \in [1,T]) ATT_{\text{dr}, g_{\text{c}}}(g,g+e).
\end{align} 
Note that when $e=-1$, $ES(e)=0$ by construction, as we fix the baseline period at the last untreated period for group $g$, $g-1$. Based on these event-study aggregations, it is also possible to conduct sensitivity analysis for the plausibility of Assumption \ref{ass:PT-NYT} using the results in \citet{rambachan_roth}.
\end{remark}

\begin{remark}
    The DR DDD estimand can also be understood as a DDD estimand that builds on an efficient influence function in DDD setups with two periods. See Lemma \ref{lemma:eif} in the appendix for such results.
\end{remark}

\subsection{Estimation and inference}
In this section, we now propose simple-to-use plug-in estimators for the $ATT(g,t)$s and $ES(e)$s parameters, and discuss how one can conduct valid inference for these parameters. We focus on the doubly robust DDD estimator; the results for the RA and IPW DDD estimators are analogous. 

First, notice that for any $g_{\text{c}}\in \mathcal{S}$ such that $g_{\text{c}}>\max\{g,t\}$, Theorem \ref{thm:Identification} suggests that we can estimate $ATT(g,t)$ by using the sample analogue of \eqref{eqn:DRDDD_stagg},
\begin{small}
\begin{align}\label{eqn:DRDDD_stagg_estimated}
\widehat{{ATT}}_{\text{dr}, g_{\text{c}}}(g,t) =&~ \mathbb{E}_n\left[ \left(\widehat{w}^{S=g,Q=1}_{\text{trt}}(S,Q) - \widehat{w}^{S=g,Q=1}_{g,0}(S,Q,X) \right)\left(Y_t- Y_{g-1} - \widehat{m}_{Y_t-Y_{g-1}}^{S=g,Q=0}(X) \right)\right]\nonumber \\
& +\mathbb{E}_n\left[ \left(\widehat{w}^{S=g,Q=1}_{\text{trt}}(S,Q) - \widehat{w}^{S=g,Q=1}_{g_{\text{c}},1}(S,Q,X) \right)\left(Y_t- Y_{g-1} -  \widehat{m}_{Y_t-Y_{g-1}}^{S=g_{\text{c}},Q=1}(X)  \right)\right] \\
&- \mathbb{E}_n\left[ \left(\widehat{w}^{S=g,Q=1}_{\text{trt}}(S,Q) -\widehat{w}^{S=g,Q=1}_{g_{\text{c}},0}(S,Q,X)\right) \left(Y_t- Y_{g-1} - \widehat{m}_{Y_t-Y_{g-1}}^{S=g_{\text{c}},Q=0}(X)\right)\right] \nonumber,
\end{align}
\end{small}
where the estimated weights $\widehat{w}$ are given by
\begin{small}
\vspace{-.5cm}
\begin{align*}
    \widehat{w}^{S=g,Q=1}_{\text{trt}}(S, Q)\equiv\dfrac{1{\{S=g,Q=1\}} }{\mathbb{E}_n[1{\{S=g,Q=1\}}]}, \quad \widehat{w}^{S=g, Q=1}_{g',q'}(S, Q, X) \equiv \dfrac{\dfrac{1{\{S=g',Q=q'\}} \cdot \widehat{p}^{S=g,Q=1}_{g',q'}(X) }{1 - \widehat{p}^{S=g,Q=1}_{g',q'}(X)}} {\mathbb{E}_n\bracks{\dfrac{1{\{S=g',Q=q'\}} \cdot \widehat{p}^{S=g,Q=1}_{g',q'}(X) }{1 - \widehat{p}^{S=g,Q=q}_{g',q'}(X)}}},
\end{align*}
\end{small}
and $\widehat{m}_{Y_t-Y_{g-1}}^{S=g_{\text{c}}, Q=1}(X)$ and $\widehat{p}^{S=g, Q=1}_{g',q'}(X)$ are (potentially misspecified) working models for the outcome regression ${m}_{Y_t-Y_{g-1}}^{S=g_{\text{c}}, Q=1}(X)$ and the generalized propensity score $\widehat{p}^{S=g, Q=1}_{g',q'}(X)$. These estimators extend the DR DiD estimator of \citet{Callaway2021} to the DDD setup, and remain consistent if \emph{either} outcome regression or generalized propensity score models are correctly specified. It is also worth stressing that we do not need that \emph{all} generalized propensity score working models or \emph{all} outcome regression working models in \eqref{eqn:DRDDD_stagg_estimated} to be correctly specified to get a consistent DDD estimator for $ATT(g,t)$; it suffices that any of the working models within each of the 3 DR DiD components of \eqref{eqn:DRDDD_stagg_estimated} to be correctly specified, allowing a greater deal of estimation flexibility.\footnote{Some people may call this a multiply-robust estimator, as one has more than two opportunities to estimate the target parameter consistently. For simplicity, we retain the doubly robust terminology to avoid new acronyms.  }

As Corollary \ref{cor:over-id} highlights, one can also combine several $\widehat{{ATT}}_{\text{dr}, g_{\text{c}}}(g,t)$ that leverage different comparison groups $\gc$, i.e., for any (consistently estimated) weights $\widehat{w}^{\text{g,t}}_{\gc}$ that sum up to one over $ \mathcal{G}_{\text{c}}$, 
\begin{equation}
        \widehat{ATT}_{\text{dr},\widehat{w}}(g,t) = \sum_{\gc\in \mathcal{G}_{\text{c}}^{\text{g,t}}} \widehat{w}^{\text{g,t}}_{\gc}~ \widehat{ATT}_{\text{dr}, g_{\text{c}}}(g,t) = \widehat{w}^{\text{g,t}~'}\widehat{ATT}_{\text{dr}}(g,t),\label{eqn:ATT_gt_w_estimated}
\end{equation}
where $\widehat{ATT}_{\text{dr}}(g,t)$ is the $k_{g,t}\times 1$ vector of $\widehat{{ATT}}_{\text{dr}, g_{\text{c}}}(g,t)$ for all $\gc \in \mathcal{G}_{\text{c}}^{\text{g,t}}$, and $\widehat{w}^{\text{g,t}~'}$ is a $k_{g,t}\times 1$ vector of (estimated) weights that sum up to one, i.e., for a generic vector of ones $\textbf{1}$, $\textbf{1}'w^{\text{g,t}}=1$.

A natural question that arises is: how should one choose these weights $\widehat{w}^{\text{g,t}}_{\gc}$? We propose to choose the weights that lead to the asymptotically most precise (minimum variance) estimator for $ATT(g,t)$, that is, to pick weights that solve
\begin{align}
   \min_{w^{\text{g,t}}} w^{\text{g,t}~'}~\widehat{\Omega}_{g,t}~ w^{\text{g,t}} \text{~~~~~subject to } \textbf{1}'w^{\text{g,t}} = 1, \label{eqn:minimization}
\end{align}
where $\widehat{\Omega}_{g,t}$ is a $k_{g,t}\times k_{g,t}$ consistent estimator for the variance-covariance matrix of $\widehat{ATT}_{\text{dr}}(g,t)$. Notice that the solution of \eqref{eqn:minimization} admits a closed-form solution, and the optimal weights are given by
\begin{align}
  \widehat{w}^{\text{g,t}}_{\text{gmm}} =\dfrac{\widehat{\Omega}^{-1}_{g,t}\textbf{1}}{\textbf{1}'\widehat{\Omega}^{-1}_{g,t}\textbf{1}} \label{eqn:opt_weights}.
\end{align}
In turn, this implies that the linear combination of $\widehat{{ATT}}_{\text{dr}, g_{\text{c}}}(g,t)$ that leads to the most precise estimator for $ATT(g,t)$ is given by
\begin{equation}
        \widehat{ATT}_{\text{dr},\text{gmm}}(g,t) = \dfrac{\mathbf{1}' \widehat{\Omega}_{g,t}^{-1}}{\mathbf{1}' \widehat{\Omega}_{g,t}^{-1}\mathbf{1}} \widehat{ATT}_{\text{dr}}(g,t). \label{eqn:DRDDD_optimal}
\end{equation}

In many situations with multiple periods and variation in treatment time, researchers are interested in summarizing the $ATT(g,t)$'s into fewer parameters that highlight treatment effect heterogeneity with respect to the time elapsed since treatment take-up. That is, very often, researchers are interested in estimating event-study type parameters $ES(e)$ as defined in \eqref{eqn:ES}. A very natural estimator for $ES(e)$ is the plug-in estimator, where we replace $ATT(g,t)$ with $\widehat{ATT}_{\text{dr},\text{gmm}}(g,t)$ (or $\widehat{ATT}_{\text{dr},\gc}(g,t)$), and $\P(G=g|G+e \in [1,T])$ by its sample analogue, that is,
\begin{equation}
\label{eqn:es_estimator}
    \widehat{ES}_{\text{dr},\text{gmm}}(e) = \sum_{g\in\mathcal{G}_{\text{trt}}} \P_n(G=g|G+e \in [1,T]) \widehat{ATT}_{\text{dr},\text{gmm}}(g,g+e),
\end{equation}
where $\P_n(G=g|G+e \in [1,T]) = \sum_{i=1}^n 1\{G_i = g\} 1\{G_i + e \in [1,T]\}\big/ \sum_{j=1}^n 1\{G_j + e \in [1,T]\}$. We can define $ \widehat{ES}_{\text{dr},\gc}(e)$ analogously by replacing $\widehat{ATT}_{\text{dr},\text{gmm}}(g,g+e)$ with $\widehat{ATT}_{\text{dr},\gc}(g,g+e)$ on \eqref{eqn:es_estimator}. Based on it, we can also estimate an overall summary parameter by averaging all post-treatment event times, i.e., 
\begin{align}
\widehat{ES}_{\text{avg},\text{gmm}} = \dfrac{1}{N_E}\sum_{e\in \mathcal{E}} \widehat{ES}_{\text{dr},\text{gmm}}(e)\label{eqn:overall_ATT_estimate}.
\end{align}

\begin{remark}\label{rem:gmm}
   It is also worth noticing that $\widehat{ATT}_{\text{dr},\text{gmm}}(g,t)$ in \eqref{eqn:DRDDD_optimal} can be interpreted as an optimal Generalized Method of Moments (GMM) estimator based on re-centered influence functions. To see this, let ${\mathbb{IF}}_{\text{dr}, g_{\text{c}}}(g,t)$ denote the influence function of $\sqrt{n}\left(\widehat{ATT}_{\text{dr}, g_{\text{c}}}(g,t) - ATT_{\text{dr}, g_{\text{c}}}(g,t)\right)$. Let ${\mathbb{RIF}}_{\text{dr}, g_{\text{c}}}(g,t) = {\mathbb{IF}}_{\text{dr}, g_{\text{c}}}(g,t) + {ATT}_{\text{dr},\gc}(g,t)$ denote its re-centered influence function, and denote the $k_{g,t}\times 1$ vector of all ${\mathbb{RIF}}_{\text{dr}, g_{\text{c}}}(g,t)$ for $\gc \in \mathcal{G}_{\text{c}}^{\text{g,t}}$ by ${\mathbb{RIF}}_{\text{dr}}(g,t)$. Since influence functions are mean zero, and that $ATT(g,t) = ATT_{\text{dr}, g_{\text{c}}}(g,t)$ for any $\gc \in \mathcal{G}_{\text{c}}^{\text{g,t}}$, we have the vector of moment conditions $\E[{\mathbb{RIF}}_{\text{dr}}(g,t) - \theta^{g,t}] =  0$, with $\theta^{g,t} = ATT(g,t)$. From standard GMM results \citep{Newey_McFadden_1994_Handbook}, it follows that, under mild regularity conditions, the optimal (population) GMM estimator for $\theta^{g,t}$ is given by $$\theta^{g,t}_{\text{gmm}} = \dfrac{\mathbf{1}' {\Omega}_{g,t}^{-1}}{\mathbf{1}' {\Omega}_{g,t}^{-1}\mathbf{1}} \E[ {\mathbb{RIF}}_{\text{dr}}(g,t)] = \dfrac{\mathbf{1}' {\Omega}_{g,t}^{-1}}{\mathbf{1}' {\Omega}_{g,t}^{-1}\mathbf{1}} ATT_{\text{dr}}(g,t),$$ where the last equality follows from $\E[ {\mathbb{IF}}_{\text{dr}}(g,t)] = 0$. Thus, $\widehat{ATT}_{\text{dr},\text{gmm}}(g,t)$ in \eqref{eqn:DRDDD_optimal} is the sample-analogy of the efficient population GMM $\theta^{g,t}_{\text{gmm}}$. 
\end{remark}

\subsubsection{Asymptotic theory for ATT(g,t)'s}

In what follows, we derive the large sample properties of our DR DDD estimators $\widehat{ATT}_{\text{dr},\gc}(g,t)$ and $\widehat{ATT}_{\text{dr},\text{gmm}}(g,t)$. All our results are derived for the large $n$, fixed $T$ paradigm. For a generic $Z$, let $|| Z ||  = \sqrt{trace(Z'Z)}$ denote the Euclidean norm of $Z$ and set $W_i=(Y_{i,t=1},\dots,Y_{i,t=T}, X_i', G_i, S_i, Q_i)'$; we will omit the index $i$ to unclutter the notation. Let $g(\cdot)$ be a generic notation for the outcome regressions $m_{Y_t-Y_{t'}}^{S=g',Q=q}(X)$ and generalized propensity scores $p^{S=g,Q=1}_{g', q'}(X)$, and, with some abuse of notation, let $g(\cdot; \gamma)$ denote a parametric model for $g(\cdot)$ that is known up to the finite-dimensional parameters $\gamma$. For a generic $\kappa^{g,t}_{\gc} = (\gamma^{ps~\prime}_{g,t,\gc}, \gamma^{reg~\prime}_{ g,t,\gc})'$, with $\gamma^{ps}_{g,t,\gc}$ and $\gamma^{reg}_{ g,t,\gc}$ being nuisance parameters for the generalized propensity score and outcome regressions, respectively, let
\begin{small}
\begin{align*}
h^{g,t}_{\gc}(W; \kappa^{g,t}_{\gc}) =&~  \left({w}^{S=g,Q=1}_{\text{trt}}(W) - {w}^{S=g,Q=1}_{g,0}(W;\gamma^{ps}_{g,t,\gc}) \right)\left(Y_t- Y_{g-1} - {m}_{Y_t-Y_{g-1}}^{S=g,Q=0}(X;\gamma^{reg}_{ g,t,\gc}) \right)\nonumber \\
& +\left({w}^{S=g,Q=1}_{\text{trt}}(W) - {w}^{S=g,Q=1}_{g_{\text{c}},1}(W;\gamma^{ps}_{g,t,\gc}) \right)\left(Y_t- Y_{g-1} -  {m}_{Y_t-Y_{g-1}}^{S=g_{\text{c}},Q=1}(X;\gamma^{reg}_{ g,t,\gc})  \right) \\
&-  \left({w}^{S=g,Q=1}_{\text{trt}}(W) -{w}^{S=g,Q=1}_{g_{\text{c}},0}(W;\gamma^{ps}_{g,t,\gc})\right) \left(Y_t- Y_{g-1} - {m}_{Y_t-Y_{g-1}}^{S=g_{\text{c}},Q=0}(X;\gamma^{reg}_{ g,t,\gc})\right) \nonumber,
\end{align*}
\end{small}
where the weights $w(W;\gamma^{ps}_{g,t,\gc})$ are defined similarly to those in \eqref{eqn:weights}, with the difference being that the true unknown generalized propensity score models are replaced by working parametric counterparts, $p^{S=g,Q=1}_{g', q'}(X;\gamma^{ps}_{g,t,\gc})$, and the true unknown outcome regression models ${m}_{Y_t-Y_{g-1}}^{S=g',Q=q}(X)$ are also replaced with parametric working models, ${m}_{Y_t-Y_{g-1}}^{S=g',Q=q}(X;\gamma^{reg}_{ g,t,\gc})$. We denote the vector of pseudo-true parameters by $\kappa^{g,t}_{0, \gc}$ and let $\dot{h}^{g,t}_{\gc}(\kappa) = \partial h^{g,t}_{\gc}(W;\kappa) / \partial \kappa$.

To derive our results, we make the following relatively mild assumptions.

\begin{namedassumption}{WM}[Working Model Conditions]\label{ass:wm}
\label{ass:smoothness}
        (i) $g(x;\gamma)$ is a parametric model for $g(x)$, where $ \gamma \in \Theta \subset \mathbb{R}^{d_k}$ is a compact set; (ii) the mapping $\theta \mapsto g(X ; \theta)$ is a.s.~continuous; (iii) the pseudo-true parameter $\theta_0 \in \operatorname{int}(\Theta)$ satisfies that for an appropriate criterion function $Q: \Theta \rightarrow  \R$ and for any $\epsilon > 0$, there exists some $\delta >0 $ such that $\inf _{\theta \in \Theta:\left\|\theta-\theta_0\right\| \geq \epsilon} Q(\theta)-Q\left(\theta_0\right)>\delta$; (iv) there exists an open neighborhood $\Theta_0 \subset \Theta$ containing $\theta_0$ such that $g(X;\gamma)$ is a.s.~continuously differentiable in a neighborhood of $\gamma_0 \in \Theta_0$. In addition, (v) there exists some $\epsilon > 0$ such that, for all $(g,g',q') \in \mathcal{G}_{\text{trt}}\times \mathcal{G}^{g,t}_c \times \{0,1\}$, we have that  
        $0 \leq p^{S=g,Q=1}_{g',q'}(X;\theta) \leq 1 - \epsilon$ a.s. for all $\theta \in \operatorname{int}(\Theta_{ps})$, where $\Theta_{ps}$ denotes the parameter space of $\gamma$ for the generalized propensity score working model.
\end{namedassumption}

\begin{namedassumption}{ALR}[$\sqrt{n}$-Asymptotically Linear Representation]
\label{ass:ALR}
Let $\hat{\theta}$ be a strongly consistent estimator of $\theta_0 \mapsto g(x ; \theta_0)$ and satisfy the following linear expansion
\begin{align}
    \sqrt{n}\left(\widehat{\theta}-\theta_0\right)=\frac{1}{\sqrt{n}} \sum_{i=1}^n l \left(W_i ; \theta_0\right)+o_p(1)
\end{align}
where $l \left(\cdot ; \cdot\right)$ is a function such that $\E\left[l \left(W_i ; \theta_0\right)\right] = 0$; $\E\left[l \left(W_i ; \theta_0\right) \cdot l \left(W_i ; \theta_0\right)^{'}\right] < \infty$ and is positive definite; and $\lim_{\delta \to 0}\E\left[\sup_{\theta \in \Theta_0: \left\|\theta-\theta_0\right\| \leq \delta} \left\| l(W; \theta) -  l(W; \theta_0) \right\|^2 \right] = 0$.
\end{namedassumption}

\begin{namedassumption}{IC}[Integrability Conditions]
\label{ass:integrability}
  For each $g \in \mathcal{G}_{\text{trt}}$, $t \in \{2,\dots, T\}$, and $g'\in \mathcal{G}^{g,t}_c$, assume that $\E[\| h^{g,t}_{\gc}(W; \kappa^{g,t}_{0, \gc}) \|^{2}] < \infty$ and $\E\left[\sup_{\kappa \in \Gamma_0} \left| \dot{h}^{g,t}_{\gc}(\kappa) \right|\right] < \infty$, where $\Gamma_0$ is a small neighborhood of the pseudo-true parameter $\kappa^{g,t}_{0,\gc}$.
\end{namedassumption}

Assumptions \ref{ass:smoothness}, \ref{ass:ALR}, and \ref{ass:integrability} are standard in the literature; see e.g., \cite{Abadie2005, WOOLDRIDGE20071281, SantAnna2020, Callaway2021}. Assumptions \ref{ass:smoothness} and \ref{ass:ALR} impose a well-behaved parametric structure for the first-step estimators for the nuisance parameters. This assumption is made for statistical convenience and acknowledges that, in many DDD applications, the number of units in each group is small, making it difficult to adopt a nonparametric approach reliably. It is relatively straightforward to relax these conditions and allow for nonparametric or data-adaptive/machine-learning-based estimators; see, e.g., \citet{Ahrens2025_DML_JEL} for an empirically-oriented discussion of causal double machine learning methods. Assumption \ref{ass:integrability} imposes mild regularity constraints on the moments of the estimating equations, preventing ill-behaved variance properties and ensuring the stability of higher-order approximations. 

In what follows, we omit $W$ and $X$  from the weights and outcome regressions to minimize notation, and for a generic $\kappa^{g,t}_{\gc}$, let 
\begin{align}
 \psi^{g,t}_{\gc}(W; \kappa^{g,t}_{\gc}) =  \psi^{g,t}_{S=g,Q=0}(W; \kappa^{g,t}_{\gc}) +  \psi^{g,t}_{S=\gc,Q=1}(W; \kappa^{g,t}_{\gc}) -  \psi^{g,t}_{S=\gc,Q=0}(W; \kappa^{g,t}_{\gc}),
 \label{eqn:composite_influence_function}
\end{align}
where, for $(g',q') \in \{(g,0), (\gc,1), (\gc,0)\}$, $\psi^{g,t}_{S=g',Q=q'}(W; \kappa^{g,t}_{\gc})$ is an influence function for one of the three DR DiD components of the DR DDD, and is given by
\begin{align}
 \psi^{g,t}_{S=g',Q=q'}(W; \kappa^{g,t}_{\gc}) = \psi^{g,t,1}_{S=g',Q=q'}(W; \kappa^{g,t}_{\gc}) - \psi^{g,t,0}_{S=g',Q=q'}(W; \kappa^{g,t}_{\gc}) -  \psi^{g,t,est}_{S=g',Q=q'}(W; \kappa^{g,t}_{\gc}),
 \label{eqn:composite_influence_function_each}
\end{align}
with
\begin{align*}
    \psi^{g,t,1}_{S=g',Q=q'}(W; \kappa^{g,t}_{\gc})= &~{w}^{S=g,Q=1}_{\text{trt}}\left(Y_t- Y_{g-1} - {m}_{Y_t-Y_{g-1}}^{S=a,Q=b}(\gamma^{reg}_{ g,t,\gc}) \right) \\
    &~~~~-  {w}^{S=g,Q=1}_{\text{trt}}\E\left[{w}^{S=g,Q=1}_{\text{trt}}\left(Y_t- Y_{g-1} - {m}_{Y_t-Y_{g-1}}^{S=a,Q=b}(\gamma^{reg}_{ g,t,\gc}) \right) \right]\\
     \psi^{g,t,0}_{S=g',Q=q'}(W; \kappa^{g,t}_{\gc})= &~{w}^{S=g,Q=1}_{g',q'}(\gamma^{ps}_{ g,t,\gc})\left(Y_t- Y_{g-1} - {m}_{Y_t-Y_{g-1}}^{S=a,Q=b}(\gamma^{reg}_{ g,t,\gc}) \right) \\
    &~~~~-  {w}^{S=g,Q=1}_{g',q'}(\gamma^{ps}_{ g,t,\gc})\E\left[{w}^{S=g,Q=1}_{g',q'}(\gamma^{ps}_{ g,t,\gc})\left(Y_t- Y_{g-1} - {m}_{Y_t-Y_{g-1}}^{S=a,Q=b}(\gamma^{reg}_{ g,t,\gc}) \right) \right]
 \end{align*}   
 and 
 \begin{align*}
     \psi^{g,t,est}_{S=g',Q=q'}(W; \kappa^{g,t}_{\gc}) = l^{g,t,reg}_{S=g',Q=q'}(\gamma^{ref}_{ g,t,\gc})' M^{g,t, 1}_{S=g',Q=q'}(\kappa^{g,t}_{\gc}) + l^{g,t,ps}_{S=g',Q=q'}(\gamma^{ps}_{ g,t,\gc})' M^{g,t,2}_{S=g',Q=q'}(\kappa^{g,t}_{\gc})
 \end{align*}
where $l^{g,t,reg}_{S=g',Q=q'}(\cdot)$ is the asymptotic linear representation of the outcome evolution for the group with $S=g'$ and $Q=q'$ as described in Assumption \ref{ass:ALR}, $l^{g,t,ps}_{S=g',Q=q'}(\cdot)$ is defined analogously for the generalized propensity score that uses group $S=a,Q=b$ as a comparison group, and 
 {\fontsize{10.4pt}{10pt}\selectfont
 \begin{align*}
    M^{g,t,1}_{S=g',Q=q'}(\kappa^{g,t}_{\gc}) &= \E\left[ \left({w}^{S=g,Q=1}_{\text{trt}} - {w}^{S=g,Q=1}_{g',q'}(\gamma^{ps}_{g,t,\gc})\right) \dot{m}_{Y_t-Y_{g-1}}^{S=a,Q=b}(\gamma^{reg}_{g,t,\gc}) \right],\\
    M^{g,t,2}_{S=g',Q=q'}(\kappa^{g,t}_{\gc}) & = \E \left[ {\alpha}^{S=g, Q=1}_{g',q'}(\gamma^{ps}_{ g,t,\gc}) \left(Y_t- Y_{g-1} - {m}_{Y_t-Y_{g-1}}^{S=a,Q=b}(\gamma^{reg}_{ g,t,\gc})  \right) \cdot \dot{p}^{S=g,Q=q}_{g',q'}(\gamma^{ps}_{g,t,\gc})\right]\\
    &~~~ - \E \left[ {\alpha}^{S=g, Q=1}_{g',q'}(\gamma^{ps}_{ g,t,\gc}) \left(\E\left[ {w}^{S=g,Q=1}_{g',q'}(\gamma^{ps}_{g,t,\gc})\left(Y_t- Y_{g-1} - {m}_{Y_t-Y_{g-1}}^{S=a,Q=b}(\gamma^{reg}_{ g,t,\gc}) \right)\right] \right) \cdot \dot{p}^{S=g,Q=q}_{g',q'}(\gamma^{ps}_{g,t,\gc})\right] 
 \end{align*}
}
with $\dot{m}_{Y_t-Y_{g-1}}^{S=a,Q=b}(\gamma^{reg}_{g,t,\gc}) = \partial{m}_{Y_t-Y_{g-1}}^{S=a,Q=b}(\gamma^{reg}_{g,t,\gc})\big/ \partial \gamma^{reg}_{g,t,\gc}$, $\dot{p}^{S=g,Q=q}_{g',q'}(\gamma^{ps}_{g,t,\gc}) = \partial {p}^{S=g,Q=q}_{g',q'}(\gamma^{ps}_{g,t,\gc}) \big/\gamma^{ps}_{g,t,\gc} $, and
 {\fontsize{10.4pt}{10pt}\selectfont
\begin{align*}
{\alpha}^{S=g, Q=1}_{g',q'}(\gamma^{ps}_{ g,t,\gc}) = \left. \dfrac{1{\{S=a,Q=b\}} }{\left(1 - {p}^{S=g,Q=1}_{g',q'}(X;\gamma^{ps}_{g,t,\gc})\right)^2} \right/ \mathbb{E}\bracks{\dfrac{1{\{S=a,Q=b\}} \cdot {p}^{S=g,Q=1}_{g',q'}(X;\gamma^{ps}_{g,t,\gc}) }{1 - {p}^{S=g,Q=q}_{g',q'}(X;\gamma^{ps}_{g,t,\gc})}}.
\end{align*}
}

For each $g\in \mathcal{G}_{\text{trt}}$ and each $t\in \{2,3,\dots, \}$, let  ${ATT}_{\text{dr}}(g,t)$ denote the $k_{g,t}\times 1$ vector of ${{ATT}}_{\text{dr}, g_{\text{c}}}(g,t)$ for all (non-collinear) $\gc \in \mathcal{G}_{\text{c}}^{\text{g,t}}$, and ${\Omega}_{g,t}$ be the asymptotic variance-covariance matrix of $\sqrt{n} \left( \widehat{ATT}_{\text{dr}}(g,t) -  {ATT}_{\text{dr}}(g,t)\right)$, i.e., $ {\Omega}_{g,t} = \E\left[ \psi^{g,t}(W; \kappa^{g,t})\psi^{g,t}(W; \kappa^{g,t})'\right]$, with $\psi^{g,t}(W; \kappa^{g,t})$ being the $k_{g,t} \times 1$ vector that stacks all non-collinear $\psi^{g,t}_{\gc}(W; \kappa^{g,t}_{\gc})$ for $\gc \in \mathcal{G}_{\text{c}}^{\text{g,t}}$. Let  ${ATT}_{\text{dr},\text{gmm}}(g,t) = ({\textbf{1}'{\Omega}^{-1}_{g,t}\textbf{1}})^{-1} {\textbf{1}'{\Omega}^{-1}_{g,t}}~ ATT_{\text{dr}}(g,t)$, and for a generic set of weights that sum up to one, let $ {ATT}_{\text{dr},{w}}(g,t) = \sum_{\gc\in \mathcal{G}_{\text{c}}^{\text{g,t}}} w^{\text{g,t}}_{\gc}~ ATT_{\text{dr}, g_{\text{c}}}(g,t)$, and recall that $\widehat{ATT}_{\text{dr},\widehat{w}}(g,t)$ is its empirical analogue as defined in \eqref{eqn:ATT_gt_w_estimated}. 
Finally, let 
$\widehat{\Omega}_{g,t}$ be the empirical analogue of ${\Omega}_{g,t}$, where one replaces expectations by sample analogues and $\kappa^{g,t}$ with $\widehat{\kappa}^{g,t}$, and consider the following claim:
\begin{align}
    &\text{For each } g\in \mathcal{G}_{\text{trt}}, t \in \{2,\dots, T\} \text{ such that }t\geq g, \text{ and each } \gc\in \mathcal{G}_{\text{c}}^{\text{g,t}} ,\nonumber\\
    &\text{ we have that,  for each } (g',q') \in \{(g,0), (\gc,1), (\gc,0)\},\nonumber\\
    & \exists \gamma^{ps}_{0, g,t,\gc} \in \Theta^{ps}: \P({p}^{S=g,Q=q}_{g',q'}(X;\gamma^{ps}_{0,g,t,\gc}) = {p}^{S=g,Q=q}_{g',q'}(X)) = 1 \text{ or} \label{claim:doubly_robust}\\
     & \exists \gamma^{reg}_{0, g,t,\gc} \in \Theta^{reg}: \P({m}^{S=g,Q=q}_{g',q'}(X;\gamma^{reg}_{0,g,t,\gc}) = {m}^{S=g,Q=q}_{g',q'}(X)) = 1. \nonumber
\end{align}

Claim \eqref{claim:doubly_robust} states that for each $(g,t)$-pair and each suitable comparison group $\gc$, either the working parametric model for the generalized propensity score is correctly specified, or the working outcome regression model for the comparison group is correctly specified for each of the three DiD components of our DDD estimator. Thus, eight possible working model combinations would lead to consistent DDD estimation of the $ATT(g,t)$  parameter.

The next theorem establishes the limiting distribution of $\widehat{ATT}_{\text{dr},\gc}(g,t)$ and $\widehat{ATT}_{\text{dr},\text{gmm}}(g,t)$.

\begin{theorem}[Consistency and Asymptotic Normality]\label{thm:CAN}
    Let Assumptions \ref{ass:sampling_panel}, \ref{ass:overlap_staggered}, \ref{ass:anticipation}, \ref{ass:PT-NYT}, \ref{ass:wm}, \ref{ass:ALR}, and \ref{ass:integrability} hold. Then, for all $g\in \mathcal{G}_{\text{trt}}$, $t\in \{2,\dots, T\}$, and $g_{\text{c}} \in \mathcal{G}_{\text{c}}^{\text{g,t}}$ such that $t\ge g$, provided that \eqref{claim:doubly_robust} is true,
\begin{align*}
    \sqrt{n}\left(\widehat{ATT}_{\text{dr},\gc}(g,t) - {ATT}(g,t) \right) &= \dfrac{1}{\sqrt{n}}\sum_{i=1}^n \psi^{g,t}_{\gc}(W_i; \kappa^{g,t}_{0, \gc})+ o_p(1)  \overset{d}{\rightarrow} N(0,\Omega_{g,t,\gc}),
    \end{align*}
   where $\Omega_{g,t,\gc} = \E\left[ \psi^{g,t}_{\gc}(W_i; \kappa^{g,t}_{0, \gc})\psi^{g,t}_{\gc}(W_i; \kappa^{g,t}_{0, \gc})'\right]$. Furthermore, 
   \begin{align*}
    \sqrt{n}\left(\widehat{ATT}_{\text{dr},\text{gmm}}(g,t) - {ATT}(g,t) \right) &=  \dfrac{\textbf{1}'{\Omega}^{-1}_{g,t}}{\textbf{1}'{\Omega}^{-1}_{g,t}\textbf{1}} \dfrac{1}{\sqrt{n}}\sum_{i=1}^n \psi^{g,t}(W_i; \kappa^{g,t}_{0})+ o_p(1) \overset{d}{\rightarrow} N(0,\Omega_{g,t, \text{gmm}}),
    \end{align*}
    where $\Omega_{g,t, \text{gmm}} = \left({\textbf{1}'{\Omega}^{-1}_{g,t}\textbf{1}}\right)^{-1} \leq \Omega_{g,t,\gc}$ for any $\gc \in \mathcal{G}_{\text{c}}^{\text{g,t}}$. In fact, for any set of weights $w$ that sum up to one over the $\mathcal{G}_{\text{c}}^{\text{g,t}}$, $\Omega_{g,t, \text{gmm}} \leq \Omega_{g,t, w}$, with $\Omega_{g,t, w}$ defined as the asymptotic variance of $\sqrt{n}\left(\widehat{ATT}_{\text{dr},\widehat{w}}(g,t) - {ATT}_{\text{dr},{w}}(g,t) \right)$.
\end{theorem}

Theorem \ref{thm:CAN} provides the influence function for estimating each $ATT(g,t)$, using different comparison groups $\gc$, as well as establishes the consistency and asymptotic normality of our DR DDD estimator $\widehat{ATT}_{\text{dr},\gc}(g,t)$. Theorem \ref{thm:CAN} also highlights that combining different comparison groups as our DR DDD estimator $\widehat{ATT}_{\text{dr},\text{gmm}}(g,t)$ does is effective in terms of asymptotically improving precision. That is, Theorem \ref{thm:CAN} highlights that $\widehat{ATT}_{\text{dr},\text{gmm}}(g,t)$ is optimal in the sense that it asymptotically achieves the minimum variance across all weighted average estimators that combine multiple $\widehat{ATT}_{\text{dr},\gc}(g,t)$s. Importantly, Theorem \ref{thm:CAN} also highlights the doubly (or multiply) robust property of our DDD estimators: they recover the $ATT(g,t)$ provided that each of the three DR DiD estimators has a correctly specified outcome regression or generalized propensity score working model.

\begin{remark}\label{rem:simul_inf}
  Although Theorem \ref{thm:CAN} provides pointwise inference results for each $ATT(g,t)$, it is straightforward to extend it to hold simultaneously across multiple $ATT(g,t)$'s. For instance, by letting $\widehat{ATT}_{\text{gmm},t\ge g}$ and ${ATT}_{\text{gmm},t\ge g}$ denote the vector of $\widehat{ATT}_{\text{dr},\text{gmm}}(g,t)$ and ${ATT}_{\text{dr},\text{gmm}}(g,t)$, respectively, for all $g\in\mathcal{G}_{\text{trt}}$, $t\in \{2,\dots, T,\}$ such that $t\ge g$, it is straightforward to show that $\sqrt{n}\left(\widehat{ATT}_{\text{gmm},t\ge g} - {ATT}_{\text{gmm},t\ge g}\right) \overset{d}{\rightarrow}  N(0,\Omega)$, with $\Omega = \E[\psi^{t\ge g}_{\text{gmm}}(W_i; \kappa^{t\ge g}_{0}) \psi^{t\ge g}_{\text{gmm}}(W_i; \kappa^{t\ge g}_{0})']$, with $\psi^{t\ge g}_{\text{gmm}}(W_i; \kappa^{t\ge g}_{0})$ the asymptotic linear representation of  $\sqrt{n}\left(\widehat{ATT}_{\text{gmm},t\ge g} - {ATT}_{\text{gmm},t\ge g}\right)$. One can then construct simultaneous confidence bands using a simple-to-use multiplier bootstrap as discussed in Theorem 3 and Algorithm 1 of \citet{Callaway2021}. It is also straightforward to conduct cluster-robust inference; see Remark 10 of \citet{Callaway2021}. As these results are commonly accessible, we will not include them here to conserve space.
\end{remark}

\subsubsection{Asymptotic theory for event-study parameters}

In this section, we derive large sample properties for our event-study estimator $\widehat{ES}_{\text{dr},\text{gmm}}(e)$ as defined in \eqref{eqn:es_estimator}. Given that $\P_n(G=g|G+e \in [1,T])$ is an $\sqrt{n}$-consistent and asymptotically normal estimator of $\P(G=g|G+e \in [1,T])$, then for all $g \in \mathcal{G}_{\text{trt}}$, we have that
\begin{equation}
\label{alr_weights_es}
    \sqrt{n}(\P_n(G=g|G+e \in [1,T]) - \P(G=g|G+e \in [1,T])) = \frac{1}{\sqrt{n}} \sum_{i=1}^{n} \xi^{g,e}(W_i) + o_p(1),
\end{equation}
with $\E[\xi^{g,e}(W)] = 0$ and $\E[\xi^{g,e}(W) \xi^{g,e}(W)'] < \infty$ being positive definite, and 
\begin{equation*}
\xi^{g,e}(W) = \frac{1}{\P(G+e \in [1,T])} \cdot \bigg[ 1\{G=g, G+e \in [1,T]\} - \P(G=g| G+e \in [1,T]) \cdot 1\{G+e \in [1,T] \} \bigg]
\end{equation*} 


The following corollary can be used to conduct asymptotically valid (pointwise) inference for the event-study type parameter $ES(e)$.

\begin{corollary}
\label{cor:asy_es}
    Under the assumptions of Theorem \ref{thm:CAN}, for each $e$ such that $\P(1 \leq G+e \leq T)$, as $n\rightarrow \infty$,
    \begin{align*}
       \sqrt{n}(\widehat{ES}_{\text{dr},\text{gmm}}(e) - ES(e)) &= \frac{1}{\sqrt{n}} \sum_{i=1}^{n} l^{es,e}_{\text{gmm}}(W_i) + o_p(1)\\
       & \overset{d}{\rightarrow} N(0, \E[l^{es,e}_{\text{gmm}}(W)^2]),
    \end{align*}
    with $l^{es,e}_{\text{gmm}}(W) = \sum_{g \in \mathcal{G}_{\text{trt}}} \Big( \P(G=g| G+e \in [1,T]) \cdot \dfrac{\textbf{1}'{\Omega}^{-1}_{g,t}}{\textbf{1}'{\Omega}^{-1}_{g,t}\textbf{1}} \psi^{g,t}(W_i; \kappa^{g,t}_{0}) + \xi^{g,e}(W_i) \cdot ATT(g,t) \Big)$.
\end{corollary}

The results in Corollary \ref{cor:asy_es} also apply to estimators of $ES(e)$ using $\widehat{ATT}_{\text{dr},g_c}(g,g+e)$ on \eqref{eqn:es_estimator}.  Corollary \ref{cor:asy_es} focuses on pointwise inference procedures. Still, as discussed in Remark \ref{rem:simul_inf}, it is straightforward to extend it to hold for all event-times $e$ and conduct simultaneous-based inference. The asymptotic results for our overall summary parameter $\widehat{ES}_{\text{avg},\text{gmm}}$ as defined in \eqref{eqn:overall_ATT_estimate} follow from the delta method and are omitted.

\section{Monte Carlo Simulations}\label{simulations}

In this section, we evaluate the finite sample properties of our proposed DR DDD estimators via Monte Carlo simulations. We examine two scenarios: (i) when covariates play a crucial role in identification across two time periods, and (ii) when there are multiple time periods with variation in treatment timings. For the first scenario, we have panel data for two time periods, $t=1,2$, four covariates, two enabling-groups $S \in \mathcal{S}_{des-1} \equiv \{2, \infty\}$, and there are two eligibility groups: $Q=1$ and $Q=0$. In the setup with staggered adoption, we consider the simplest possible case with three time periods, $t=1,2,3$, with $S \in \mathcal{S}_{des-2} \equiv \{2,3,\infty\}$, and we abstract from covariates in the main text. We relegate simulation results with DDD staggered adoption with covariates to the Supplemental Appendix. In the main text, we compare the performance of different DDD estimators via graphs: one that presents the density of the point estimates across the 1,000 Monte Carlo repetitions, and one that presents the length of confidence intervals in each Monte Carlo draw. In the Supplemental Appendix, we also report the traditional summary statistics for the Monte Carlo involving average bias, root mean square error (RMSE), empirical 95\% coverage probability, and the average length of a 95\% confidence interval under 1,000 Monte Carlo repetitions. Light-gray confidence intervals mean that they do not contain the true parameter of interest, $ATT(2,2)$ in our simulations, and are appropriately colored when they contain it. We focus on results with $n=5,000$ but report results for different sample sizes in the Supplemental Appendix \ref{appendix:simulations}.

\subsection{Simulations for DDD with two periods and covariates}
\label{sec:sims_2_periods}

We describe the data-generating process (DGP) for the 2-period DDD setup. For a generic four-dimensional vector $O$, the conditional probability of each unit belonging to a subgroup $(g,q) \in \{2,\infty\} \times \{0,1\}$ is
\begin{equation}
\label{ps_score_sims}
    \P[S = g, Q=q | O] \equiv p^{S=g, Q=q}(O) =  \dfrac{\exp(f^{ps}_{S=g,Q=q}(O))}{\sum_{(g,q) \in \mathcal{S}_{\text{des-1}} \times \{0,1\}}\exp(f^{ps}_{S=g,Q=q}(O))},
\end{equation}
where $f^{ps}_{S=g,Q=q}(O))$ is a linear index with heterogeneous coefficients across sub-groups; we defined these in the Supplemental Appendix \ref{appendix:sims_2_periods} to save space. Of course, each unit belongs to a single subgroup, and we assigned these subgroups as follows:  
\begin{equation}
    \label{eq_appen:assignment_process}
    (S,Q) := \begin{cases}
(\infty,0), & \text { if } U \leqslant p^{S=\infty, Q=0}(O), \\ 
(\infty,1), & \text { if } p^{S=\infty, Q=0}(O) < U \leq \sum_{j=0}^1 p^{S=\infty, Q=j}(O), \\ 
(2,0), & \text { if } \sum_{j=0}^1 p^{S=\infty, Q=j}(O) < U \leq \sum_{j=0}^1 p^{S=\infty, Q=j}(O) + p^{S=2, Q=0}(O), \\
(2,1), & \text { if }  \sum_{j=0}^1 p^{S=\infty, Q=j}(O) + p^{S=2, Q=0}(O) < U,\end{cases}
\end{equation}
with $U$ being a uniform random variable in $[0,1]$, independent of all other variables. 

The potential outcomes are defined as 
\begin{align}
    Y_{i,1}(\infty) &= f^{reg}(O_{i}, S_i) + \nu_i(O_i, S_i, Q_i) + \varepsilon_{i,1}(\infty) \nonumber\\
    Y_{i,2}(\infty) &= 2 f^{reg}(O_{i}, S_i)+ \nu_i(O_i, S_i, Q_i) + \varepsilon_{i,2}(\infty) \label{eqn:PO_2_periods}\\
    Y_{i,2}(2) &= 2 f^{reg}(O_{i}, S_i)+ \nu_i(O_i, S_i, Q_i) + \varepsilon_{i,2}(2),\nonumber
\end{align}  
where $f^{reg}(O_{i}, S_i)$ is a linear regression specification with heterogeneous coefficients across the enabling groups $S$, $\nu_i(O_i, S_i, Q_i)$ is a time-invariant unobserved heterogeneity correlated with covariates and sub-groups, and $\varepsilon_{i,1}(\infty), \varepsilon_{i,2}(\infty)$ and $\varepsilon_{i,2}(2)$ are independent standard normal random variables; we provide a precise definition of $f^{reg}(O_{i}, S_i)$ and $\nu_i(O_i, S_i, Q_i)$ in the Supplemental Appendix \ref{appendix:sims_2_periods}. Note that our designs' $ATT(2,2)$ equals zero, though there is treatment effect heterogeneity across units. We observe untreated outcomes for all units in period $t=1$; in period $t=2$, we observed $Y_{i,2}(2)$ if unit $i$ belongs to group $S=2$, $Q=1$, and observe $Y_{i,2}(\infty)$ otherwise.

Building on \cite{kang_schafer_2007} and \citet{SantAnna2020}, we allow propensity score and/or outcome regression models to be misspecified. We consider four different types of DGP: DGP 1, where all models are correctly specified; DGP 2, where outcome models are correctly specified but the propensity score model is misspecified; DGP 3, where the propensity score is correctly specified but outcome regressions are not; and DGP 4, where all models are misspecified. The source of misspecification in these nuisance models is related to whether they depend on $X$ or $Z$, where $X$ is a nonlinear transformation of all the $Z$'s. In our simulations, the observed data is $W_i = \{Y_{i,1}, Y_{i,2}, S_{i}, Q_{i}, X_{i}\}_{i=1}^{n}$, so using $Z$ as linear covariates in these nuisance models lead to working model misspecification; we relegate to the Supplemental Appendix \ref{appendix:sims_2_periods} the definition of $X$'s and $Z$'s. 

We compare the performance of four different estimators for $ATT(2,2)$, just like in Section \ref{sec:ddd_2periods_x}: our DR DDD estimator as defined in \eqref{eqn:DRDDD_2periods} (we label it as DRDDD), 3WFE OLS estimator for $\tilde{\beta}_{\text{3wfe}}$ based on \eqref{eqn:3WFE_with_X} (we label it as 3WFE), 3WFE OLS estimator of $\check{\beta}_{\text{3wfe}}$ based on \eqref{eqn:3WFE_with_X_Mundlak} (we label it as M-3WFE), and the difference of two \citet{SantAnna2020}' DR DiD estimators (we label it as DRDID-DIF). We summarize the results of our simulations in Figure \ref{fig:Sims_2_periods_MC}, where we consider a sample size $n=5,000$ and conducted 1,000 Monte Carlo repetitions. The left panels display the density of the point estimates across all Monte Carlo draws, and the right panels display the 95\% confidence intervals for each Monte Carlo draw. See Table \ref{mc_sim_att22} in the Supplemental Appendix for additional results.

\begin{figure}[htp]
  \centering

  \begin{subfigure}[t]{0.9\textwidth}
  
    \centering
      \includegraphics[width=\linewidth]{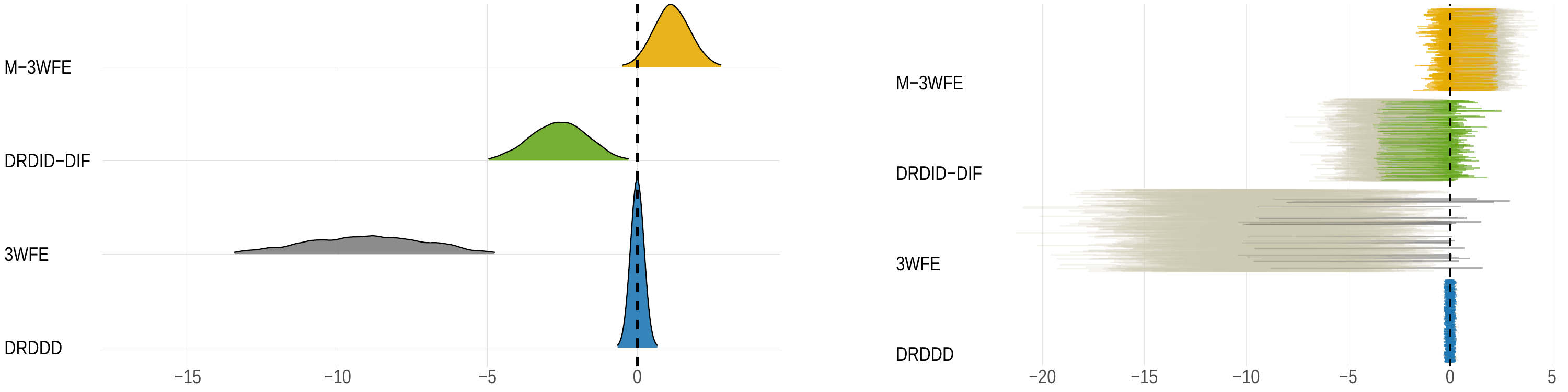}
      \caption{DGP 1: All working models are correctly specified}
    
  \end{subfigure}

  \begin{subfigure}[t]{0.9\textwidth}
 
    \centering
      \includegraphics[width=\linewidth]{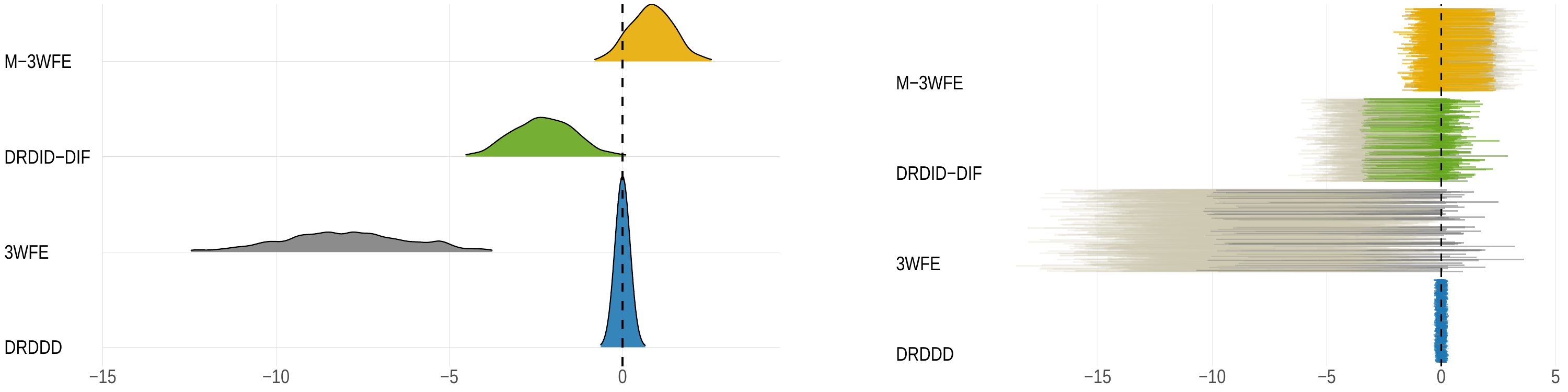}
       \caption{DGP 2: Outcome working models are correctly specified}
    
  \end{subfigure}

  \begin{subfigure}[t]{0.9\textwidth}
  
    \centering
      \includegraphics[width=\linewidth]{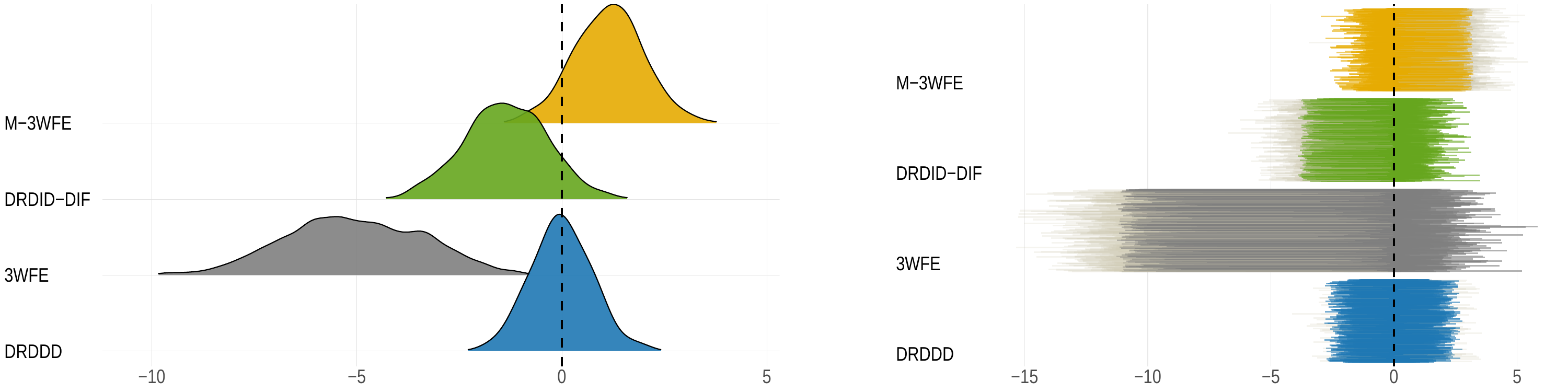}
      \caption{DGP 3: Propensity-score working models are correctly correct}
    
  \end{subfigure}

  \begin{subfigure}[t]{0.9\textwidth}
  
    \centering
      \includegraphics[width=\linewidth]{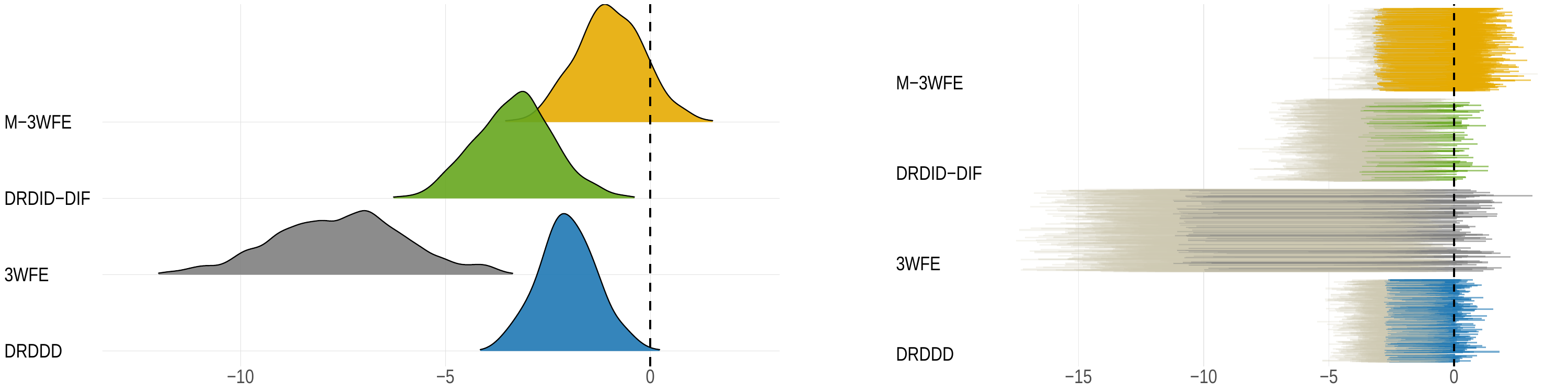}
      \caption{DGP 4: All working models are misspecified}
    
  \end{subfigure}

    \caption{Monte Carlo Simulation Results for DDD: two-period setup with covariates}
    \label{fig:Sims_2_periods_MC}
        \justifying
	\noindent \scriptsize{Notes: Simulation designs as discussed in text, with $n=5,000$ and $1,000$ Monte Carlo repetitions. True $ATT(2,2)$ is zero and is indicated in the solid vertical line in all panels. 3WFE corresponds to the OLS estimates of $\tilde{\beta}_{\text{3wfe}}$ based on \eqref{eqn:3WFE_with_X}. M-3WFE corresponds to the OLS estimates of $\check{\beta}_{\text{3wfe}}$ based on \eqref{eqn:3WFE_with_X_Mundlak}. DRDID-DIF corresponds to the difference between two doubly robust DiD estimators \citep{SantAnna2020}. DRDDD corresponds to our proposed doubly robust DDD estimator described in \eqref{eqn:DRDDD_2periods}. All densities (left) and confidence intervals (right) are computed across all simulation draws. Light grey areas in the right plots indicate confidence intervals that exclude the true $ATT(2,2)$, where increased prominence suggests lower empirical coverage.
 }
\end{figure}

The results in Figure \ref{fig:Sims_2_periods_MC} are self-explanatory and firmly support our theoretical results. When outcome regression or propensity score models (but not necessarily both) are correctly specified, our DR DDD estimators are appropriately centered (so they are unbiased), their confidence intervals are the narrowest across all other estimators, and they still have appropriate coverage across the first three DGPs (94.4\%, 94.5\%, and 94.6\%, respectively). For instance, when all working models are correctly specified, the average length of the 95\% confidence interval of the M-3WFE estimator is 7 times longer than our DR DDD estimator; this difference is much larger for the other considered estimators. In fact, the performance of all other estimators in all our considered DGPs is poor, as they have non-negligible bias, high RMSE, and poor coverage properties. In DGP 4, when all working models are misspecified, we note that our estimator is biased, directly affecting the confidence intervals' coverage probabilities. None of the considered DDD estimators perform well when all working models are misspecified (DGP 4), highlighting that all estimators indeed depend on modeling assumptions.

\subsection{Simulations for DDD with variation in treatment timing}
\label{sims_staggered}
We now discuss the staggered DDD setup with three time periods, $t=1,2,3$, three enabling groups, $S \in  \{2,3,\infty\} $, and two eligibility groups $Q \in \{0,1\}$. We abstract from covariates and defer a discussion about them to the Supplemental Appendix. 

Each unit $i$, we have that $p^{S=2,Q=0}=0.20$, $p^{S=2,Q=1}=0.15$, $p^{S=3,Q=0}=0.30$, $p^{S=3,Q=1}=0.20$, $p^{S=\infty,Q=0}=0.05$, and $p^{S=\infty,Q=1}=0.10$. We then randomly assign the realized value of $(S,Q)$ based on the above distribution. The potential outcomes are generated as 
\begin{align}
    Y_{i,1}(\infty) &= (1 + Q_i)\alpha +  \nu_i(S_i, Q_i) + \varepsilon_{i,1}(\infty) \nonumber\\
    Y_{i,2}(\infty) &=  (2 + Q_i) \alpha  + 1.1 \nu_i(S_i, Q_i) + \varepsilon_{i,2}(\infty) \nonumber\\
    Y_{i,3}(\infty) &=  (3 + Q_i) \alpha  + 1.2 \nu_i(S_i, Q_i) + \varepsilon_{i,3}(\infty) \nonumber\\
    Y_{i,2}(2) &=  (2 + Q_i)  \alpha  + 1.1 \nu_i(S_i, Q_i)  + ATT(2,2) Q_i +  \varepsilon_{i,2}(2) \label{eqn:PO_staggered}\\
    Y_{i,3}(2) &=  (3 + Q_i)  \alpha  + 1.2 \nu_i(S_i, Q_i) +  ATT(2,3) Q_i + \varepsilon_{i,3}(2)\nonumber\\
    Y_{i,3}(3) &=  (3 + Q_i)  \alpha  +  1.2 \nu_i(S_i, Q_i) + ATT(3,3) Q_i + \varepsilon_{i,3}(3)\nonumber,
\end{align}  
where we set $\alpha = 278.5$, $ATT(2,2)=10$, $ATT(2,3)=20$, $ATT(3,3)=25$, all $\varepsilon_{i,t}(\cdot)$ are independent standard normal, and $\nu_i(S_i, Q_i)$ is an unobserved heterogeneity term that we formally define in the Supplemental Appendix \ref{appendix:sims_staggered}. The observed data is represented as $W_i = \{Y_{i, 1}, Y_{i,2}, Y_{i,3}, S_{i}, Q_{i}\}_{i=1}^{n}$, where the $Y_{i,1} = Y_{i,1}(\infty)$,  $Y_{i,2} = Y_{i,2}(2)$ for units with $S_i=2,Q_i=1$, and $Y_{i,2} = Y_{i,2}(\infty)$ otherwise, and $Y_{i,3} = Y_{i,3}(3)$ for units with $S_i=3,Q_i=1$, $Y_{i,3} = Y_{i,3}(2)$ for units with $S_i=2,Q_i=1$, and $Y_{i,3} = Y_{i,3}(\infty)$ for all other units. In what follows, we focus our attention on estimators for $ATT(2,2)$, though we compare estimators for $ATT(2,3)$ and $ATT(3,3)$ in the Supplemental Appendix \ref{appendix:sims_staggered}. We consider $n=5,000$ here and refer to Supplemental Appendix \ref{appendix:sims_staggered} for information on other sample sizes.

We compare the performance of three staggered DDD estimators for the $ATT(2,2)$ as in Section \ref{sec:ddd_stag_noX}. More precisely, we consider our optimally GMM-weighted estimator $\widehat{ATT}_{\text{gmm}}(g,t)$ as defined in \eqref{eqn:att_gt_gmm}, which is a special case of the \eqref{eqn:DRDDD_optimal} without covariates, i.e., with $X_i=1$ for all units; we refer to this estimator as $DDD_{\text{gmm}}$. We also consider our DDD estimator that uses never-treated units as the comparison group as defined in \eqref{eqn:att_gt_one-at-time} with $ g_{\text{c}} = \infty$; we refer to this estimator as $DDD_{\text{nev}}$.\footnote{This estimator is a special case of \eqref{eqn:DRDDD_stagg_estimated} when covariates are trivial}. Finally, we consider a DDD estimator that pools all not-yet-treated units a la \citet{Callaway2021}, $\widehat{ATT}_{\text{cs-nyt}}(g,t)$, as defined in \eqref{eqn:att_gt_pooled}; we refer to this estimator as $DDD_{\text{cs-nyt}}$. Our theoretical results indicated that $DDD_{\text{gmm}}$ and $DDD_{\text{nev}}$ should both be valid DDD estimators under our identification assumptions. At the same time, we have no statistical guarantee about the performance of $DDD_{\text{cs-nyt}}$. Our theoretical results also suggest that $DDD_{\text{gmm}}$ should be more precisely estimated than $DDD_{\text{nev}}$, as it uses more information.

\begin{figure}[htp]
\centering

  \begin{subfigure}[t]{\textwidth}
 
    \centering
      \includegraphics[width=\linewidth]{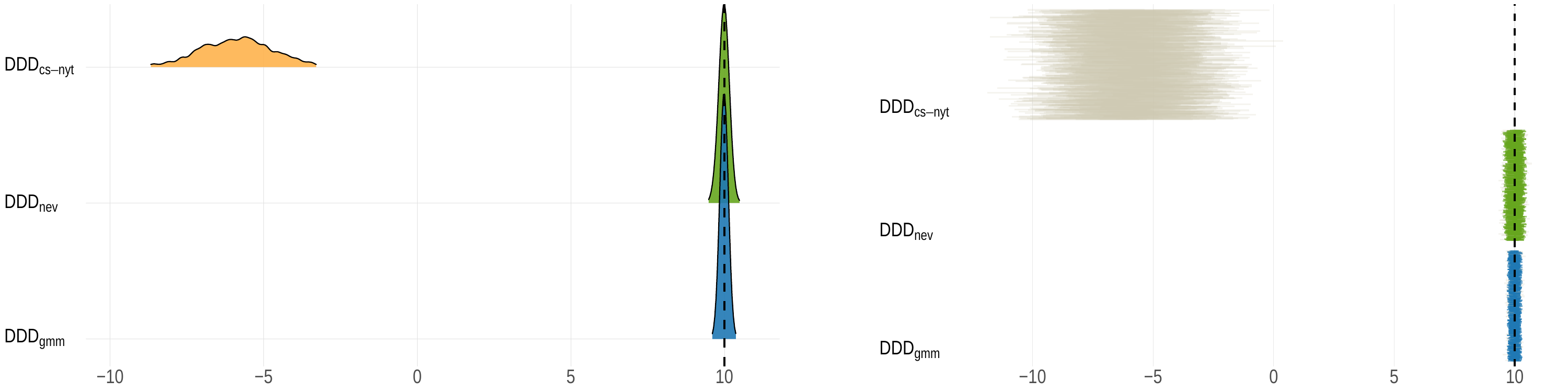}
   \caption{Comparison across the three DDD estimators}
  \end{subfigure}

  \begin{subfigure}[t]{\textwidth}
 
    \centering
      \includegraphics[width=\linewidth]{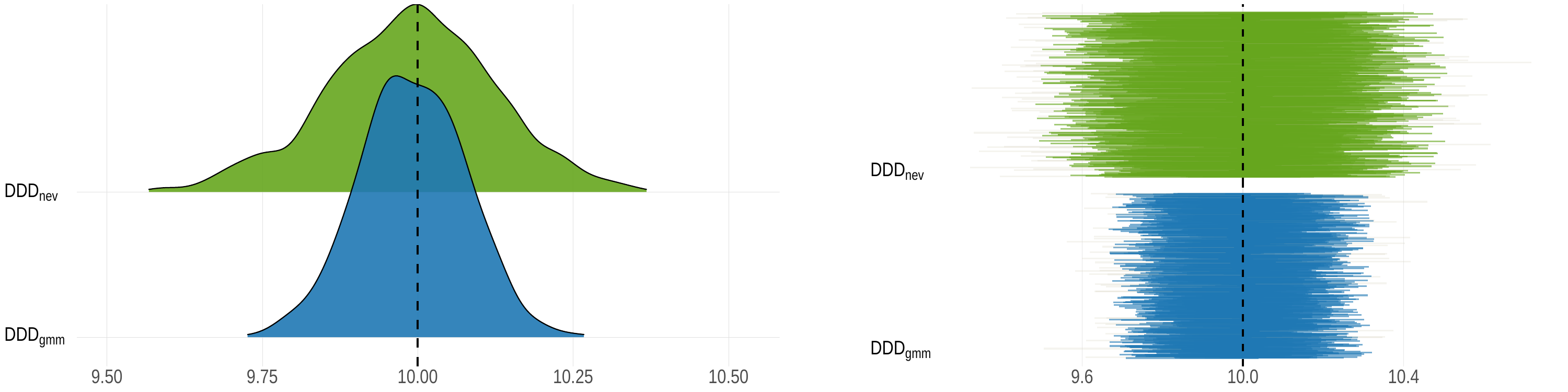}
       \caption{Comparing the performance of $DDD_{\text{{nev}}}$ and $DDD_{\text{{gmm}}}$}
  \end{subfigure}
  
  \protect\caption[position=bottom]{Monte Carlo Simulation Results for DDD: staggered setup without covariates}
    \label{fig:Sims_staggered_MC}
  \justifying
  \noindent \scriptsize{Notes: Simulation designs are detailed in the text, using $n=5,000$ and 1,000 Monte Carlo repetitions. The true $ATT(2,2)$ is 10, marked by a solid vertical line in all panels. $DDD_{\text{{nev}}}$ denotes our DDD estimator with $g_c = \infty$ from Equation \eqref{eqn:DRDDD_stagg}. $DDD_{\text{{gmm}}}$ is our proposed DDD estimator from Equation \eqref{eqn:DRDDD_optimal}. $DDD_{cs-nyt}$ is the estimator pooling all not-yet-treated units as defined in \eqref{eqn:att_gt_pooled}. In the top-left panel, we plot the densities of the $DDD_{\text{{nev}}}$, $DDD_{\text{{gmm}}}$, and $DDD_{cs-nyt}$ estimates across all simulation draws. The bottom-left panel zooms into the densities for $DDD_{\text{{nev}}}$ and $DDD_{\text{{gmm}}}$. The top-right panel shows the confidence intervals for these estimators, while the bottom-right panel focuses on $DDD_{\text{{nev}}}$ and $DDD_{\text{{gmm}}}$ only. Light grey areas in the right plots indicate confidence intervals that exclude the true $ATT(2,2)$, where increased prominence suggests lower empirical coverage.
 }
\end{figure}

Figure \ref{fig:Sims_staggered_MC} summarizes our simulation results; see also Table \ref{app:mc_sim_attgt} in the Supplemental Appendix. Panel (a) of Figure \ref{fig:Sims_staggered_MC} compares the performance of the three estimators. As it is evident, $DDD_{\text{cs-nyt}}$ is severely biased in our design, and its confidence interval never covered the true $ATT(2,2)$ in the Monte Carlo draws. This result further highlights that proceeding as if DDD is just the difference of two DiD procedures can lead to misleading conclusions. At the same time, the simulations highlight that  $DDD_{\text{{nev}}}$ and $DDD_{\text{{gmm}}}$ are unbiased and have good coverage properties, $93.2\%$ and $95\%$, respectively. In Panel (b), we drop $DDD_{\text{cs-nyt}}$ and focus on comparing the performance of our two proposed DDD estimators. As the plot makes it clear, the gains in efficiency of using all not-yet-treated comparison groups as in $DDD_{\text{{gmm}}}$ are notable. The average length of the confidence intervals of $DDD_{\text{{nev}}}$ is more than $50\%$  higher than that of $DDD_{\text{{gmm}}}$. Thus, in practice, we recommend favoring $DDD_{\text{{gmm}}}$, especially when the sample size of the never-enabling comparison group is low.

\section{Empirical Illustrations}
\label{application}

In this section, we revisit (i) \cite{cai_insurance_2016} and analyze the effect of agriculture insurance provision on household financial decisions, (ii) \citet{carbon_pricing} and analyze the effect of the emission trading scheme implemented in China on carbon emissions, and (iii) \citet{hansen_national_2023} and assess the impact of genetically modified crop adoption on countrywide yields. These three studies examine DDD setups with a single treatment date and covariates, as well as staggered DDD setups with and without covariates. The aim is to illustrate how our proposed estimator performs compared to the methods used in these original studies. Overall, we find that our DR DDD procedure can (a) produce substantial gains in precision, with 3WFE 95\% confidence intervals being up to 115\% wider than ours, (b) lead to different conclusions about the effect of a policy, or (c) provide more substantial evidence supporting the effectiveness of a policy.

\subsection{Effect of insurance provision on financial decisions}

\cite{cai_insurance_2016} analyzes the effects of agricultural insurance programs implemented in rural Chinese households on their production, borrowing, and saving behaviors. These rural areas are particularly vulnerable to unique weather shocks, resulting in income volatility for households engaged in agriculture. To address this, the People's Insurance Company of China (PICC) launched the first weather-indexed crop-insurance program in 2003 for tobacco farmers in selected counties of Jiangxi province.  The contracts were compulsory for tobacco growers in the treatment regions, allowing \cite{cai_insurance_2016} to take advantage of the variation in insurance provision across regions and household types. By design, the policy aimed to simultaneously protect farm income and sustain county fiscal revenues that heavily depend on tobacco taxes.  

The insurance policy rollout introduces three sources of variation: the timing of insurance provision (pre-policy, February 2000, and post-policy, August 2003), geographic regions (counties that enabled policy vs.  those that did not), and household eligibility conditions (tobacco-specialized vs. other farm or off-farm households). This setup enables us to allow for industry-specific and region-specific trends. \cite{cai_insurance_2016} utilizes an administrative panel from the Rural Credit Cooperative (RCC), encompassing data from over 5,000 households annually surveyed between 2000 and 2008. This dataset merges financial and savings records with RCC survey details on demographics, land allocations, and crop revenues. It includes more than 3,000 tobacco-growing households (approximately 1,200 of which are located in treated counties) and 2,200 non-tobacco households. The study evaluates outcomes such as the area of tobacco production, loan size, loan limits, interest rates, savings rates, and flexible-term savings\footnote{Although the original study analyzed outcomes related to production, borrowing, and saving behaviors, the DDD design specifically estimates the impact of insurance provision on borrowing and saving behavior, accounting for region-specific effects.}.

We adopt a similar approach as in \cite{cai_insurance_2016} by employing a DDD design to evaluate how insurance provision affects household saving behaviors. \cite{cai_insurance_2016} considered the following 3WFE event study specification\footnote{See, e.g., Equations (4) and (6) in \cite{cai_insurance_2016}.} to estimate dynamic average treatment effects of insurance provision on household saving behavior
\begin{equation}
\label{model_es_insurance}
Y_{i,t} = \gamma_i + \gamma_{r,t} + \gamma_{j,t} + \sum_{e \neq-1} \beta_{e}  \mathbf{1}\{E_{i,t}=e\} +  X'_{i,j,r}\theta + u_{i,t}, 
\end{equation}
where where $i,j,r,t$ are household, sector, region, and years indices, respectively, $\gamma_i, \gamma_{r,t}, \gamma_{j,t}$ denote household fixed effects, region-by-time fixed effect and industry-by-time fixed effect, respectively, $E_{i,t} = t - G_i$ is the time relative to the period at which household $i$ has received the insurance\footnote{In this setup, household $i$ is considered to have received the treatment if it is located in a treated region and qualifies as a tobacco household.} and $u_{i,t}$ is an idiosyncratic error term. Covariates in $X_{i,j,r}$ include household size, head of household age, and education level. Regarding outcome variables $Y_{i,t}$ for measuring household saving decisions, \cite{cai_insurance_2016} examines net saving, defined as the annual increase in total savings; saving rate, defined as the ratio of net saving to current household income; and flexible-term saving, which is the ratio of net savings in checking accounts to the total net savings. We compare the estimates of $\beta_{e}$ with those of $ES(e)$, as specified in \eqref{eqn:ES}, using our DR procedure, using the same set of covariates. In our setup, $S$ represents the geographic region that enables or not the insurance policy, and $Q$ represents the household eligibility condition---i.e., tobacco-specialized or not.


\begin{figure}[pht]
  \centering
  \begin{subfigure}[t]{0.8\textwidth}
    \centering
      \includegraphics[width=\linewidth]{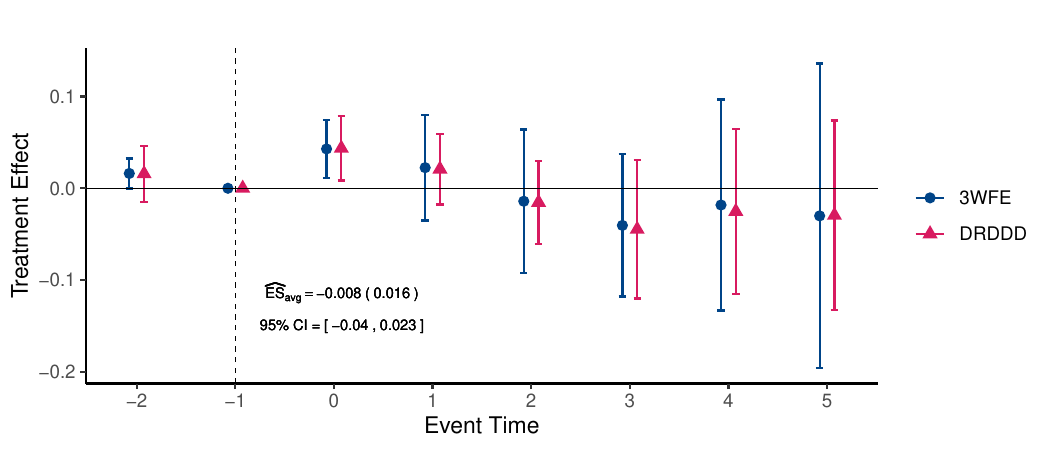}
                \caption{Panel A: Insurance Provision on Net saving}
  \end{subfigure}

  \begin{subfigure}[t]{0.8\textwidth}
    \centering
      \includegraphics[width=\linewidth]{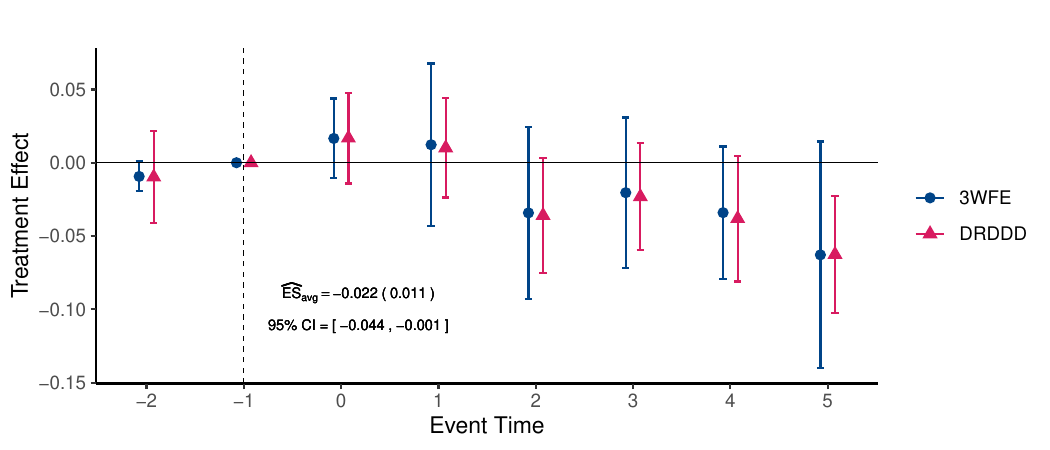}
               \caption{Panel B: Insurance Provision on Saving rate}
  \end{subfigure}

  \begin{subfigure}[t]{0.8\textwidth}
    \centering
      \includegraphics[width=\linewidth]{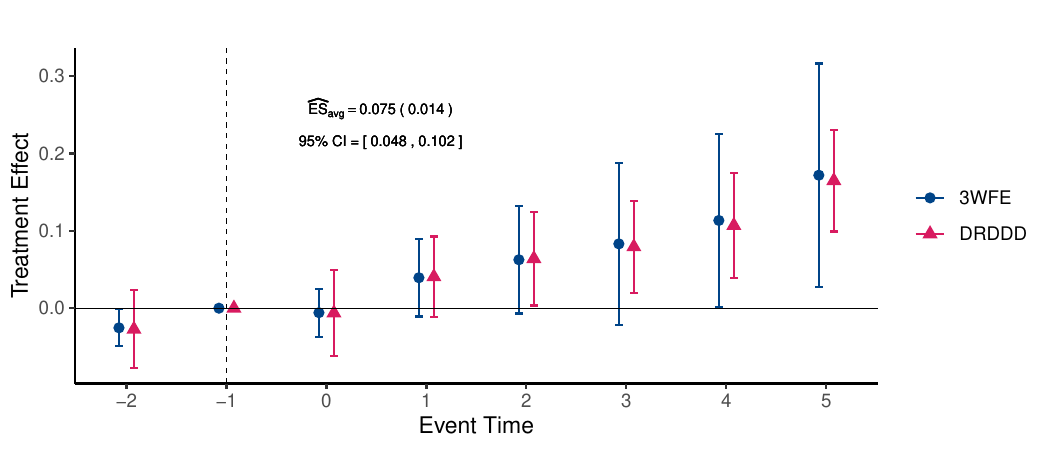}
                \caption{Panel C: Insurance Provision on Flexible-term saving}
  \end{subfigure}

    \caption{Impact of Insurance Provision on Savings Decisions.}
    \label{fig:tobacco_insurance}
        \justifying
	\noindent \scriptsize{Notes: This figure presents event study estimates with covariates including household size and head of household age. Panel A uses Net Saving as outcome, Panel B focuses on the Saving Rate, and Panel C examines Flexible-term Saving. In each panel, blue dots represent the point estimates from the specification in \eqref{model_es_insurance}, while red triangles show the $\widehat{ES}_{dr, \gc}(e)$ estimates aggregated from our proposed DR estimator as in \eqref{eqn:es_estimator}. The average $\widehat{ES}_{\text{avg}}$, calculated from all $\widehat{ES}_{dr, \gc}(e)$ values to the right of the vertical dashed line at \(e = -1\), is provided alongside its bootstrapped standard errors, based on 999 repetitions, and confidence interval. The corresponding blue and red vertical lines represent the 95\% confidence intervals.
 }
\end{figure}
Figure \ref{fig:tobacco_insurance} plots the event study estimates for the 2003 weather insurance policy. Both the 3WFE and our DR DDD procedure yield similar results in terms of magnitude. However, our estimator offers a precision advantage in this particular context, with 3WFE 95\% confidence intervals being up to 1.15 times wider than the one associated with our proposed DR DDD estimator. In Panel A of Figure \ref{fig:tobacco_insurance}, we observe an insignificant effect on net saving, in line with the results reported in \cite{cai_insurance_2016}. In Panel B, we observe a negative effect on the saving rate, consistent with previous findings by \cite{cai_insurance_2016}. However, unlike earlier results, which indicated an insignificant negative effect, our proposed estimator reveals a small yet statistically significant negative effect at the 95\%  confidence level. This suggests a change in savings behavior after obtaining insurance. The findings in Panel C, consistent in magnitude with those from \cite{cai_insurance_2016}, demonstrate that weather insurance significantly increases the proportion of flexible-term savings. This suggests that households prefer to shift their savings to more flexible instruments after acquiring insurance. Visually, results in Panel C affirm the previous findings in \cite{cai_insurance_2016}, though we emphasize the enhanced precision our method provides. For example, the confidence intervals associated with $\widehat{\beta}_e$ are 16.29\%, 85.53\%, 64.21\%, and 115.84\% wider than the confidence intervals associated with $\widehat{ES}(e)$ at event times $e=2,3,4,5$, respectively. This illustrates that our DDD procedure can provide much more powerful analysis than alternative procedures.

\subsection{Effect of emission trading scheme on carbon emissions}

We next revisit the study by \cite{carbon_pricing}, which examines the impact of the emission trading scheme (ETS) implemented by the Chinese government to regulate and reduce carbon emissions at minimal cost, focusing on its effect in promoting low-carbon innovation among firms. China tackled its dual objectives of economic growth and reducing carbon emissions by approving seven regional ETS pilot programs in 2011. These pilot programs spanned four major cities, two provinces, and one special economic zone, each developing their own allowance and enforcement regulations within national guidelines, to peak national carbon emissions by around 2030. 

\cite{carbon_pricing} leverages three sources of variation. First, it considers firms' patent applications before and after the introduction of ETS pilots. Second, the pilots were launched in different regions whereas some regions are not impacted by the introduction of an ETS. Third, the pilots covered various manufacturing sectors within these regions such as power and heating, chemical, cement etc. Therefore, the authors utilize these variations in ETS pilots over time, across sectors, and across regions, employing a DDD approach to identify the effects of ETS on low-carbon innovation.

This setup exemplifies a DDD model with variation in treatment timing and covariates. \cite{carbon_pricing} uses the year of announcement (i.e., 2011) rather than the actual launch year as the post-policy indicator. The regional carbon market pilots were phased in starting from 2013, with initial launches in Shenzhen, Shanghai, Beijing, Guangdong, and Tianjin, followed by the introduction in Hubei and Chongqing the subsequent year. Thus, we adjust the construction of the treatment variable to use the official launch date instead of the announcement date, allowing for staggered treatment adoption.\footnote{We also considered the original specification using treatment announcement as the definition of the treatment. In such a case, the results obtained using our DR DDD procedure and those in \cite{carbon_pricing} are very close; these results are available upon request. Thus, we stress that the empirical results in this section should be interpreted as an extension of \cite{carbon_pricing}, and not a ``replication''.}

The dataset utilized in \cite{carbon_pricing} includes information on publicly listed firms in China from 2003 to 2015, sourced from the Shanghai and Shenzhen stock exchanges. This dataset comprises both financial data, obtained through the China Stock Market and Accounting Research (CSMAR), and patent application details, collected from China's State Intellectual Property Office (SIPO). The dataset includes $18,937$ firm-year observations from $1,956$ firms spanning the years 2003 to 2015. The three outcomes evaluated by \cite{carbon_pricing} are the number of low-carbon patent applications, the number of patent applications in other non-low-carbon technologies, and the ratio of low-carbon patents relative to the total number of patents. 

In terms of estimators, we use a 3WFE similar to \cite{carbon_pricing}\footnote{See, e.g., Equation (1) in \cite{carbon_pricing}}, for firm $i$ at year $t$
\begin{equation}
\label{model_3wfe_carbon}
    Y_{i,t} = \gamma_i + \gamma_{r,t} + \gamma_{s,t} + \beta_{3wfe} D_{i,t} + X'_{i,s,r}\theta + \varepsilon_{i,t},
\end{equation}
where $\beta_{3wfe}$ is meant to capture the impact of ETS on low-carbon technology innovation, $D_{i,t}$ is the treatment indicator post-announcement of ETS for covered sectors in pilot regions, $\gamma_i$ is a firm fixed effect, $\gamma_{r,t}$ and $\gamma_{s,t}$ are region-by-time and sector-group-by-time fixed effects, respectively. Covariates $X_{i,s,r}$ include firm-level attributes such as assets, revenue, and current liabilities. As we are also interested in treatment effect dynamics, we also consider a 3WFE event-study specification 
\begin{equation}
    \label{model_es_carbon}
Y_{i,t} = \gamma_i + \gamma_{r,t} + \gamma_{s,t} +\sum_{e \neq-1} \beta_{e} \mathbf{1}\{E_{i,t}=e\} + X'_{i,s,r}\theta + u_{i,t},
\end{equation}
where $E_{i,t} = t - G_i$ denotes the time relative to when the regional ETS pilot was initiated for firm $i$, and $u_{i,t}$ is an idiosyncratic error term. In this setup, a firm $i$ is considered to have received the treatment if it is located in a region that is a carbon market pilot and its sector qualifies as a regulated sector by the ETS. We compare these 3WFE estimators with our DR DDD estimator \eqref{eqn:es_estimator}. In our setup, $S$ represents the year a region launched the ETS pilot, and $Q$ represents the manufacturing sector that is eligible for ETS.

\begin{figure}[!htp]
  \centering
  \begin{subfigure}[t]{0.75\textwidth}
    \centering
      \includegraphics[width=\linewidth]{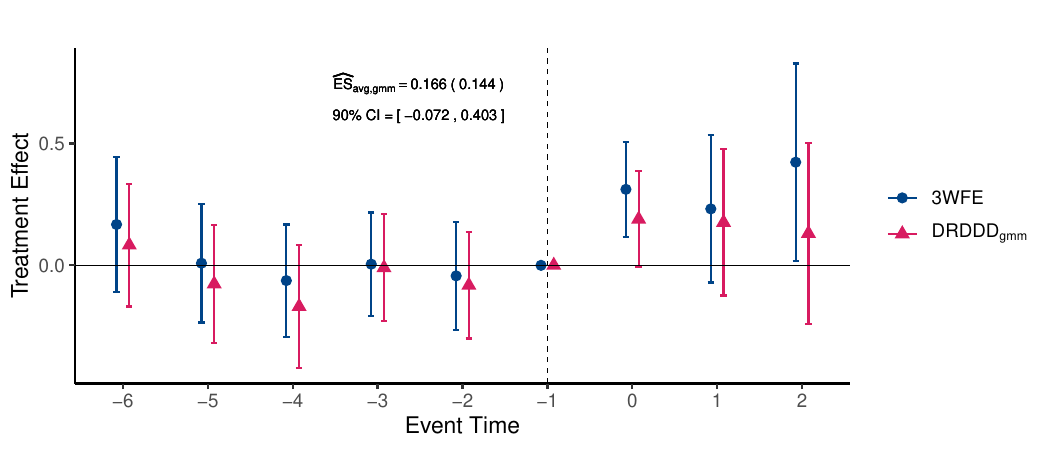}
                \caption{Panel A: Emission Trading Scheme regulation on low-carbon patents}
  \end{subfigure}

  \begin{subfigure}[t]{0.75\textwidth}
    \centering
      \includegraphics[width=\linewidth]{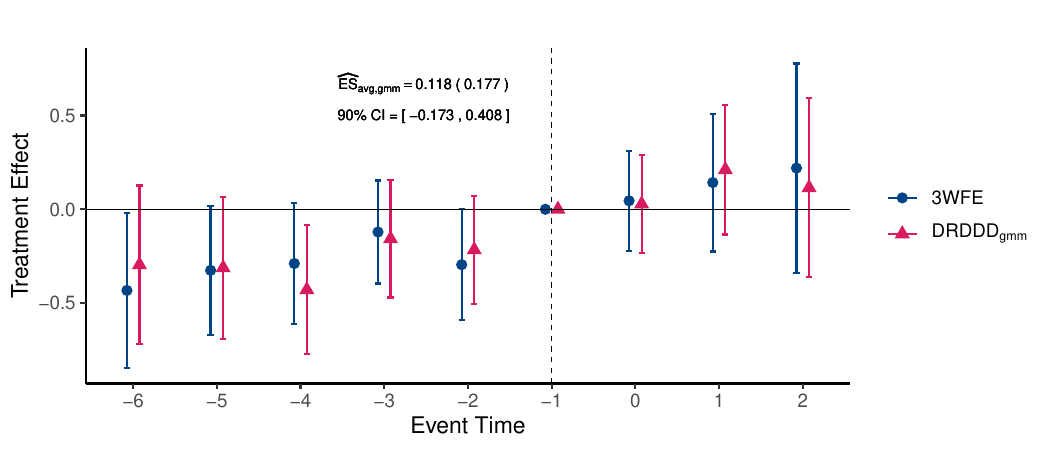}
               \caption{Panel B: Emission Trading Scheme regulation on other non-low-carbon patents}
  \end{subfigure}

  \begin{subfigure}[t]{0.75\textwidth}
    \centering
      \includegraphics[width=\linewidth]{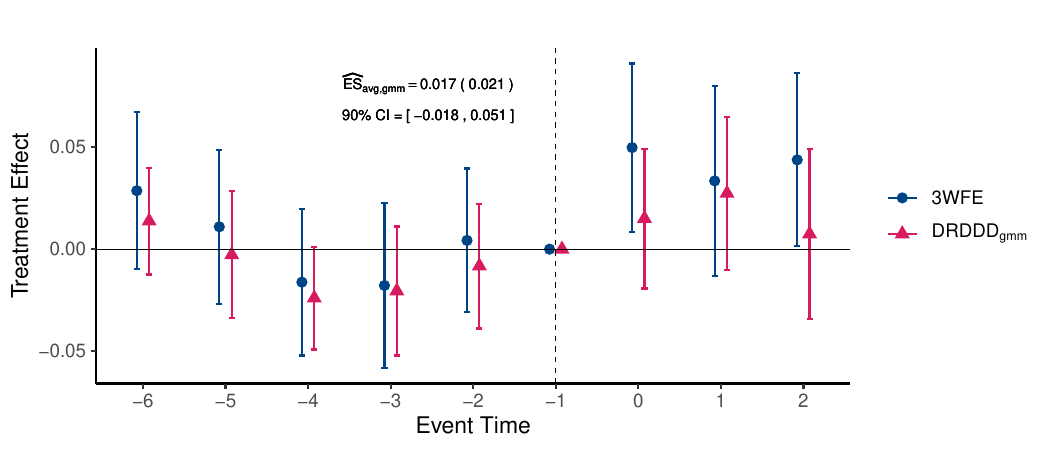}
                \caption{Panel C: Emission Trading Scheme regulation on share of low-carbon patents}
  \end{subfigure}

    \caption{Impact of ETS on share of low-carbon patents.}
    \label{fig:carbon_pricing_app}
        \justifying
	\noindent \scriptsize{Notes: This figure displays event study estimates incorporating firm characteristics such as total assets, total revenue, and current liabilities as covariates. Panel A examines the number of low-carbon patents, Panel B considers other non-carbon patents, and Panel C uses the proportion of low-carbon patents relative to the total number of patents. In each panel, blue dots represent the point estimates from the specification in \eqref{model_es_carbon}, while red triangles show the $\widehat{ES}_{\text{dr},\text{gmm}}(e)$ estimates aggregated from our proposed DR estimator as in \eqref{eqn:es_estimator}. The average $\widehat{ES}_{\text{avg}, \text{gmm}}$, calculated from all $\widehat{ES}_{\text{dr},\text{gmm}}(e)$ values to the right of the vertical dashed line at \(e = -1\), is provided alongside its bootstrapped standard errors, based on 999 repetitions, and confidence interval. The corresponding blue and red vertical lines represent the 90\% confidence intervals.
 }
\end{figure}


Figure \ref{fig:carbon_pricing_app} compares the event study estimates from the specification \eqref{model_es_carbon} with the estimates $\widehat{ES}_{\text{dr},\text{gmm}}(e)$, as specified in \eqref{eqn:es_estimator}, using our DR procedure. Consistent with \cite{carbon_pricing}, we present results at a 90\% confidence level. 
In Panel A, as reported by \cite{carbon_pricing}, the ETS pilot appears to encourage innovation, but we cannot reject the null hypothesis of no effect at the same confidence level. Panel B's results align with \cite{carbon_pricing}, indicating minimal evidence for the ETS's crowding-out effect on other patent applications. The findings in Panel C show no significant impact of ETS regulation on the proportion of low-carbon patents. This contrasts with the results associated with a 3WFE, which indicated a positive and significant effect at a 90\% confidence level\footnote{See Table 1 in \cite{carbon_pricing}}. For instance, our method estimates an aggregate effect of 1.7\% without sufficient evidence to reject the null hypothesis of no effect, while specification \eqref{model_3wfe_carbon} estimates a significant aggregate effect of 5.3\% of ETS adoption on share of low-carbon patents. 

Overall, our findings suggest there is no evidence to support that the ETS promotes innovation in low-carbon patents. Additionally, our tighter confidence intervals reflects that DR DDD can provide practical gains in power in empirical applications. For example, in Panel C of Figure \ref{fig:carbon_pricing_app}, the confidence intervals associated with $\widehat{\beta}_e$ are 20.28\%, 24.26\% and 1.51\% wider than the confidence intervals associated with $\widehat{ES}_{\text{dr},\text{gmm}}(e)$ at event times $e=0,1,2$, respectively. This implies that $\widehat{\beta}_e$ would need a sample with 44\%, 54\%, and 3\%  more observations than our DR DDD event study estimator for event times zero, one, and two to achieve the same precision as DR DDD.\footnote{These calculations are based on asymptotic relative efficiency. For any parameter $\eta$ of a distribution $F$, and for estimators $\widehat{\eta}_{1}$ and $\widehat{\eta }_{2}$ approximately $N\left(  \eta,V_{1}/n\right)  $ and $N\left(  \eta, V_{2}/n\right)  $, respectively, the asymptotic relative efficiency of $\widehat{\eta}_{2}$ with respect to $\widehat{\eta}_{1}$ is given by $V_{1}/V_{2}$; see, e.g., Section 8.2 in \cite{VanderVaart1998}.}


\subsection{Effect of genetically modified crops on yields}


Finally, we revisit the study by \cite{hansen_national_2023} on the impact of genetically modified (GM) crop adoption on countrywide yields. While farmers in North and South America and parts of Asia adopted these seeds almost immediately, regulators in the European Union, much of Africa, and several middle‑income economies imposed outright bans or stringent de-facto restrictions, often in response to food-safety scares and NGO pressure.  By 2019 only 29 countries permitted commercial cultivation, another 42 permitted imports but not planting, and the rest maintained comprehensive bans.  This sharp global policy divergence, paired with the agronomic premise that pest‑ and herbicide-tolerant GM varieties raise yields mainly where weeds and insects are prevalent, sets the stage for a compelling cross‑country analysis of aggregate production impacts.

\cite{hansen_national_2023} exploits the fact that GM technologies were adopted at different times in various countries starting in the mid-1990s to examine the effects of GM adoption on outcomes such as crop yields, harvested area, and trade flows, using balanced crop-country panel data covering over 120 countries from 1986 to 2019. Different sources of variation underpin the DDD design in this context.  First, there was a variation in policy timing as each country passed GM cultivation legislation in different years.  This allows one to compare outcomes before and after GM commercialization is allowed, between countries that have already passed these legislations and those that have not yet.
Second, crop eligibility for GM traits exists for only four field crops---cotton, maize, soybean, and rapeseed--- leaving rice, wheat, and dozens of other staples as the primary comparisons.
Differencing outcomes over time, across crops, and across policy regimes therefore allows for common shocks (e.g., weather, prices) and crop‑specific global trends, isolating the causal effect of GM adoption.


\cite{hansen_national_2023} utilizes publicly available country-level data on production quantities from FAOSTAT, covering yields (production per hectare), producer prices, and trade flows (exports and imports by commodity). Following their setting, our focus is also on countries with a minimum of $100,000$ hectares of cropland, thereby drawing attention to nations with significant agricultural activities. Similarly, we also include 60 field crops, including cereals, pulses, roots and tubers, oil crops, and fiber crops, as outlined by the FAO \citep{FAO2012}. \cite{hansen_national_2023} collected official legislative documents, USDA Foreign Agricultural Service reports, and industry publications to determine the earliest legal and feasible year of commercial cultivation for each country-crop combination.\footnote{For further details, refer to the documentation for GM approval dates in \cite{hansen_national_2023} Online Appendix B.} In some cases, GM approval did not directly result in large-scale planting, while in others, formal bans restricted or effectively prohibited cultivation. Thus, following \cite{hansen_national_2023}, we establish $G_{i,c}$ as the first year in which a GM variety of a given crop in a given country could be harvested and sold commercially for human consumption or animal feed without violating a ban.\footnote{\cite{hansen_national_2023} uses a different notation than us, and refers to our $G_{i,c}$ as $E_{ic}$.}

\cite{hansen_national_2023} considers the following event-study specification\footnote{See, e.g., Equation (1) in \cite{hansen_national_2023}.} to asses the dynamic average effect of GM adoption on countrywide yields
\begin{equation}
\label{model_es_crops}
    \ln Y_{i,c,t} = \gamma_{c,i} + \gamma_{i,t} + \gamma_{c,t} + \sum_{e\not=-1} {\beta}_{e} \cdot \mathbf{1}\{E_{i,c,t} = e\} + \varepsilon_{i,c,t}  
\end{equation}
where $Y_{i,c,t}$ denotes the yield of crop $c$ in country $i$ for the year $t$. The variable $E_{i,c,t} = t - G_{i,c}$ represents the time since a GM variety of a crop became legally permissible for harvest and sale for human consumption, and $\varepsilon_{i,c,t}$ is an idiosyncratic error error. The specification in \eqref{model_es_crops} includes three fixed effects: country-by-year fixed effects ($\gamma_{i,t}$) to account for shocks affecting all crops in a particular country and year, crop-by-country fixed effects ($\gamma_{c,i}$) to capture local time-invariant conditions for each crop, and crop-by-year fixed effects ($\gamma_{c,t}$) to reflect global trends specific to each crop. We compare event-study estimates based on \eqref{model_es_crops} with our proposed DR DDD-based event study in \eqref{eqn:es_estimator}. In our setup, $S$ represents the time a country allowed for GM commercialization, and $Q$ represents the eligible crop.

Figure \ref{fig:gmcrops} illustrates the event study estimates of the impact of GM crop adoption on yields of cotton, maize, rapeseed, and soybean together. Panel A shows results aligned with those in \cite{hansen_national_2023}, with a slight enhancement in the precision of the confidence intervals. For instance, at event times $e = 0$, the confidence intervals associated with $\widehat{\beta}_{e}$ is  46.47\% wider than the one associated with $\widehat{ES}_{\text{dr}, \text{gmm}}(e)$. Both methods show that approving GM technology leads to a significant increase in crop yields by approximately 13\%. It is not surprising that both methods yield similar results in this context. Despite the challenges associated with 3WFE specifications with variation in treatment timing \citep{ strezhnev2023}, there exists a very large number of never-enabling crop-country combinations, making these ``forbidden comparisons'' less likely to shift results. Indeed, as noted by \cite{hansen_national_2023}, only about 2 percent of the units are treated during the period of analysis.

\begin{figure}[hptb]
  \centering
  \begin{subfigure}[t]{0.8\textwidth}
    \centering
      \includegraphics[width=\linewidth]{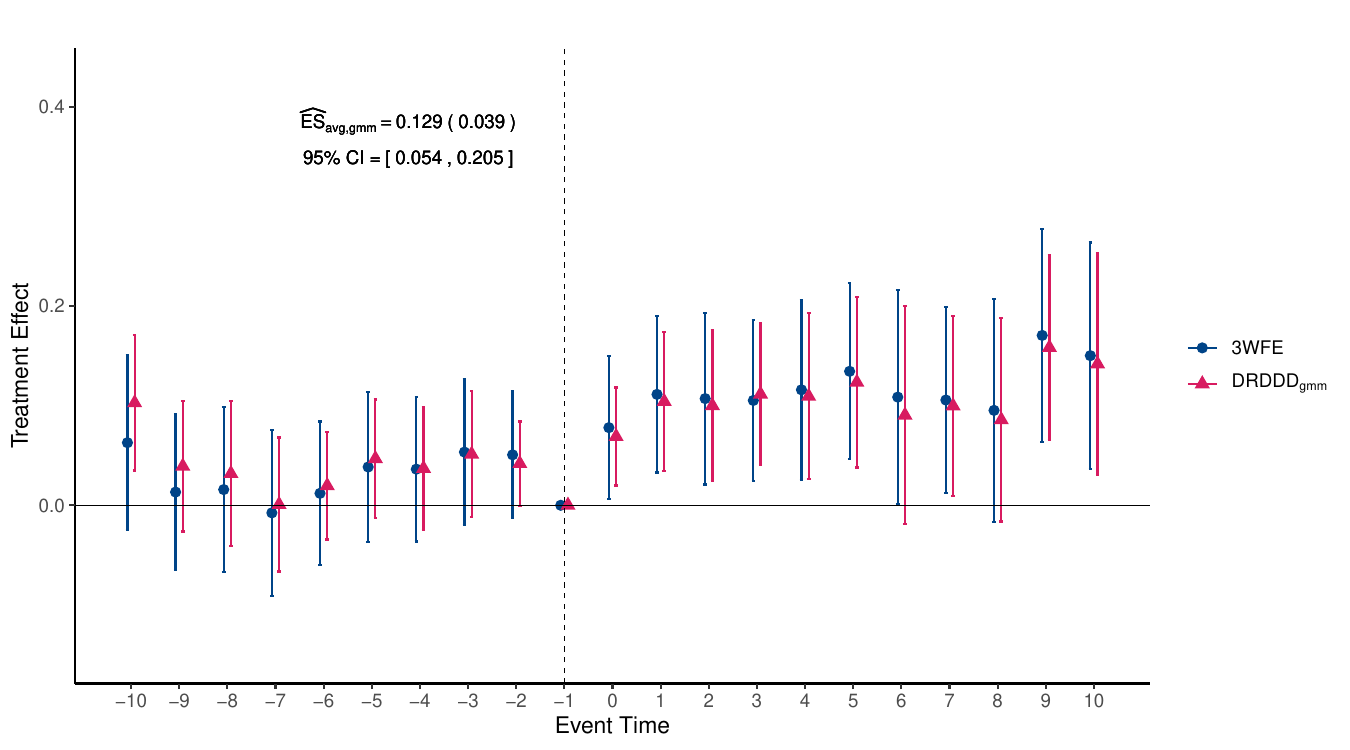}
                \caption{Panel A: GM adoption on log-yields (full sample)}
  \end{subfigure}
 
  \begin{subfigure}[t]{0.8\textwidth}
    \centering
      \includegraphics[width=\linewidth]{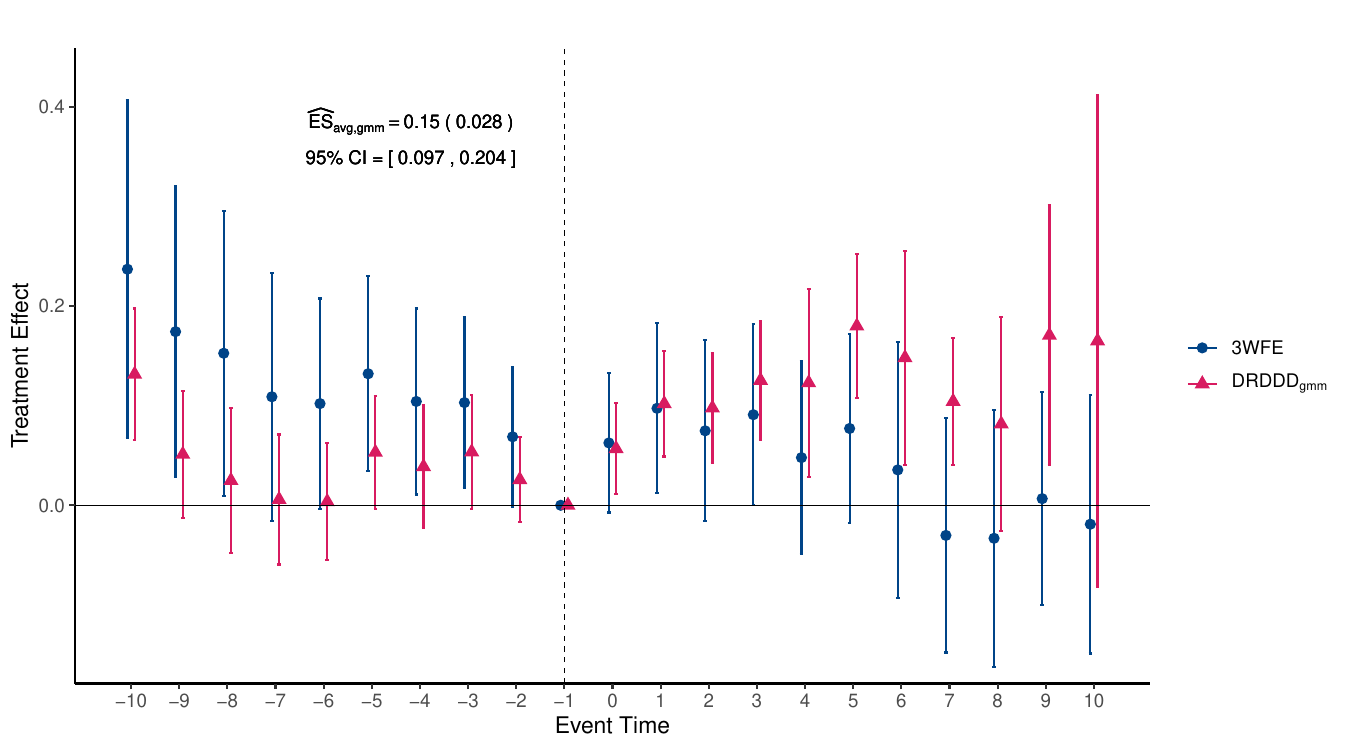}
               \caption{Panel B: GM adoption on log-yields (dropping all never-treated)}
  \end{subfigure}

    \caption{Impact of GM adoption on yields}
    \captionsetup{justification=justified}
    \label{fig:gmcrops}
        \justifying
	\noindent \scriptsize{Notes: This figure shows event study estimates using the logarithm of crop yields as the outcome variable, without any covariates. The estimation window spans from 1986 to 2019, excluding country-crop combinations treated in 2010 or later. Panel A considers a fully balanced sample, as in \cite{hansen_national_2023}, whereas Panel B excludes units never treated and forces the last cohort to be considered never-treated, with data post-treatment being filtered out. In each panel, blue dots represent point estimates from the specification in equation \eqref{model_es_crops}, while red triangles show the $\widehat{ES}_{\text{dr},\text{gmm}}(e)$ estimates aggregated from our proposed DR estimator as in \eqref{eqn:es_estimator}. The average $\widehat{ES}_{\text{avg},\text{gmm}}$, calculated from all $\widehat{ES}_{\text{dr},\text{gmm}}(e)$ values to the right of the vertical dashed line at \(e = -1\), is provided alongside its bootstrapped standard errors, based on 999 repetitions, clustered at the country-crop level and its confidence interval. The corresponding blue and red vertical lines represent the 95\% confidence intervals.
 }
\end{figure}


Having said that, a potential concern that can arise is related to whether it is indeed desirable to use these never-enabling crop-country combinations as a ``de facto'' comparison group. It may be the case that countries that never allow harvesting and commercialization of GM crops during a very long period of analysis may already face lower yields or stricter farm rules, so they might not be a necessarily good comparison group for eventually treated crop-countries observations.

To assess the robustness of the findings in Panel A of Figure \ref{fig:gmcrops}, we drop data for every truly never-enabling unit, and re-run the 3WFE regression specification \eqref{model_es_crops} in this subsample as well as our staggered DDD estimator in \eqref{eqn:es_estimator}. As Panel B in Figure \ref{fig:gmcrops} highlights, this has important consequences for the analysis: if one were to rely on \eqref{model_es_crops}, one would find much smaller effects compared to those in Panel A, with non-negligible pre-treatment trends that may cast doubt about the plausibility of the DDD design. On the other hand, when using our proposed staggered DDD estimator, we can see that the results remain positive and statistically significant for most event-times, with no very serious violation of pre-treatment trends. In fact, our estimates for $ES_{\text{avg}}$ indicate an approximate 15\% increase in crop yields following GM adoption, a slightly larger (but statistically distinguishable) than those in Panel A. Of course, statistical precision is less accurate in Panel B than in Panel A, as we use a much smaller sample size.

This result essentially highlights that the main findings in \cite{hansen_national_2023} are robust to dropping the never-enabling set of units from the comparison group, as long as you use DDD procedures that are meant to work well in such setups.

\section{Concluding remarks}
\label{conclusion}

This paper studied DDD estimators, paying close attention to situations where covariates are important for identification and to setups with staggered treatment adoption. Our findings challenge the conventional wisdom that DDD can be understood as the difference between two DiDs. We showed that when DDD-type parallel trends hold after conditioning on covariates, DDD estimators cannot generally be expressed as such, even in cases with only two time periods. In addition, when treatment adoption is staggered, pooling all not-yet-treated units is not generally valid, and proceeding as such can lead to misleading conclusions even when covariates are not crucial for the DDD identification arguments. These results highlight the need for more careful consideration when applying DDD strategies. 

To address these challenges, we proposed DR DDD estimators that can appropriately handle covariates and can also be used in DDD setups with staggered treatment adoption. Importantly, we proposed a DR DDD estimator that leverages information across different comparison groups, and our simulation results highlighted that the gains in precision can be substantial compared to alternatives. As such, we recommend practitioners to favor our proposed DDD estimator $\widehat{ATT}_{\text{dr},\text{opt}}(g,t)$ in applications. Finally, a companion \texttt{R} package, \href{https://marcelortiz.com/triplediff/}{\texttt{triplediff}}, is freely available on GitHub to automate all these DDD estimators proposed in this article, simplifying their adoption for practitioners. 

We envision extending our proposed DDD tools to incorporate data-adaptive and machine-learning estimators for the nuisance parameters, with a particular emphasis on training these models to ensure robust performance in finite samples. Additionally, we aim to generalize our DDD framework to accommodate more flexible sampling schemes, such as unbalanced panel data or repeated cross-sectional data; see, for example, \citet{Abadie2005}, \citet{Callaway2021}, and \citet{SantAnna2023} for related developments in DiD settings. Other promising avenues for future research include exploring semiparametric efficiency bounds and efficient estimation in overidentified DDD models \citep{Chen_SantAnna_Xie_effientDiD_2025}, as well as using a DDD strategy to uncover persuasion effects \citep{Jun_Lee_2024_DiD_persuasion}. We leave a full treatment of these topics to future work.

\small{
\setlength{\bibsep}{1pt plus 0.3ex}
\putbib
}
\end{bibunit}

\newpage
\begin{bibunit}
\appendix
\newgeometry{margin=2.5cm}
\clearpage
\renewcommand\thefigure{OA-\arabic{figure}}
\renewcommand\thetable{OA-\arabic{table}}
\renewcommand*{\thepage}{OA - \arabic{page}}
\renewcommand\thesection{\Alph{section}}
\renewcommand\thesubsection{\Alph{section}.\arabic{subsection}}



\setcounter{figure}{0}
\setcounter{table}{0}
\setcounter{page}{1}

\begin{center}
	{\Large{Better Understanding Triple Differences Estimators}:\\Supplemental Appendix}\\[1em]
 Marcelo Ortiz-Villavicencio ~~~~~~~~~~ Pedro H.C. Sant'Anna \\[0.5em]

        \date{\today} \\[1em]
\end{center}







This Supplementary Appendix includes: (a) proofs for the results presented in the main paper; and (b) details about the data-generating process (DGP) used in the Monte Carlo simulations to illustrate the finite sample properties of our DR DDD method.

\section{Proofs of Main Results}

We begin by proving auxiliary lemmas that are foundational for establishing the main results of the paper.  Initially, consider the conditional $ATT(g,t)$, and for simplicity, we omit the unit indexing $i$.
\begin{equation*}
    CATT_{X}(g,t) \equiv \E[Y_{t}(g) - Y_{t}(\infty) | X, S =g, Q = 1] 
\end{equation*}

\begin{lemma}
\label{lemma:catt}
    Let Assumptions \ref{ass:sampling_panel}, \ref{ass:overlap_staggered}, \ref{ass:anticipation}, and \ref{ass:PT-NYT} hold. Then, for all $g\in \mathcal{G}_{\text{trt}}$, $t\in \{2,\dots, T\}$, and $g_{\text{c}} \in \mathcal{S}$ such that $t\ge g$ and $g_{\text{c}}>t$,
    \begin{align*}
        CATT_{X}(g,t) &= \left(\E[Y_t - Y_{g-1}| X, S=g, Q=1] - \E[Y_t - Y_{g-1} | X, S=g, Q=0]\right) \\
        &~~~~-  \left(\E[Y_t - Y_{g-1} |X, S=g_c, Q=1] - \E[Y_t - Y_{g-1} |X, S=g_c, Q=0]\right) ~\text{almost surely.}
    \end{align*}
\end{lemma}
\begin{proof}
     All equalities below are understood to hold almost surely (a.s.), conditioning on $X$ throughout.
     \begin{align*}
         CATT_{X}(g,t) &= \E[Y_{t}(g) - Y_{g-1}(\infty) | X, S =g, Q = 1] - \E[Y_{t}(\infty) - Y_{g-1}(\infty) | X, S =g, Q = 1] \\
         &= \E[Y_{t}(g) - Y_{g-1}(\infty) | X, S =g, Q = 1] - \E[Y_t(\infty) - Y_{g-1}(\infty) | X, S=g, Q=0] \\
         &~~~~~-\E[Y_t(\infty) - Y_{g-1}(\infty) | X, S=g', Q=1] + \E[Y_t(\infty) - Y_{g-1}(\infty) | X, S=g', Q=0]\\
         &= \E[Y_{t} - Y_{g-1} | X, S =g, Q = 1] - \E[Y_t - Y_{g-1} | X, S=g, Q=0] \\
         &~~~~~-\E[Y_t - Y_{g-1} | X, S=g_c, Q=1] + \E[Y_t - Y_{g-1} | X, S=g_c, Q=0],
     \end{align*}
     where the first equality is derived by algebraically adding and subtracting the term $\E[Y_{g-1}(\infty) | X, S =g, Q = 1]$. The second equality is obtained based on Assumption \ref{ass:PT-NYT}, and the final equality arises from the definition in \eqref{def:observed_outcome} along with Assumption \ref{ass:anticipation}.
\end{proof}

Consider the population regression function representing the changes in outcomes from period $t$ to $t'$ given the covariates as $m_{Y_t - Y_{t'}}^{S=g, Q=q}(X) = \E[Y_t - Y_{t'} \mid S=g, Q=q, X]$.  Given the weights defined in equation \eqref{eqn:weights}, we then proceed to derive the efficient influence function in the following lemma. For notational simplicity, assume the simplest scenario with two time periods $t \in \{1,2\}$, two treatment groups $S_i \in \{2, \infty\}$, and two eligibility groups $Q_i \in \{0, 1\}$, thus we derive the efficient influence function for the $ATT(2,2)$ in the following lemma.

\begin{lemma}[Efficient Influence Function]
\label{lemma:eif}
Suppose Assumptions \ref{ass:sampling_panel}, \ref{ass:overlap_staggered}, \ref{ass:anticipation}, and \ref{ass:PT-NYT} hold. Let $\tau(X) \equiv Y_2 - Y_1 - m_{Y_2 - Y_1}^{S=2, Q=0}(X) - m_{Y_2 - Y_1}^{S=\infty, Q=1}(X) + m_{Y_2 - Y_1}^{S=\infty, Q=0}(X) $ and the observed data $O \equiv (Y_2, Y_1, S,Q,X)$. Then, the efficient influence function for the $ATT(2,2)$ is given by
\begin{align*}
    \eta_{eff}(O) &= w^{S=2, Q=1}_{\text{trt}}(S,Q) \cdot \left(\tau(X) - ATT(2,2) \right) \\
    &-  w^{S=2, Q=0}_{\text{comp}}(S,Q,X) \cdot \left(Y_2 - Y_1  - m_{Y_2 - Y_1}^{S=2, Q=0}(X)\right)\\
    &- w^{S=\infty, Q=1}_{\text{comp}}(S,Q,X) \cdot \left(Y_2 - Y_1 - m_{Y_2 - Y_1}^{S=\infty, Q=1}(X)\right)\\
    &+ w^{S=\infty, Q=0}_{\text{comp}}(S,Q,X) \cdot \left(Y_2 - Y_1  - m_{Y_2 - Y_1}^{S=\infty, Q=0}(X)\right),
\end{align*}
where 
\begin{small}
\begin{align*}
    w^{S=2,Q=1}_{\text{trt}}(S, Q)\equiv\dfrac{1{\{S=2,Q=1\}} }{\mathbb{E}[1{\{S=2,Q=1\}}]}, \quad w^{S=g, Q=q}_{\text{comp}}(S, Q, X) \equiv \dfrac{\dfrac{1{\{S=g,Q=q\}} \cdot p^{S=2,Q=1}(X) }{p^{S=g,Q=q}(X)}} {\mathbb{E}\bracks{\dfrac{1{\{S=g,Q=q\}} \cdot p^{S=2,Q=1}(X) }{p^{S=g,Q=q}(X)}}}.
\end{align*}
\end{small}
\end{lemma}

\begin{proof}
    We follow \cite{Hahn1998} and \cite{Chen2008808} to structure a efficient influence function as in \cite{newey1990}. Let the conditional density $f(y_2, y_1 \mid g,q,x) = f(Y_2 = y_2, Y_1 = y_1 \mid S=g, Q=q, X=x)$. The proof is organized into three steps: (i) First, we define the tangent space of the statistical model. (ii) Next, we demonstrate that the target parameter related to the parametric sub-model is pathwise differentiable. (iii) Finally, we confirm that all conditions specified in \cite{newey1990} are met to ensure the derived influence function is efficient.

\textsc{Step 1:} The density (likelihood) of the observed data $O \in \R^{2}\times \mathcal{S} \times \{0,1\} \times \R^{d}$ can be represented by
\begin{align*}
f(y_2, y_1, g, q, x)= & f(x) \times \left[f(y_2, y_1 \mid 2,1, x)\cdot p^{S=2, Q=1}(x)\right]^{\bar{g}\bar{q}} \times \left[f(y_2, y_1 \mid 2,0, x)\cdot  p^{S=2, Q=0}(x) \right]^{\bar{g}(1 - \bar{q})}  \\
& \times \left[f(y_2, y_1 \mid \infty,1, x) \cdot p^{S=\infty, Q=1}(x) \right]^{(1-\bar{g}) \bar{q}} \\
& \times \left[f(y_2, y_1 \mid \infty,0, x) \cdot  p^{S=\infty, Q=0}(x) \right]^{(1-\bar{g})(1-\bar{q})}
\end{align*}
where $\bar{g} = 1\{S=2\}$ and $\bar{q} = 1\{Q=1\}$. Consider the regular parametric submodel, with the true model indexed by $\theta = 0$,
\begin{align*}
f_{\theta}(y_2, y_1, g, q, x)= & f_{\theta}(x) \times \left[f_{\theta}(y_2, y_1 \mid 2,1, x)\cdot p_{\theta}^{S=2, Q=1}(x)\right]^{\bar{g}\bar{q}} \times \left[f_{\theta}(y_2, y_1 \mid 2,0, x)\cdot  p_{\theta}^{S=2, Q=0}(x) \right]^{\bar{g}(1 - \bar{q})}  \\
& \times \left[f_{\theta}(y_2, y_1 \mid \infty,1, x) \cdot p_{\theta}^{S=\infty, Q=1}(x) \right]^{(1-\bar{g}) \bar{q}} \\
& \times \left[f_{\theta}(y_2, y_1 \mid \infty,0, x) \cdot  p_{\theta}^{S=\infty, Q=0}(x) \right]^{(1-\bar{g})(1-\bar{q})}
\end{align*}
The corresponding score function of this submodel for all $(g,q) \in \mathcal{S} \times \{0,1\}$ is given by
\begin{align*}
s_\theta(y_2, y_1, g, \ell, x)= & \bar{g} \bar{q} s_\theta(y_2, y_1 \mid 2,1, x)+ \bar{g} \bar{q} \frac{\dot{p}^{S=2, Q=1}_\theta(x)}{p^{S=2, Q=1}_\theta(x)} + \bar{g}(1-\bar{q}) s_\theta(y_2, y_1 \mid 2,0, x) + \bar{g}(1-\bar{q}) \frac{\dot{p}^{S=2, Q=0}_\theta(x)}{p^{S=2, Q=0}_\theta(x)}\\
&+(1-\bar{g}) \bar{q} s_\theta(y_2, y_1 \mid \infty,1, x)+ (1-\bar{g}) \bar{q} \frac{\dot{p}^{S=\infty, Q=1}_\theta(x)}{p^{S=\infty, Q=1}_\theta(x)} \\
&+ (1-\bar{g})(1-\bar{q}) s_\theta(y_2, y_1 \mid \infty,0, x) + (1-\bar{g})(1-\bar{q}) \frac{\dot{p}^{S=\infty, Q=0}_\theta(x)}{p^{S=\infty, Q=0}_\theta(x)} \\
& +t_\theta(x),
\end{align*}
where
\begin{align*}
    s_{\theta}\left(y_2, y_1 \mid g,q, x\right)=\frac{\partial}{\partial \theta} \log f_{ \theta}\left(y_2, y_1 \mid g, q, x\right), \quad \dot{p}^{S=g, Q=q}_\theta(x)=\frac{\partial}{\partial \theta} p^{S=q, Q=q}_\theta(x), \text { and }  t_\theta(x)=\frac{\partial}{\partial \theta} \log f_\theta(x).
\end{align*}
Therefore, the tangent space is the characterized by 
\begin{align*}
\mathcal{T} &= \big\{\bar{g} \bar{q} s_{2,1}(y_2, y_1, x)+ \bar{g} \bar{q} p_{2,1}(x)+\bar{g}(1-\bar{q}) s_{2,0}(y_2, y_1, x)+\bar{g}(1-\bar{q}) p_{2,0}(x) \\
&+ (1-\bar{g}) \bar{q} s_{\infty,1}(y_2, y_1, x)+ (1-\bar{g}) \bar{q} p_{\infty,1}(x) \\
&+ (1-\bar{g})(1-\bar{q}) s_{\infty,0}(y_2, y_1, x)+ (1-\bar{g})(1-\bar{q}) p_{\infty,0}(x)+l(x)\big\},
\end{align*}
where $p_{g,q}(\cdot)$ and $l(\cdot)$ are the scores of propensity scores and marginal density of $x$, respectively. Furthermore, for any functions $\left\{s_{g ,q}(\cdot, \cdot, \cdot), p_{g, q}(\cdot)\right\}_{(g, q) \in \mathcal{S} \times \{0,1\}}$, and $l(\cdot)$ such that, for every $(g, q) \in \mathcal{S} \times \{0,1\}$,
\begin{align}
\label{condition_newey1}
& s_{g ,q}(\cdot, \cdot, \cdot) \in L_2(\mathcal{Y}_2 \times \mathcal{Y}_1 \times \mathcal{X}), \text { with } \int \int s_{g, q}(y_2,y_1, x) f(y_2, y_1 \mid g, q, x) d y_1 d y_2=0, \forall x \in \mathcal{X}, \\
& p_{g,q}(\cdot) \in L_2(\mathcal{X}), \text { with } \sum_{(g, q) \in \mathcal{S} \times \{0,1\}} \int p_{g,q}(x) f(x) d x=0 \label{condition_newey2}\\
&l(\cdot) \in L_2(\mathcal{X}), \text { with } \int l(x) f(x) d x=0 \label{condition_newey3}.
\end{align}
\textsc{Step 2:} Given Lemma \ref{lemma:catt}, $ATT(2,2)$ can be identified as
\begin{align*}
    ATT(2,2) &= \E\bigg[\big( \E[Y_2 - Y_1 \mid S=2, Q=1, X] - \E[Y_2 - Y_1 \mid S=2, Q=0, X]\big) \nonumber\\ 
    &~~~~- \big(\E[Y_2 - Y_1 \mid S=\infty, Q=1, X] - \E[Y_2 - Y_1 \mid S=\infty, Q=0, X]\big) \bigg\rvert S=2, Q=1 \bigg]
\end{align*}
In order to show that $ATT(2,2)$ is pathwise differentiable within the parameterized submodel, let's denote $\tau \equiv ATT(2,2)$ and define:
\begin{small}
    \begin{align*}
\tau(\theta)&= \frac{\left(\iiint \Delta y p^{S=2, Q=1}_\theta(x) f_\theta(y_2, y_1 \mid 2,1, x) f_\theta(x) d y_2 d y_1 d x-\iiint \Delta y p^{S=2,Q=1}_\theta(x) f_\theta(y_2, y_1 \mid 2,0, x) f_\theta(x) d y_2 d y_1 d x\right)}{\int p^{S=2,Q=1}_\theta(x) f_\theta(x) d x} \nonumber \\
&-\frac{\left(\iiint \Delta y p^{S=2, Q=1}_\theta(x) f_\theta(y_2, y_1 \mid \infty,1, x) f_\theta(x) d y_2 d y_1 d x-\iiint \Delta y p^{S=2, Q=1}_\theta(x)  f_\theta(y_2, y_1 \mid \infty,0, x) f_\theta(x) d y_2 d y_1 d x\right)}{\int p^{S=2, Q=1}_\theta(x) f_\theta(x) d x}
\end{align*}
\end{small}
where $\Delta y = y_2 - y_1$. We now need to compute the derivative of $\tau(\theta)$ with respect to $\theta$ and evaluate it at the point $\theta = 0$, specifically $d \tau(\theta) / d \theta \rvert_{\theta=0}$. For notational convenience, let
\begin{align*}
& N(\theta)=\sum_{(g, q) \in \mathcal{S} \times \{0,1\}}(-1)^{\bar{g}+\bar{q}} \iiint \Delta y p^{S=2, Q=1}_\theta(x) f_\theta\left(y_2, y_1 \mid g, q, x\right) f_\theta(x) d y_2 d y_1 d x \\
& D(\theta)=\int p^{S=2, Q=1}_\theta(x) f_\theta(x) d x
\end{align*}
where the factor $(-1)^{\bar{g}+\bar{q}}$ determines the sign corresponding to each pair $(g, q) \in \mathcal{S} \times \{0,1\}$. Thus, we have $\tau(\theta)=N(\theta)/D(\theta)$. By differentiating $\tau(\theta)$ with respect to $\theta$ using the quotient rule we obtain:
$$
\frac{d \tau(\theta)}{d \theta}=\frac{N^{\prime}(\theta) D(\theta)-N(\theta) D^{\prime}(\theta)}{[D(\theta)]^2}
$$
Evaluating at $\theta=0$ (and denoting quantities at $\theta=0$ by dropping the $\theta$ subscript, e.g. $p(1,1; x)=p_{\theta=0}(1,1; x)$ and $f(x)=f_{\theta=0}(x)$), we obtain
$$
\left.\frac{d \tau(\theta)}{d \theta}\right|_{\theta=0}=\frac{N^{\prime}(0)}{D(0)}-\tau(0) \frac{D^{\prime}(0)}{D(0)}
$$
In our notation, $D(0)=\int p^{S=2,Q=1}(x) f(x) d x=p^{S=2, Q=1}$ and $\tau=\tau(0)$. Both $N(\theta)$ and $D(\theta)$ involve integrals of functions that depend on $\theta$ via three components: the  propensity score $p^{S=2, Q=1}_\theta(x)$, the joint conditional density $f_\theta\left(y_2, y_1 \mid g, q, x\right)$, and the marginal density $f_\theta(x)$.

Under standard regularity conditions, we differentiate inside the integrals using the chain rule. In particular, the following identities will be used:
\begin{align}
    \left.\frac{\partial f_\theta\left(y_2, y_1 \mid g, q, x\right)}{\partial \theta}\right|_{\theta=0}&=s\left(y_2, y_1 \mid g, q, x\right) f\left(y_2, y_1 \mid g, q, x\right),\\
    \left.\frac{\partial p^{S=2,Q=1}_\theta(x)}{\partial \theta}\right|_{\theta=0}&=\dot{p}^{S=2, Q=1}(x),\\
    \left.\frac{\partial f_\theta(x)}{\partial \theta}\right|_{\theta=0}&=t(x) f(x).
\end{align}
We gather the contributions from each derivative of the joint conditional density across all possible combinations of $(g, q) \in \mathcal{S} \times \{0,1\}$,
\begin{align*}
    \sum_{(g, q) \in \mathcal{S} \times \{0,1\}}(-1)^{\bar{g}+\bar{q}} \iiint \Delta y p^{S=2, Q=1}(x) s\left(y_2, y_1 \mid g, q, x\right) f\left(y_2, y_1 \mid g, q, x\right) f(x) d y_2 d y_1 d x
\end{align*}
Similarly, we collect the contributions from each derivative of $p^{S=2, Q=1}(x)$,
\begin{align*}
    \sum_{(g, q) \in \mathcal{S} \times \{0,1\}}(-1)^{\bar{g}+\bar{q}} \iiint \Delta y \dot{p}^{S=2, Q=1}(x) f\left(y_2, y_1 \mid g, q, x\right) f(x) d y_2 d y_1 d x .
\end{align*}
and the contributions from each derivative of $f(x)$
\begin{align*}
    \sum_{(g, q) \in \mathcal{S} \times \{0,1\}}(-1)^{\bar{g}+\bar{q}} \iiint \Delta y p^{S=2, Q=1}(x) f\left(y_2, y_1 \mid g, q, x\right) t(x) f(x) d y_2 d y_1 d x .
\end{align*}
Then,  when combining the contributions from the derivatives of $p^{S=2, Q=1}(x)$ and $f(x)$ in $N^{\prime}(0)$ and comparing with $D^{\prime}(0)$, one finds that subtracting by $\tau$ appears naturally. Thus, putting everything together, the derivative of $\tau(\theta)$ at $\theta=0$ becomes
\begin{align*}
\left.\frac{d \tau(\theta)}{d \theta}\right|_{\theta=0}&= \sum_{(g, q) \in \mathcal{S} \times \{0,1\}}(-1)^{\bar{g}+\bar{q}} \frac{\iiint \Delta y p^{S=2, Q=1}(x) s\left(y_2, y_1 \mid g, q, x\right) f\left(y_2, y_1 \mid g, q, x\right) f(x) d y_2 d y_1 d x}{p^{S=2, Q=1}} \\
&+\frac{\int(\tau(x)-\tau) \dot{p}^{S=2, Q=1}(x) f(x) d x}{p^{S=2, Q=1}} \\
&+\frac{\int(\tau(x)-\tau) p^{S=2, Q=1}(x) t(x) f(x) d x}{p^{S=2, Q=1}} .
\end{align*}

A parameter is considered pathwise differentiable if its derivative, evaluated along any smooth parametric submodel, can be expressed as the covariance between a suitable function and the score of the model. Let 
\begin{align*}
F_\tau(O)= & \frac{\bar{g} \bar{q} \left\{\Delta y-m_{\Delta}^{G=2,P=1}(x)\right\}}{p^{S=2, Q=1}} \\
& +\frac{p^{S=2, Q=1}(x)}{p^{S=2, Q=1}}\left\{-\frac{\bar{g}(1-\bar{q})\left\{\Delta y-m_{y_2 - y_1}^{S=2,Q=0}(x)\right\}}{p^{S=2, Q=0}(x)}-\frac{(1-\bar{g}) \bar{q} \left\{\Delta y-m_{y_2 - y_1}^{S=\infty,Q=1}(x)\right\}}{p^{S=\infty, Q=1}(x)}\right. \\
& \left.+\frac{(1-\bar{g})(1-\bar{q})\left\{\Delta y-m_{y_2 - y_1}^{S=\infty,Q=0}(x)\right\}}{p^{S=\infty, Q=0}(x)}\right\} \\
& +\frac{\bar{g} \bar{q}}{p^{S=2, Q=1}} \sum_{(g, q) \in \mathcal{S} \times \{0,1\}}(-1)^{\bar{g} + \bar{q}}\left\{m_{y_2 - y_1}^{S=g, Q=q}(x)-\int m_{y_2 - y_1}^{S=g, Q=q}(x) f(x) d x\right\}.
\end{align*}
For the parametric submodel with score function $s_\theta(y_1, y_0, g,q, x)$, we can express the derivative of $\tau(\theta)$ with respect to $\theta$ as follows:
$$
\frac{\partial \tau\left(\theta\right)}{\partial \theta}\bigg\rvert_{\theta = 0}=\mathbb{E}\left[F_\tau\left(O\right) \cdot s_{0}\left(O\right)\right]
$$
therefore, concluding that $\tau$ is pathwise differentiable.

\textsc{Step 3:} To show that $F_\tau(O)$ is the efficient influence function for $\tau$, we verify equation (9) and Theorem 3.1 in \cite{newey1990}. This requires showing that $F_\tau(\cdot)$ belongs to the tangent space $\mathcal{T}$. Notice that terms in $F_{\tau}(\cdot)$ can be re-expressed for all pairs $(g,q)$ as follows: 
\begin{align*}
    s_{2,1}(y_2, y_1, x)&=\frac{1}{p^{S=2, Q=1}} \left\{ y_2 - y_1 -m_{y_2 - y_1}^{S=2, Q=1}(x)\right\},\\
    s_{g,q}(y_2, y_1, x) &= \frac{1}{p^{S=2, Q=1}}(-1)^{\bar{g}+\bar{q}} \frac{p^{S=2, Q=1}(x)\left\{y_2 - y_1 -m_{y_2 - y_1}^{S=g, Q=q}(x)\right\}}{p^{S=g, Q=q}(x)},\\
    p_{2,1}(x) &= \frac{1}{p^{S=2, Q=1}} \sum_{(g, q) \in \mathcal{S} \times \{0,1\}}(-1)^{\bar{g}+\bar{q}}\left\{m_{y_2 - y_1}^{S=g, Q=q}(x)-\int m_{y_2 - y_1}^{S=g, Q=q}(x) f(x) d x\right\},\\
    p_{g,q}(x) &= 0 \quad \text{and}, \quad l(x) =0.
\end{align*}
Hence, one can easily confirm that the conditions outlined in equations (\ref{condition_newey1}) through (\ref{condition_newey3}) are satisfied, thereby establishing the desired result. Finally, notice that for any pair $(g,q)$ we have
\begin{align*}
    p^{S=2, Q=1} = \E\left[\frac{1\{S=g, Q=q\} \cdot p^{S=2, Q=1}(X)}{p^{S=g, Q=q}(X)}\right]
\end{align*}
and through straightforward algebraic manipulation, $F_{\tau}(\cdot)$ can be restated as
\begin{align*}
    F_{\tau}(O) &= \frac{1\{S=2, Q=1\}}{\mathbb{E}\left[1\{S=2, Q=1\}\right]} \cdot \Big(Y_2 - Y_1 -m_{Y_2 - Y_1}^{S=2, Q=0}(X)-m_{Y_2 - Y_1}^{S=\infty, Q=1}(X)+m_{Y_2 - Y_1}^{S=\infty, Q=0}(X) - \tau\Big)\\
    &- \Big(Y_2 - Y_1-m_{Y_2 - Y_1}^{S=2, Q=0}(X)\Big) \cdot \left(\frac{\frac{1\{S=2, Q=0\} \cdot p^{S=2, Q=1}(X)}{p^{S=2, Q=0}(X)}}{\mathbb{E}\left[\frac{1\{S=2, Q=0\} \cdot p^{S=2, Q=1}(X)}{p^{S=2, Q=0}(X)}\right]} \right)\\
    &- \Big(Y_2 - Y_1 -m_{Y_2 - Y_1}^{S=\infty, Q=1}(X)\Big) \cdot \left(\frac{\frac{1\{S=\infty, Q=1\} \cdot p^{S=2, Q=1}(X)}{p^{S=\infty, Q=1}(X)}}{\mathbb{E}\left[\frac{1\{S=\infty, Q=1\} \cdot p^{S=2, Q=1}(X)}{p^{S=\infty, Q=1}(X)}\right]} \right) \\
    &+ \Big(Y_2 - Y_1 -m_{Y_2 - Y_1}^{S=\infty, Q=0}(X)\Big) \cdot \left(\frac{\frac{1\{S=\infty, Q=0\} \cdot p^{S=2, Q=1}(X)}{p^{S=\infty, Q=0}(X)}}{\mathbb{E}\left[\frac{1\{S=\infty, Q=0\} \cdot p^{S=2, Q=1}(X)}{p^{S=\infty, Q=0}(X)}\right]} \right) \\
    &= \eta_{eff}(O)
\end{align*}
Therefore, the proof is concluded.
\end{proof}

\textbf{Proof of Theorem \ref{thm:Identification}}

\begin{proof}
    We start by showing that $ATT(g,t) = ATT_{ra, g_c}(g,t)$. Given Lemmas \ref{lemma:catt} and \ref{lemma:eif},
    \begin{align*}
        ATT(g,t) &= \E[CATT_x(g,t) | S=g, Q=1]\\
        &= \E\left[\E[Y_{t} - Y_{g-1} | X, S =g, Q = 1] | S=g, Q=1\right] - \E\bigg[ \E[Y_t - Y_{g-1} | X, S=g, Q=0] \\
        &~~~+ \E[Y_t - Y_{g-1} | X, S=g_c, Q=1] - \E[Y_t - Y_{g-1} | X, S=g_c, Q=0] \bigg| S=g, Q=1\bigg]\\
        &= \E[Y_{t} - Y_{g-1} | S=g, Q=1] - \E\bigg[ m_{Y_t - Y_{g-1}}^{S=g, Q=0}(X) \\
        &~~~+ m_{Y_t - Y_{g-1}}^{S=g_c, Q=1}(X) - m_{Y_t - Y_{g-1}}^{S=g_c, Q=0}(X) \bigg| S=g, Q=1\bigg]\\
        &= \E\left[Y_{t} - Y_{g-1} - m_{Y_t - Y_{g-1}}^{S=g, Q=0}(X) - m_{Y_t - Y_{g-1}}^{S=g_c, Q=1}(X) + m_{Y_t - Y_{g-1}}^{S=g_c, Q=0}(X) | S=g, Q=1 \right]\\
        &= \E\Big[{w}^{S=g,Q=1}_{\text{trt}}(S, Q)\Big(Y_t - Y_{g-1} - {m}_{Y_t-Y_{g-1}}^{S=g,Q=0}(X) - {m}_{Y_t-Y_{g-1}}^{S=g_{\text{c}},Q=1}(X) + {m}_{Y_t-Y_{g-1}}^{S=g_{\text{c}},Q=0}(X) \Big)\Big]
    \end{align*}
   Hence, we have established that $ATT(g,t) = ATT_{ra, g_c}(g,t)$. Our next objective is to demonstrate the equality $ATT(g,t) = {ATT}_{\text{ipw}, g_{\text{c}}}(g,t)$. Specifically, we aim to prove that
    \begin{align*}
        \E[w_{g',q'}^{S=g, Q=1}(S,Q,X) \cdot(Y_t - Y_{g-1})] = \dfrac{\E[1\{S=g, Q=1\} \cdot \E[Y_t - Y_{g-1} | X, S=g', Q=q'] ]}{\E[1\{S=g, Q=1\}]}.
    \end{align*}
    By LIE and conditional probabilities, we observe that
    \begingroup
\footnotesize
    \begin{align}
    \label{b1}
     &\E\left[\frac{1\{S=g', Q=q'\} \cdot p_{g',q'}^{S=q, Q=1}(X)}{1 - p_{g',q'}^{S=q, Q=1}(X)} \cdot (Y_t - Y_{g-1}) \right] \nonumber\\ 
     &= \E\left[ \dfrac{\E[1\{S=g, Q=1 \}|X] \cdot 1\{S=g', Q=q'\} }{\E[1\{S=g', Q=q'\}|X]} \cdot (Y_t - Y_{g-1}) \right] \nonumber\\
     &= \E\left[\frac{\E[1\{S=g, Q=1 \}|X]}{\E[1\{S=g', Q=q'\}|X]} \E\left[ 1\{S=g', Q=q'\} \cdot (Y_t - Y_{g-1}) | X\right] \right] \nonumber\\
     &= \E\left[\E[1\{S=g, Q=1 \}|X] \cdot  \E\left[(Y_t - Y_{g-1}) | X, S=g', Q=q'\right] \right] \nonumber\\
     &= \E\left[1\{S=g, Q=1 \} \cdot \E\left[(Y_t - Y_{g-1}) | X, S=g', Q=q'\right] \right]
    \end{align}
    \endgroup
    Given that
    \begin{align}
    \label{b2}
        \E\left[\dfrac{1\{S=g', Q=q'\} \cdot p_{g',q'}^{S=q, Q=1}(X)}{1 - p_{g',q'}^{S=q, Q=1}(X)}\right] &= \E\left[ \dfrac{\E[1\{S=g, Q=1 \}|X] \cdot 1\{S=g', Q=q'\} }{\E[1\{S=g', Q=q'\}|X]} \right] \nonumber\\
        &= \E\left[ \dfrac{\E[1\{S=g, Q=1 \}|X] \cdot \E[1\{S=g', Q=q'\}|X] }{\E[1\{S=g', Q=q'\}|X]} \right] \nonumber\\
        &=  \E\left[\E[1\{S=g, Q=1 \}|X] \right] \nonumber\\
        &= \E[1\{S=g, Q=1\}].
    \end{align}
Combining \eqref{b1} and \eqref{b2}
\begin{align*}
    \E[w_{g', q'}^{S=g, Q=1}(S,Q,X) \cdot (Y_t - Y_{g-1})] &= \dfrac{\E\left[1\{S=g, Q=1 \} \cdot \E\left[(Y_t - Y_{g-1}) | X, S=g', Q=q'\right] \right]}{\E[1\{S=g, Q=1\}]}\\
    &= \E\left[ \E[Y_t - Y_{g-1}| X, S=g', Q=q']  | S=g, Q=1\right]
\end{align*}
Thus, analogous reasoning as in RA demonstrates that $ATT(g,t) = {ATT}_{\text{ipw}, g_{\text{c}}}(g,t)$. Finally, we need to show that $ATT(g,t) = ATT_{\text{dr}, g_{\text{c}}}(g,t)$.
    \begingroup
\footnotesize
\begin{align*}
&{ATT}_{\text{dr}, g_{\text{c}}}(g,t) \\
&= \E\left[ \left({w}^{S=g,Q=1}_{\text{trt}}(S,Q) - {w}^{S=g,Q=1}_{g,0}(S,Q,X) \right)\left(Y_t- Y_{g-1} - {m}_{Y_t-Y_{g-1}}^{S=g,Q=0}(X) \right)\right]\nonumber \\
&~~ +\E\left[ \left({w}^{S=g,Q=1}_{\text{trt}}(S,Q) - {w}^{S=g,Q=1}_{g_{\text{c}},1}(S,Q,X) \right)\left(Y_t- Y_{g-1} -  {m}_{Y_t-Y_{g-1}}^{S=g_{\text{c}},Q=1}(X)  \right)\right] \\
&~~- \E\left[ \left({w}^{S=g,Q=1}_{\text{trt}}(S,Q) -{w}^{S=g,Q=1}_{g_{\text{c}},0}(S,Q,X)\right) \left(Y_t- Y_{g-1} - {m}_{Y_t-Y_{g-1}}^{S=g_{\text{c}},Q=0}(X)\right)\right] \nonumber\\
&= \E\left[ \left({w}^{S=g,Q=1}_{\text{trt}}(S,Q) - {w}^{S=g,Q=1}_{g,0}(S,Q,X) \right) \left(Y_t- Y_{g-1}\right) \right] \\
&~~ + \E\left[ \left({w}^{S=g,Q=1}_{\text{trt}}(S,Q) - {w}^{S=g,Q=1}_{g_c,1}(S,Q,X) \right) \left(Y_t- Y_{g-1}\right) \right] \\
&~~ - \E\left[ \left({w}^{S=g,Q=1}_{\text{trt}}(S,Q) - {w}^{S=g,Q=1}_{g_c,0}(S,Q,X) \right) \left(Y_t- Y_{g-1}\right) \right] \\
&~~ - \E\left[ \left({w}^{S=g,Q=1}_{\text{trt}}(S,Q) - {w}^{S=g,Q=1}_{g,0}(S,Q,X) \right) {m}_{Y_t-Y_{g-1}}^{S=g,Q=0}(X) \right] \\
&~~ - \E\left[ \left({w}^{S=g,Q=1}_{\text{trt}}(S,Q) - {w}^{S=g,Q=1}_{g,0}(S,Q,X) \right) {m}_{Y_t-Y_{g-1}}^{S=g_{\text{c}},Q=1}(X) \right]  \\
&~~ + \E\left[ \left({w}^{S=g,Q=1}_{\text{trt}}(S,Q) - {w}^{S=g,Q=1}_{g,0}(S,Q,X) \right) {m}_{Y_t-Y_{g-1}}^{S=g_{\text{c}},Q=0}(X) \right] \\
&= \E\left[\left({w}^{S=g,Q=1}_{\text{trt}}(S,Q) - {w}^{S=g,Q=1}_{g,0}(S,Q,X) \right) \left( Y_t-Y_{g-1}\right)\right]\nonumber\\ 
&~~ - \E\left[\left({w}^{S=g,Q=1}_{g_{\text{c}},1}(S,Q,X) - {w}^{S=g,Q=1}_{g_{\text{c}},0}(S,Q,X)\right) \left( Y_t-Y_{g-1}\right)\right] \\
&~~ - \E\left[ \left({w}^{S=g,Q=1}_{\text{trt}}(S,Q) - {w}^{S=g,Q=1}_{g,0}(S,Q,X) \right) {m}_{Y_t-Y_{g-1}}^{S=g,Q=0}(X) \right] \\
&~~ - \E\left[ \left({w}^{S=g,Q=1}_{\text{trt}}(S,Q) - {w}^{S=g,Q=1}_{g,0}(S,Q,X) \right) {m}_{Y_t-Y_{g-1}}^{S=g_{\text{c}},Q=1}(X) \right]  \\
&~~ + \E\left[ \left({w}^{S=g,Q=1}_{\text{trt}}(S,Q) - {w}^{S=g,Q=1}_{g,0}(S,Q,X) \right) {m}_{Y_t-Y_{g-1}}^{S=g_{\text{c}},Q=0}(X) \right] \\
&= ATT(g,t) \\
&~~ - \frac{1}{\E[1\{S=g, Q=1\}]} \E\left[ \left(1\{S=g, Q=1\} -  \dfrac{1\{S=g, Q=0\} \E[1\{S=g, Q=1\} | X]}{\E[1\{S=g, Q=0\}| X]} \right){m}_{Y_t-Y_{g-1}}^{S=g,Q=0}(X) \right] \\
&~~ - \frac{1}{\E[1\{S=g, Q=1\}]} \E\left[ \left(1\{S=g, Q=1\} -  \dfrac{1\{S=g_c, Q=1\} \E[1\{S=g, Q=1\} | X]}{\E[1\{S=g_c, Q=1\}| X]} \right){m}_{Y_t-Y_{g-1}}^{S=g_c,Q=1}(X) \right] \\
&~~ + \frac{1}{\E[1\{S=g, Q=1\}]} \E\left[ \left(1\{S=g, Q=1\} -  \dfrac{1\{S=g_c, Q=0\} \E[1\{S=g, Q=1\} | X]}{\E[1\{S=g_c, Q=0\}| X]} \right){m}_{Y_t-Y_{g-1}}^{S=g_c,Q=0}(X) \right]\\
 &= ATT(g,t) - \dfrac{1}{\E[1\{S=g, Q=1\}]} \E\left[\left(\E[1\{S=g, Q=1\} | X] - \E[1\{S=g, Q=1\} | X] \right)\cdot {m}_{Y_t-Y_{g-1}}^{S=g,Q=0}(X) \right]\\
&~~ - \dfrac{1}{\E[1\{S=g, Q=1\}]} \E\left[\left(\E[1\{S=g, Q=1\} | X] - \E[1\{S=g, Q=1\} | X] \right)\cdot {m}_{Y_t-Y_{g-1}}^{S=g_c,Q=1}(X) \right]\\
&~~ + \dfrac{1}{\E[1\{S=g, Q=1\}]} \E\left[\left(\E[1\{S=g, Q=1\} | X] - \E[1\{S=g, Q=1\} | X] \right)\cdot {m}_{Y_t-Y_{g-1}}^{S=g_c,Q=0}(X) \right]\\
 &= ATT(g,t),
\end{align*}
\endgroup
The second and third equalities result from straightforward algebraic manipulations. The fourth equality derives from $ATT(g,t) = {ATT}_{\text{ipw}, g_{\text{c}}}(g,t)$ and references equations \eqref{b1} and \eqref{b2}. The fifth equality is a consequence of the LIE. This concludes the proof.
\end{proof}

\textbf{Proof of Corollary \ref{cor:over-id}}

\begin{proof}
    Based on results from Theorem \ref{thm:Identification}, it is established that for any $g \in \mathcal{G}_{trt}$ and $t \in \{2, \dots, T\}$ where $t \geq g$, and for $g_c \in \mathcal{S}$ such that $g_c > t$, the following holds: $ATT(g,t) = ATT_{\text{dr}, g_{\text{c}}}(g,t)$. Multiplying both sides of the previous expression by an arbitrary weight $w_{g_c}^{g,t}$ and summing over all admissible comparison groups gives
    \begin{align*}
        \sum_{g_c \in \mathcal{G}_{c}^{g,t}} w_{g_c}^{g,t} ATT(g,t) &=  \sum_{g_c \in \mathcal{G}_{c}^{g,t}} w_{g_c}^{g,t} ATT_{\text{dr}, g_{\text{c}}}(g,t)\\
        ATT(g,t) \underbrace{\sum_{g_c \in \mathcal{G}_{c}^{g,t}} w_{g_c}^{g,t}}_{=1} &= \sum_{g_c \in \mathcal{G}_{c}^{g,t}} w_{g_c}^{g,t} ATT_{\text{dr}, g_{\text{c}}}(g,t)\\
        ATT(g,t) &= \sum_{g_c \in \mathcal{G}_{c}^{g,t}} w_{g_c}^{g,t} ATT_{\text{dr}, g_{\text{c}}}(g,t).
    \end{align*}
    The left-hand side of the second equality simplifies because $ATT(g,t)$ does not depend on $g_c$; allowing us to take it outside the summation. Since the weights sum to one, the desired result is obtained.
\end{proof}

\textbf{Proof of Theorem \ref{thm:CAN}}

\begin{proof}
For the case of $\widehat{ATT}_{\text{dr},\gc}(g,t)$, Theorem \ref{thm:Identification} establishes that $ATT(g,t)$ is point-identified for all $g \in \mathcal{G}_{trt}$, $t \in \{2, \dots, T\}$, and $\gc \in \mathcal{G}_{c}^{g,t}$ such that $t \geq g$. Additionally, we observe that $\widehat{ATT}_{\text{dr},\gc}(g,t)$ is comprised of a function involving three DR DiDs. Therefore, the asymptotic linear representation of $\sqrt{n}\left(\widehat{ATT}_{\text{dr}, \gc}(g,t) - {ATT}(g,t) \right)$ follows from Theorem A.1(a) in \cite{SantAnna2020}. This is due to the fact that $\psi^{g,t}_{\gc}(W_i; \kappa^{g,t}_{0, \gc})$ is a weighted sum of the influence functions for each DR DiD, with weights reflecting the number of units in each of them. The asymptotic normality follows from the Lindeberg-Lévy central limit theorem. 

According to Corollary \ref{cor:over-id}, the estimator $\widehat{ATT}_{\text{dr},\text{gmm}}(g,t)$ implies that $ATT(g,t)$ is over-identified. This means that different $g_c \in \mathcal{G}_{c}^{g,t}$ can be used, and any weighted sum of these estimands will identify our parameter of interest. For each comparison group $g_c$, the estimator $\widehat{ATT}_{\text{dr}, \gc}(g,t)$ is determined by a function involving three DR DiD terms, and $\widehat{ATT}_{\text{dr}, \text{gmm}}(g,t)$ is then a weighted sum of these terms. Consequently, its asymptotic linear representation, as established in Theorem A.1(a) from \cite{SantAnna2020}, is rescaled by the factor $\textbf{1}'{\Omega}^{-1}_{g,t}/\textbf{1}'{\Omega}^{-1}_{g,t}\textbf{1}$. This adjustment reflects the incorporation of the $k_{g,t} \times 1$ vector comprising all ${\mathbb{RIF}}_{\text{dr}, g_{\text{c}}}(g,t)$ for every $g_c \in \mathcal{G}_{c}^{g,t}$, given that $E[\mathbb{IF}_{dr}(g,t)] = 0$. As before, asymptotic normality follows from an implementation of the Lindeberg-Lévy central limit theorem.

To demonstrate that $\Omega_{g,t, \text{gmm}} = \left({\textbf{1}'{\Omega}^{-1}_{g,t}\textbf{1}}\right)^{-1} \leq \Omega_{g,t,\gc}$ for any $\gc \in \mathcal{G}_{\text{c}}^{\text{g,t}}$, we can reformulate our minimum variance problem into a GMM problem. This involves a vector of moment conditions represented as $\E[\mathbb{RIF}_{dr}(g,t) - \theta^{g,t}] = 0$. Under the conventional regularity conditions that underpin GMM theory \citep{Newey_McFadden_1994_Handbook}, the asymptotic variance associated with any set of weights $w$ is expressed as:
\begin{align*}
    \Omega_{g,t,w} = (\textbf{1}' W \textbf{1})^{-1} \textbf{1}'W\Omega_{g,t}W\textbf{1}(\textbf{1}'W\textbf{1})^{-1}
\end{align*}
where $W$ denotes any positive definite weight matrix. According to Efficient GMM theory, choosing $W=\Omega_{g,t}^{-1}$ minimizes the asymptotic variance, leading to $\Omega_{g,t, \text{gmm}} = \left({\textbf{1}'{\Omega}^{-1}_{g,t}\textbf{1}}\right)^{-1}$. More generally, for any $W > 0$, we have:
\begin{align*}
    (\textbf{1}' W \textbf{1})^{-1} \textbf{1}'W\Omega_{g,t}W\textbf{1}(\textbf{1}'W\textbf{1})^{-1} - \left({\textbf{1}'{\Omega}^{-1}_{g,t}\textbf{1}}\right)^{-1} > 0 
\end{align*}
Consequently, it follows that $\Omega_{g,t, \text{gmm}} \leq \Omega_{g,t, w}$.
\end{proof}

\textbf{Proof of Corollary \ref{cor:asy_es}}

\begin{proof}
    The population event-study parameter is defined by $ES(e) = \sum_{g \in \mathcal{G}_{trt}} \P(G=g | G+e \in [1, T]) \cdot ATT(g, g+e)$ and its sample analogue is provided by Equation \eqref{eqn:es_estimator}. By adding and subtracting the term $\P(G=g | G+e \in [1,T])\cdot \widehat{ATT}_{\text{dr}, \text{gmm}}(g,t)$, followed by multiplying the bias term by $\sqrt{n}$, we get the following expression
    \begin{align*}
        &\sqrt{n}\left(\widehat{ES}_{\text{dr},\text{gmm}}(e) - ES(e) \right) \\
        &= \sqrt{n} \sum_{g \in \mathcal{G}_{trt}} \P(G=g | G+e \in [1,T]) \cdot \left(\widehat{ATT}_{\text{dr}, \text{gmm}}(g,t) - ATT(g,t)\right) \\
        &~~~~~~ + \sqrt{n} \sum_{g \in \mathcal{G}_{trt}}  \left(\P_n(G=g | G+e \in [1,T]) - \P(G=g | G+e \in [1,T])\right) \cdot ATT(g,t)  + o_p(1).
    \end{align*}
Given the asymptotic linear representation of both $\sqrt{n}(\widehat{ATT}_{\text{dr}, \text{gmm}}(g,t) - ATT(g,t))$ and $\sqrt{n}(\P_n(G=g | G+e \in [1,T]) - \P(G=g | G+e \in [1,T]))$ as stated in Theorem \ref{thm:CAN} and Equation \eqref{alr_weights_es}, respectively, we can substitute these expressions into the representation above and arrange the summation over $i = 1, \ldots, n$ to obtain
\begingroup
\footnotesize
\begin{align*}
     &\sqrt{n}\left(\widehat{ES}_{\text{dr},\text{gmm}}(e) - ES(e) \right) \\
     &= \frac{1}{\sqrt{n}} \sum_{i=1}^{n} \underbrace{\Big[ \sum_{g \in \mathcal{G}_{\text{trt}}} \big( \P(G=g| G+e \in [1,T]) \cdot \dfrac{\textbf{1}'{\Omega}^{-1}_{g,t}}{\textbf{1}'{\Omega}^{-1}_{g,t}\textbf{1}} \psi^{g,t}(W_i; \kappa^{g,t}_{0}) + \xi^{g,e}(W_i) \cdot ATT(g,t) \big)  \Big]}_{l^{es,e}_{\text{gmm}}(W_i)} + o_p(1).
\end{align*}
\endgroup
As before, asymptotic normality is established through the Lindeberg-Lévy central limit theorem. This concludes the proof.
\end{proof}

\section{Additional Details about Monte Carlo Simulations}\label{appendix:simulations}
This section provides additional details about the Monte Carlo designs discussed in the main text. 

\subsection{More details on two-period DDD setup with covariates}
\label{appendix:sims_2_periods}
As discussed in Section \ref{sec:sims_2_periods}, we considered a two-period DDD setup with covariates where, for a generic four-dimensional vector $O$, the conditional probability of each unit belonging to a subgroup $(g,q) \in \{2,\infty\} \times \{0,1\}$ is
\begin{equation}
\label{appendix:ps_score_sims}
    \P[S = g, Q=q | O] = p^{S=g, Q=q}(O) =  \dfrac{\exp(f^{ps}_{S=g,Q=q}(O))}{\sum_{(g,q) \in \mathcal{S}_{\text{des-1}} \times \{0,1\}}\exp(f^{ps}_{S=g,Q=q}(O))},
\end{equation}
such that, for each $(g,q) \in \mathcal{S}_{\text{des-1}} \times \{0,1\}$, we define the linear process $f^{ps}_{S=g_c,Q=q}(O) = c O'\gamma_{\gc,q}$, where $c$ is a scalar that guarantees reasonable overlap, and $\gamma_{\gc, q}$ is a $4 \times 1$ vector of coefficients for each variable in $O$ given by
\begin{eqnarray*}
\gamma_{\infty, 0} = \begin{pmatrix}
-1 \\
0.5 \\
-0.25 \\
-0.1
\end{pmatrix}; \gamma_{\infty,1} = \begin{pmatrix}
-0.5 \\
2 \\
0.5 \\
-0.2
\end{pmatrix}; \gamma_{2,0} = \begin{pmatrix}
3 \\
-1.5 \\
0.75 \\
-0.3
\end{pmatrix}; \text{and }
c = \begin{cases}
    0.2, & \text{if $S=\infty, Q=0$}\\
    0.2, & \text{if $S=\infty, Q=1$}\\
    0.05, & \text{if $S=2, Q=0$}
\end{cases}
\end{eqnarray*}
    
In this design, the probabilities for each subgroup are denoted as $\big\{p^{S=2, Q=0}(O)$, $p^{S=2, Q=1}(O)$, $p^{S=\infty, Q=0}(O)$, $p^{S=\infty, Q=1}(O)\big\}$, and their sum is equal to one. Define the random variable \( U \sim \text{Uniform}[0,1] \). The assignment process to each group is defined as in \eqref{eq_appen:assignment_process}.

The potential outcomes are defined as in \eqref{eqn:PO_2_periods}. We define $f^{reg}(O,S)$ as
\begin{align*}
    f^{reg}(O,S) = 1\{S_i = 2\} \cdot f^{reg}_{S=2}(O) + 1\{S_i = \infty\} \cdot f^{reg}_{S=\infty}(O),
\end{align*}
with $f^{reg}_{S=g}(O) =  \alpha + O' \beta_g$, with $\alpha=2010$, $\beta_{2} = (27.4, 13.7, 13.7, 13.7)'$  and $\beta_{\infty} = 0.5 \beta_{2}$. We define the unobserved heterogeneity term $\nu_i(O_i, S_i, Q_i)$ as 
\begin{align*}
    \nu_i(O_i, S_i, Q_i) \overset{d}{\sim} N(1\{S_i=2\} Q_i f^{reg}_{S=2}(O_i) +  1\{S_i=\infty\} Q_i f^{reg}_{S=\infty}(O_i), 1).
\end{align*}

Finally, as discussed in our main text, we allow misspecification of propensity score and/or outcome regression models. In all four DGPs, the observed data is $W_i = \{Y_{i,1}, Y_{i,2}, S_{i}, Q_{i}, X_{i}\}_{i=1}^{n}$, where the covariates $X_i$ are generated as follows. Let $Z_i = (Z_{i,1}, Z_{i,2}, Z_{i,3}, Z_{i,4})' \overset{d}{\sim} N(0,I_4)$. The observed vector of covariates $X_i = (X_{i,1}, X_{i,2}, X_{i,3}, X_{i,4})'$ where for every $k =1, \cdots, 4$, $X_{k} = (\Tilde{X}_{k} - \E[\Tilde{X}_{k}]) \Big / \sqrt{\Var(\Tilde{X}_{k})}$ such that
\begin{align*}
    \tilde{X}_1 &= \exp \left(0.5 Z_1\right),\\
    \tilde{X}_2 &= 10+Z_2 \big / \left(1+\exp \left(Z_1\right)\right), \\
    \tilde{X}_3 &= \left(0.6+Z_1 Z_3 / 25\right)^3, \\
    \tilde{X}_4 &= \left(20+Z_1+Z_4\right)^2.
\end{align*}  
These transformations build on \cite{kang_schafer_2007} and \citet{SantAnna2020}. We consider four different DGPs depending on whether the propensity score and/or the outcome regressions are misspecified. The specifics are outlined below:
\begin{itemize} 
\item \textit{DGP 1}: All nuisance functions depend on $X$, namely $f^{reg}_{S=g}(X)$ and $f^{ps}_{S=g,Q=q}(X)$. All working models are correctly specified in this DGP, as $X$ is observed in the data. 
\item \textit{DGP 2}: Regression models depend on $X$, $f^{reg}_{S=g}(X)$, and  the propensity score depends on $Z$ $f^{ps}_{S=g,Q=q}(Z)$. The working model for propensity score is misspecified, whereas the working models for the outcomes are correctly specified. 
\item \textit{DGP 3}: Regression models depend on $Z$, $f^{reg}_{S=g}(Z)$, and propensity score depends on $X$, $f^{ps}_{S=g,Q=q}(X)$.  The working model for propensity score is correctly specifies, whereas the working models for the outcomes are misspecified.  
\item \textit{DGP 4}: All nuisances functions depend on $Z$, namely $f^{reg}_{S=g}(Z)$ and $f^{ps}_{S=g,Q=q}(Z)$. All working models are misspecified. \end{itemize}

We graphically illustrate the results in Figure \ref{fig:Sims_2_periods_MC} when $n=5,000$. The following table presents the traditional summary statistics for the Monte Carlo involving average bias, root mean square error (RMSE), empirical 95\% coverage probability, and the average length of a 95\% confidence interval under 1,000 Monte Carlo repetitions. We report these results for $n=1,000$, $n=5,000$, and $n=10,000$. These results are self-explanatory.

\begin{sidewaystable}[p]
\centering
\caption{Monte Carlo results for $ATT(2,2)$ in DGP 1 – DGP 4, with two periods and covariates}
\label{mc_sim_att22}
\begin{adjustbox}{width=\textwidth,height=\textheight,keepaspectratio}
\begin{threeparttable}
\begin{tabular}{@{}lcccccccccccccccc@{}}
\toprule
& \multicolumn{4}{c}{DGP 1} 
& \multicolumn{4}{c}{DGP 2} 
& \multicolumn{4}{c}{DGP 3} 
& \multicolumn{4}{c}{DGP 4} \\
\cmidrule(lr){2-5}\cmidrule(lr){6-9}\cmidrule(lr){10-13}\cmidrule(lr){14-17}
Estimator 
& \small{Bias} & \small{RMSE} & \small{Cov.\ 95} & \small{ALCI}
& \small{Bias} & \small{RMSE} & \small{Cov.\ 95} & \small{ALCI}
& \small{Bias} & \small{RMSE} & \small{Cov.\ 95} & \small{ALCI}
& \small{Bias} & \small{RMSE} & \small{Cov.\ 95} & \small{ALCI} \\
\midrule

& \multicolumn{16}{c}{$n=1000$} \\ \cmidrule(lr){2-17}
DRDDD    &  -0.003 & 0.188 & 0.953 &  0.732 &  -0.008 & 0.185 & 0.944 & 0.729 &  0.032 & 1.606 & 0.952 & 6.310 &  -2.072 & 2.629 & 0.747 & 6.304 \\
3WFE     &  -9.298 &10.301 & 0.674 & 22.431 &  -8.038 & 9.090 & 0.780 &22.426 & -4.985 & 6.160 & 0.977 &24.794 &  -7.755 & 8.574 & 0.894 &24.714 \\
DRDID-DIF&  -2.681 & 3.361 & 0.762 &  8.100 &  -2.302 & 3.016 & 0.798 & 7.862 & -1.280 & 2.515 & 0.911 & 8.514 &  -3.402 & 4.054 & 0.628 & 8.332 \\
M-3WFE   &   1.192 & 1.809 & 0.855 &  5.275 &   0.865 & 1.621 & 0.903 & 5.331 &  1.188 & 2.219 & 0.898 & 7.177 &  -1.076 & 2.185 & 0.900 & 7.161 \\

\midrule
& \multicolumn{16}{c}{$n=5000$} \\ \cmidrule(lr){2-17}
DRDDD    &  -0.002 & 0.083 & 0.944 & 0.324 &  0.000 & 0.084 & 0.951 & 0.323 & -0.014 & 0.746 & 0.939 & 2.794 &  -2.019 & 2.141 & 0.190 & 2.792 \\
3WFE     &  -9.059 & 9.291 & 0.022 &10.055 &  -7.929 & 8.156 & 0.069 &10.065 & -5.049 & 5.324 & 0.601 &11.125 &  -7.563 & 7.733 & 0.101 &11.097 \\
DRDID-DIF&  -2.616 & 2.782 & 0.203 & 3.635 &  -2.270 & 2.442 & 0.283 & 3.520 & -1.372 & 1.706 & 0.698 & 3.832 &  -3.326 & 3.463 & 0.064 & 3.755 \\
M-3WFE   &   1.154 & 1.293 & 0.532 & 2.358 &   0.847 & 1.050 & 0.708 & 2.387 &  1.137 & 1.419 & 0.710 & 3.206 &  -0.985 & 1.284 & 0.777 & 3.197 \\

\midrule
& \multicolumn{16}{c}{$n=10000$} \\ \cmidrule(lr){2-17}
DRDDD    &   0.001 & 0.058 & 0.954 & 0.229 &   0.002 & 0.060 & 0.942 & 0.228 & -0.008 & 0.495 & 0.944 & 1.973 &  -2.013 & 2.078 & 0.024 & 1.973 \\
3WFE     &  -9.248 & 9.349 & 0.000 & 7.106 &  -7.970 & 8.091 & 0.001 & 7.119 & -5.108 & 5.231 & 0.146 & 7.870 &  -7.597 & 7.677 & 0.000 & 7.852 \\
DRDID-DIF&  -2.686 & 2.759 & 0.010 & 2.568 &  -2.295 & 2.381 & 0.036 & 2.489 & -1.405 & 1.556 & 0.480 & 2.716 &  -3.339 & 3.404 & 0.001 & 2.660 \\
M-3WFE   &   1.153 & 1.226 & 0.223 & 1.667 &   0.829 & 0.935 & 0.520 & 1.688 &  1.138 & 1.272 & 0.499 & 2.267 &  -0.957 & 1.127 & 0.619 & 2.260 \\

\bottomrule
\end{tabular}
\begin{tablenotes}
\item \textit{Notes:} This table summarizes the Monte Carlo experiments for four distinct DGPs as discussed in Section \ref{sec:sims_2_periods}. The DRDDD row includes our proposed doubly robust estimators with $g_c = \infty$. The 3WFE and M-3WFE rows display OLS estimates for 3WFE models in equations \eqref{eqn:3WFE_with_X} and \eqref{eqn:3WFE_with_X_Mundlak}. The DRDID-DIF row shows the difference between two Doubly Robust DiDs. Columns represent average bias, RMSE, coverage probability at 95\% (Cov. 95), and average confidence interval length (ALCI) for each estimator. The 95\% confidence intervals use point-wise asymptotic critical values. Results span sample sizes $n=\{1{,}000,5{,}000,10{,}000\}$ over 1,000 simulations, with the true $ATT(2,2)$ being zero.
\end{tablenotes}
\end{threeparttable}
\end{adjustbox}
\end{sidewaystable}

\subsection{More details on staggered DDD setups}
\label{appendix:sims_staggered}
We now provide more details about the simulation designs for staggered DDD as discussed in Section \ref{sims_staggered}. The potential outcomes are as defined in \eqref{eqn:PO_staggered} with the unobserved heterogeneity term $\nu_i(S_i, Q_i)$ being defined as
\begin{align}
    \nu_i(S_i, Q_i) \overset{d}{\sim} N(1\{S_i=2\} (2 +Q_i)\alpha + 1\{S_i=2\} (2 +Q_i)\alpha + 1\{S_i=\infty\} Q_i\alpha,1),
\end{align}
where $\alpha =278.5$. All the other relevant information is described in Section \ref{sims_staggered}.

As we summarize the simulation results in Section \ref{sims_staggered} via graphs, below we present a more traditional summary of our simulation results using a table. We stress that, in period 3, all the considered DDD estimators coincide as there is only one possible comparison group at that time period, the never-enabling units. This explains the results for $ATT(2,3)$ and $ATT(3,3)$ in Table \ref{app:mc_sim_attgt}. The $ATT(2,2)$ results are self-explanatory and highlight that (a) proceeding as if DDD is just the difference between two DiD procedures can lead to misleading conclusions, (b) using our proposed DDD estimators bypasses these limitations, and (c) the gains in precision of using our optimally weighted estimator can be large.

\begin{table}[h]
\centering
\caption{Monte Carlo results for Staggered DDD without covariates}
\label{app:mc_sim_attgt}
\begin{adjustbox}{width=\textwidth,height=\textheight,keepaspectratio}
\begin{threeparttable}
\begin{tabular}{@{}lcccccccccccc@{}}
\toprule
& \multicolumn{4}{c}{$ATT(2,2)$} 
& \multicolumn{4}{c}{$ATT(2,3)$} 
& \multicolumn{4}{c}{$ATT(3,3)$} \\
\cmidrule(lr){2-5}\cmidrule(lr){6-9}\cmidrule(lr){10-13}
Estimator 
& \small{Bias} & \small{RMSE} & \small{Cov.\ 95} & \small{ALCI}
& \small{Bias} & \small{RMSE} & \small{Cov.\ 95} & \small{ALCI}
& \small{Bias} & \small{RMSE} & \small{Cov.\ 95} & \small{ALCI} \\
\midrule

& \multicolumn{12}{c}{$n=1000$} \\ \cmidrule(lr){2-13}
$DDD_{gmm}$ & -0.003 & 0.192 & 0.956 & 0.748 & -0.004 & 0.298 & 0.947 & 1.140 & -0.007 & 0.286 & 0.941 & 1.083\\
$DDD_{nev}$ & -0.001 & 0.287 & 0.947 & 1.131 & -0.004 & 0.298 & 0.947 & 1.140 & -0.007 & 0.286 & 0.941 & 1.083\\
$DDD_{cs-nyt}$ & -15.815 & 16.034 & 0.000 & 10.833 & -0.004 & 0.298 & 0.947 & 1.140 & -0.007 & 0.286 & 0.941 & 1.083\\

\midrule
& \multicolumn{12}{c}{$n=5000$} \\ \cmidrule(lr){2-13}
$DDD_{gmm}$ & -0.009 & 0.086 & 0.950 & 0.335 & -0.007 & 0.134 & 0.938 & 0.511 & 0.001 & 0.127 & 0.935 & 0.487\\
$DDD_{nev}$ & -0.012 & 0.135 & 0.932 & 0.507 & -0.007 & 0.134 & 0.938 & 0.511 & 0.001 & 0.127 & 0.935 & 0.487\\
$DDD_{cs-nyt}$ & -15.875 & 15.919 & 0.000 & 4.845 & -0.007 & 0.134 & 0.938 & 0.511 & 0.001 & 0.127 & 0.935 & 0.487\\

\midrule
& \multicolumn{12}{c}{$n=10000$} \\ \cmidrule(lr){2-13}
$DDD_{gmm}$  & 0.001 & 0.062 & 0.941 & 0.237 & 0.001 & 0.097 & 0.940 & 0.361 & -0.001 & 0.090 & 0.941 & 0.344\\
$DDD_{nev}$  & 0.000 & 0.094 & 0.941 & 0.359 & 0.001 & 0.097 & 0.940 & 0.361 & -0.001 & 0.090 & 0.941 & 0.344\\
$DDD_{cs-nyt}$ & -15.897 & 15.918 & 0.000 & 3.428 & 0.001 & 0.097 & 0.940 & 0.361 & -0.001 & 0.090 & 0.941 & 0.344\\

\bottomrule
\end{tabular}
\begin{tablenotes}
\item \textit{Notes:} This table summarizes the Monte Carlo experiments for the DGP as discussed in Section \ref{sims_staggered}. $DDD_{\text{{nev}}}$ denotes our DDD estimator with $g_c = \infty$ from Equation \eqref{eqn:DRDDD_stagg}. $DDD_{\text{{gmm}}}$ is our proposed DDD estimator from Equation \eqref{eqn:DRDDD_optimal}. $DDD_{cs-nyt}$ is the estimator pooling all not-yet-treated units as defined in \eqref{eqn:att_gt_pooled}. Columns represent average bias, RMSE, coverage probability at 95\% (Cov. 95), and average confidence interval length (ALCI) for each estimator. The 95\% confidence intervals use point-wise asymptotic critical values. Results span sample sizes $n=\{1{,}000,5{,}000,10{,}000\}$ over 1,000 simulations, with the true $ATT(2,2)=10$,  $ATT(2,3)=20$, and  $ATT(3,3)=25$.
\end{tablenotes}
\end{threeparttable}
\end{adjustbox}
\end{table}

\subsection{Staggered DDD setups with covariates}
\label{appendix:sims_staggered_with_covariates}

In this section, we expand on the DGP mentioned in Section \ref{sims_staggered} by incorporating covariates into the analysis. Similar to before, there are three time periods $t = 1, 2, 3$, three enabling groups $S \in \{2, 3, \infty\}$, and two eligibility groups $Q \in \{0, 1\}$.

Without loss of generality, we define $f^{ps}_{S=g,Q=q}(W)= c_{q} W^{\top}\gamma_{g}, \forall(g,q) \in \mathcal{S}_{des-2} \times \{0,1\}$, where $c_{q}$ is a scalar controlling the overlap condition on the propensity scores, and $\gamma_{g}$ is a $4\times1$ vector of coefficients for each variable in $W$. For an arbitrary four-dimensional vector $O$, the conditional probability of each unit belonging to a subgroup $(g,q) \in \mathcal{S}_{des-2} \times \{0,1\}$ is

\begin{equation}
    \P[G = g, Q = q \mid O] = p^{S=g, Q=q}(O) = \dfrac{exp(f^{ps}_{S=g,Q=q}(O))}{\sum_{(g,q) \in \mathcal{S}_{des-2} \times \{0,1\}}exp(f^{ps}_{S=g,Q=q}(O))}
\end{equation}
Then, we set $c_{q} = 0.4$ if $Q = 0$, $c_{q} = -0.4$ if $Q = 1$ and,
\begin{eqnarray*}
\gamma_{2} = \begin{pmatrix}
-1 \\
0.5 \\
-0.25 \\
-0.1
\end{pmatrix}; \gamma_{3} = \begin{pmatrix}
-0.5 \\
1 \\
-0.1 \\
-0.25
\end{pmatrix}; \gamma_{\infty} = \begin{pmatrix}
-0.25 \\
0.1 \\
-1 \\
-0.1
\end{pmatrix}
\end{eqnarray*}
In this design, the probabilities for each subgroup are denoted as $\big\{p^{S=2,Q=0}(O)$, $p^{S=2,Q=1}(O)$, $p^{S=3,Q=0}(V)$, $p^{S=3,Q=1}(O)$, $p^{S=\infty,Q=0}(O)$, $p^{S=\infty,Q=1}(O)\big\}$, which sum to one. Define the random variable \( U \sim \text{Uniform}[0,1] \). The assignment process to each group is determined as follows:
\begin{align}
  (S,Q) := \begin{cases}
(\infty,0), & \text { if } U \leqslant p^{S=\infty, Q=0}(O), \\ 
(\infty,1), & \text { if } U \leqslant \sum_{j=0}^{1} p^{S=\infty, Q=j}(O), \\ 
(2,0), & \text { if }  U \leqslant \sum_{j=0}^{1} p^{S=\infty, Q=j}(O) + p^{S=2,Q=0}(V), \\
(2,1), & \text { if } U \leqslant \sum_{j=0}^{1} p^{S=\infty, Q=j}(O) + \sum_{j=0}^{1} p^{S=2, Q=j}(O), \\
(3,0), & \text { if } U \leqslant \sum_{j=0}^{1} p^{S=\infty, Q=j}(O) + \sum_{j=0}^{1} p^{S=2, Q=j}(O) + p^{S=3, Q=0}(O), \\
(3,1), & \text { if }  \sum_{j=0}^{1} p^{S=\infty, Q=j}(O) + \sum_{j=0}^{1} p^{S=2, Q=j}(O) + p^{S=3, Q=0}(O) < U .\end{cases}  
\end{align}

Next, we define the outcome regression component of our DGP. In this model, covariates $O$ enter the working model linearly. Let $f^{reg}(O_i) = \alpha + O'\beta$, where $\alpha$ is a scalar and $\beta$ is a $4\times 1$ vector of coefficients for each variable in $O$. We set $\alpha = 210$ and $\beta = (27.4, 13.7, 13.7, 13.7)^{\prime}$. The untreated potential outcomes (which are observed for all units) at period $t=1$ are defined as follows:
\begin{equation*}
    Y_{i,1}(\infty) = (1+Q_i) \cdot f^{reg}(O_i) + \nu_{i}(O_i, S_i, Q_i) + \varepsilon_{i,1}(\infty)
\end{equation*}
where $\nu_{i}(O_i, S_i, Q_i)$ denotes a time-invariant unobserved heterogeneity with $\nu_{i}(O_i, S_i, Q_i) \sim N(M_i \cdot f^{reg}(O_i) +  Q_i \cdot f^{reg}(O_i), 1)$ where $M_i= S_i$ if $S_i \in \mathcal{G}_{trt}$, zero otherwise, and $\varepsilon_{i,1}(\infty) \sim N(0,1)$. Subsequently, we generate the potential outcomes at period $t=2$ for each $g \in \mathcal{G}$:
\begin{align*}
    Y_{i,2}(\infty) &= (2 +Q_i) f^{reg}(O_i) + 2 \nu_{i}(O_i, S_i, Q_i) + \varepsilon_{i,2}(\infty)\\
    Y_{i,2}(2) &= (2 +Q_i) f^{reg}(O_i) + 2 \nu_{i}(O_i, S_i, Q_i) + Q_i \cdot ATT(2,2)  + \varepsilon_{i,2}(2)\\
    Y_{i,2}(3) &= (2 +Q_i) f^{reg}(O_i) + 2 \nu_{i}(O_i, S_i, Q_i) + \varepsilon_{i,2}(3)
\end{align*}
where $\varepsilon_{i,2}(\cdot) \sim N(0,1)$. The realized outcomes at $t=2$ are given by
\begin{equation*}
    Y_{i,2} = \sum_{g\in \mathcal{G}} 1\{G_i=g\}  Y_{i,2}(g)
\end{equation*}
Then, we generate the potential outcomes in the period $t=3$ for each $g \in \mathcal{G}$:
\begin{align*}
    Y_{i,3}(\infty) &= (3 +Q_i)  f^{reg}(O_i) + 3  \nu_{i}(O_i, S_i, Q_i) + \varepsilon_{i,3}(\infty)\\
    Y_{i,3}(2) &= (3 +Q_i)  f^{reg}(O_i) + 3  \nu_{i}(O_i, S_i, Q_i) + Q_i \cdot ATT(2,3)  + \varepsilon_{i,3}(2)\\
    Y_{i,3}(3) &= (3 +Q_i)  f^{reg}(O_i) + 3  \nu_{i}(O_i, S_i, Q_i) + Q_i \cdot ATT(3,3)  + \varepsilon_{i,3}(3)
\end{align*}
where $\varepsilon_{i,3}(\cdot)\sim N(0,1)$. The realized outcomes at $t=3$ are given by
\begin{equation*}
    Y_{i,3} = \sum_{g\in \mathcal{G}} 1\{G_i=g\}  Y_{i,3}(g)
\end{equation*}

Finally, we set $ATT(2,2) = 10$, $ATT(2,3) = 20$ and $ATT(3,3) = 25$. Given that we account for potential misspecification of nuisance functions, we can proceed as outlined in Section \ref{appendix:sims_2_periods}. The observed data is $W_i = \{Y_{i,1}, Y_{i,2}, Y_{i,3}, S_i, Q_i, X_i\}_{i=1}^{n}$, where covariates $X_i$ and $Z_i$ are generated in the same manner as described in Section \ref{appendix:sims_2_periods}. Depending on whether the propensity score and/or outcome regression are misspecified, these specifications will result in four distinct DGPs. Table \ref{mc_sim_staggered_covariates} presents summary statistics from 1,000 Monte Carlo experiments, which include average bias, root mean squared error (RMSE), coverage probability at 95\%, and average confidence interval length. We provide these statistics for the $ATT(2,2)$ under the design mentioned in the current section with sample sizes $n=1,000$, $n=5,000$, and $n=10,000$. Results for $ATT(2,3)$ and $ATT(3,3)$ are omitted to save space, as they produce similar patterns. These results clearly illustrate the effectiveness of our proposed estimator and align with findings from the previous sections.

\begin{sidewaystable}[p]
\centering
\caption{Monte Carlo results for $ATT(2,2)$ in DGP 1 – DGP 4, with multiple periods and covariates}
\label{mc_sim_staggered_covariates}
\begin{adjustbox}{width=\textwidth,height=\textheight,keepaspectratio}
\begin{threeparttable}
\begin{tabular}{@{}lcccccccccccccccc@{}}
\toprule
& \multicolumn{4}{c}{DGP 1} 
& \multicolumn{4}{c}{DGP 2} 
& \multicolumn{4}{c}{DGP 3} 
& \multicolumn{4}{c}{DGP 4} \\
\cmidrule(lr){2-5}\cmidrule(lr){6-9}\cmidrule(lr){10-13}\cmidrule(lr){14-17}
Estimator 
& \small{Bias} & \small{RMSE} & \small{Cov.\ 95} & \small{ALCI}
& \small{Bias} & \small{RMSE} & \small{Cov.\ 95} & \small{ALCI}
& \small{Bias} & \small{RMSE} & \small{Cov.\ 95} & \small{ALCI}
& \small{Bias} & \small{RMSE} & \small{Cov.\ 95} & \small{ALCI} \\
\midrule

& \multicolumn{16}{c}{$n=1000$} \\ \cmidrule(lr){2-17}
$DRDDD_{gmm}$    &  0.009 & 0.236 & 0.943 & 0.901 & -0.002 & 0.230 & 0.945 & 0.893 & -0.175 & 3.300 & 0.943 & 12.396 & -1.499 & 3.765 & 0.917 & 12.722\\
$DRDDD_{nev}$     &  0.012 & 0.298 & 0.939 & 1.137 & -0.002 & 0.288 & 0.937 & 1.121 & -0.367 & 3.984 & 0.937 & 14.726 & 0.947 & 3.956 & 0.937 & 14.403\\
$DRDDD_{cs-nyt}$ &  -8.877 & 9.294 & 0.079 & 10.298 & -8.138 & 8.546 & 0.123 & 10.280 & -8.962 & 9.891 & 0.404 & 16.173 & -11.244 & 12.075 & 0.232 & 16.642\\
$DRCS$-$DIF$  &   -3.765 & 4.668 & 0.704 & 10.348 & -4.377 & 5.109 & 0.608 & 10.116 & -3.919 & 5.355 & 0.789 & 14.263 & -6.866 & 7.767 & 0.520 & 14.161\\

\midrule
& \multicolumn{16}{c}{$n=5000$} \\ \cmidrule(lr){2-17}
$DRDDD_{gmm}$    &  0.006 & 0.106 & 0.947 & 0.407 & -0.002 & 0.105 & 0.935 & 0.399 & 0.083 & 1.450 & 0.953 & 5.629 & -1.170 & 1.956 & 0.856 & 5.954\\
$DRDDD_{nev}$     &  0.005 & 0.131 & 0.949 & 0.507 & -0.005 & 0.131 & 0.944 & 0.497 & 0.012 & 1.746 & 0.948 & 6.639 & 1.128 & 2.032 & 0.900 & 6.568\\
$DRDDD_{cs-nyt}$&  -8.949 & 9.028 & 0.000 & 4.557 & -8.063 & 8.140 & 0.000 & 4.547 & -8.636 & 8.821 & 0.005 & 7.192 & -11.009 & 11.173 & 0.002 & 7.576\\
$DRCS$-$DIF$    &   -3.793 & 3.983 & 0.108 & 4.566 & -4.339 & 4.483 & 0.024 & 4.479 & -3.844 & 4.165 & 0.333 & 6.356 & -6.720 & 6.924 & 0.013 & 6.280\\

\midrule
& \multicolumn{16}{c}{$n=10000$} \\ \cmidrule(lr){2-17}
$DRDDD_{gmm}$    &   -0.001 & 0.073 & 0.955 & 0.289 & -0.004 & 0.073 & 0.948 & 0.283 & 0.047 & 1.021 & 0.951 & 4.004 & -1.088 & 1.544 & 0.820 & 4.249\\
$DRDDD_{nev}$     &  -0.004 & 0.091 & 0.956 & 0.359 & -0.004 & 0.092 & 0.937 & 0.351 & 0.036 & 1.201 & 0.950 & 4.711 & 1.159 & 1.658 & 0.846 & 4.666\\
$DRDDD_{cs-nyt}$ &  -8.928 & 8.965 & 0.000 & 3.226 & -8.111 & 8.151 & 0.000 & 3.211 & -8.746 & 8.842 & 0.000 & 5.090 & -11.108 & 11.188 & 0.000 & 5.373\\
$DRCS$-$DIF$   &   -3.770 & 3.861 & 0.004 & 3.221 & -4.347 & 4.415 & 0.000 & 3.170 & -3.896 & 4.051 & 0.063 & 4.511 & -6.852 & 6.941 & 0.000 & 4.419\\

\bottomrule
\end{tabular}
\begin{tablenotes}
\item \textit{Notes:} This table presents the Monte Carlo experiments for the DGPs as detailed in the text. $DDD_{\text{{nev}}}$ denotes our DDD estimator with $g_c = \infty$ from Equation \eqref{eqn:DRDDD_stagg}. $DDD_{\text{{gmm}}}$ is our proposed DDD estimator from Equation \eqref{eqn:DRDDD_optimal}. $DDD_{cs-nyt}$ is the estimator pooling all not-yet-treated units as defined in \eqref{eqn:att_gt_pooled}. $DRCS$-$DIF$ uses the differences between two DiDs following \cite{Callaway2021} as described in Remark \ref{rem:no_difference_of_2_DiDs}. Each column shows average bias, RMSE, coverage probability at 95\% (Cov. 95), and average confidence interval length (ALCI) for each estimator, using point-wise asymptotic critical values for the confidence intervals. The results cover sample sizes $n=\{1{,}000,5{,}000,10{,}000\}$ over 1,000 simulations, with the true $ATT(2,2)=10$.
\end{tablenotes}
\end{threeparttable}
\end{adjustbox}
\end{sidewaystable}

\newpage

\small{
\setlength{\bibsep}{1pt plus 0.3ex}
\putbib
}
\end{bibunit}
\end{document}